\documentclass[onecolumn,draftcls,11pt]{IEEEtran}
\def\InputAlphabet{\mathcal{X}}
\def\StateAlphabet{\mathcal{S}}
\def\OutputAlphabet{\mathcal{Y}}
\def\define{:{=}~}
\def\hata{{\hat{a}}}
\def\hatm{{\hat{m}}}

\def\hatv{{\hat{v}}}
\def\hatnu{{\hat{\nu}}}
\def\boldx{{\boldsymbol{x}}}
\def\boldm{{\boldsymbol{m}}}

\def\boldX{{\boldsymbol{X}}}
\def\bolds{{\boldsymbol{s}}}
\def\boldu{{\boldsymbol{u}}}
\def\boldv{{\boldsymbol{v}}}
\def\boldS{{\boldsymbol{S}}}
\def\boldU{{\boldsymbol{U}}}
\def\boldV{{\boldsymbol{V}}}

\def\boldSi{\boldsymbol{S_{i}}}

\def\boldSi{\boldsymbol{S_{i}}}
\def\boldXt{\boldsymbol{X_{t}}}
\def\boldSt{\boldsymbol{S_{t}}}

\def\bolde{\mathbf{e}}
\def\boldR{\mathbf{R}}
\def\boldTau{\boldsymbol{\tau}}
\def\naturals{\mathbb{N}}
\def\integers{\mathbb{Z}}

\def\reals{\mathbb{R}}
\def\fieldq{\mathcal{F}_{q}}
\def\fieldpi{\mathcal{F}_{\pi}}
\def\setU{\mathcal{U}}
\def\setX{\mathcal{X}}
\def\setS{\mathcal{S}}
\def\StateAlphabet{\setS}
\def\setY{\mathcal{Y}}
\def\setM{\mathcal{M}}
\def\setV{\mathcal{V}}
\def\prob{\text{Pr}}
\def\cl{\mbox{cl}}
\def\cocl{\mbox{cocl}}

\font\Bigmath=cmsy10 scaled \magstep2
\font\bigmath=cmsy10 scaled \magstep1
\def\diamondplusminusone{\mathrel{%
  \ooalign{\raise.29ex\hbox{$\scriptscriptstyle+$}\cr\hss$\diamond$\hss}}}
\def\diamondplus{\mathrel{%
  \ooalign{$+$\cr\hss\lower.255ex\hbox{\Bigmath\char5}\hss}}}
\def\diamondplusplusone{\mathrel{%
  \ooalign{$\scriptstyle+$\cr\hss\lower.29ex\hbox{\bigmath\char5}\hss}}}
\input pstricks

\newcommand{\msout}[1]{\text{\sout{\ensuremath{#1}}}}
\def\setX{\mathcal{X}}
\def\setS{\mathcal{S}}
\def\setY{\mathcal{Y}}
\usepackage{epsf}
\usepackage{amssymb}
\usepackage{graphicx}
\usepackage{amsmath}
\usepackage{mathrsfs}
\usepackage{ulem}
\newcommand{\comment}[1]{}
\begin{document}

\newtheorem{remark}{\it Remark}
\newtheorem{thm}{Theorem}
\newtheorem{definition}{Definition}
\newtheorem{lemma}{Lemma}
\newtheorem{example}{Example}

\title{Achievable rate region based on coset codes for multiple access channel with
states}

\author {Arun~Padakandla and S. Sandeep~Pradhan,
~\IEEEmembership{Member,~IEEE}%
\thanks{Arun Padakandla and S. Sandeep Pradhan are with the Department of Electrical
and Computer Engineering, University of Michigan, Ann Arbor
48109-2122, USA.}
\thanks{This work was supported by NSF grant CCF-1116021.}}
\maketitle
\begin{abstract}
We prove that the ensemble the nested coset codes built on finite fields achieves the
capacity of arbitrary discrete memoryless point-to-point channels. Exploiting it's
algebraic structure, we develop a coding technique for communication over general discrete
multiple access channel with channel state information distributed at the transmitters. We
build an algebraic coding framework for this problem using the ensemble of Abelian group
codes and thereby derive a new achievable rate region. We identify non-additive and
non-symmteric examples for which the proposed achievable rate region is strictly larger
than the one achievable using random unstructured codes.
\end{abstract}
\section{Introduction}
\label{Sec:Introduction}
The most common technique of proving achievability of rate regions in information theory
is random
coding\footnote{The other known techniques are based on Feinstein's lemma
\cite{195409TIT_Fei} and
graph decomposition \cite{198101TIT_CsiKor}.}. Traditionally, the distribution induced on
the ensemble of
codes is such that
individual codewords are mutually independent. Furthermore, in communication models with
multiple terminals, codebooks associated with these terminals are mutually independent of each other. Such
an analysis has proved sufficient for single user and particular multi-terminal
communication problems.\footnote{However, characterization of optimal performance in
many multi-terminal communication problems such as distributed source coding,
interference channel, broadcast channel, multiple description coding remain open.}

The problem of distributed reconstruction of modulo-2 sum of binary correlated sources studied by
K\"orner and Marton \cite{197903TIT_KorMar} proved to be the first exception. As against
to partitioning the source codes independently and uniformly into bins, they propose
partitioning using cosets of a common linear code, thus building dependency across
codebooks and codewords. Crucially exploiting the property of closure under addition of
these cosets, they analyze a coding technique that strictly outperforms the best known
strategy based on independent unstructured codes. Recently, a similar phenomenon has been
identified by Philosof and Zamir \cite{200906TIT_PhiZam} for a particular example of
multiple access channel with state information distributed at the transmitters (MAC-DSTx).
Restricting their attention to a binary symmetric noiseless additive doubly dirty MAC-DSTx
(BDD-MAC), they propose a partition of the two channel codes into bins using cosets of a
common linear code.\footnote{Recall that communicating over a channel with state
information at transmitter involves binning of the codebooks of the two transmitters
\cite{1980MMPCT_GelPin}.} They propose a coding 
technique, henceforth referred to as PZ-technique, that achieves the capacity of BDD-MAC
and thereby prove strict sub-optimality of the best known coding technique based on
independent unstructured codes. This is in contrast to point-to-point channels with state
information at the transmitter (PTP-STx) where unstructured codes achieve the
capacity \cite{1980MMPCT_GelPin}.

Nevertheless ingenious, PZ-technique \cite{200906TIT_PhiZam} is very specific to the additive and symmetric nature of the BDD-MAC studied therein. This technique being strictly more efficient than the currently known best strategy based on independent unstructured codes raises the following question. Is there a general coding framework for communicating over an \textit{arbitrary} discrete MAC-DSTx, that reduces to the PZ-technique for the BDD-MAC, and that would yield an achievable rate region strictly larger than the best known achievable rate region using unstructured independent codes even for non-additive and non-symmetric MAC-DSTx?	

In this article, we propose an algebraic framework for communication over an arbitrary
MAC-DSTx and thereby answer the above questions in the affirmative. Our first step is to
generalize the ensemble of linear codes employed in \cite{200906TIT_PhiZam}. It is well
known that linear codes do not achieve the capacity of point-to-point channels with or
without state information available at the transmitter. They achieve only the mutual
information of the channel with uniform input distribution. We propose, in section
\ref{Sec:NestedCosetPTP-STxCodes}, an ensemble of codes, called \textit{nested coset
codes}, composed of two linear codes with one of them being a subset of the other, and
prove in section
\ref{SubSec:NestedCosetCodesAchieveCapacityOfArbitraryPointToPointChannels},
that they achieve the capacity of arbitrary PTP-STx, which is the first main result of
this article. Using nested coset codes we are able to induce non-uniform single-letter
distributions on the input alphabet while retaining much of useful algebraic structure and
thus match any input distribution to the channel.\footnote{The technique of Gallager
\cite{Gal-ITRC68} involving a non linear mapping preceded by a linear code does not
preserve the algebraic structure of the code.} Achieving the capacity of arbitrary PTP-STx
relies
on employing joint typical encoding and decoding. The foundation of our framework is
therefore a study of codes endowed with an algebraic structure, nested coset codes in this
case, using typical set encoding and decoding.

We present our coding scheme for MAC-DSTx in three pedagogical stages. We begin by
identifying two key elements of PZ-technique 1) decoding mod$-2$ sum, instead of the pair
of codewords chosen by the two transmitters and 2) choosing the bins of each user's code
to be cosets of a common linear code to enable containment of the range of this mod$-2$
sum. The first stage, presented in section
\ref{SubSec:AnAchievableRateRegionForArbitraryMAC-DSTxUsingNestedCosetCodes}, captures all
of the nontrivial elements of our framework in it's simplest setting. In this stage we
employ nested coset codes built on finite fields, to decode the sum of codewords. The
analysis of this technique enables us to derive a new achievable rate region for MAC-DSTx.
The key elements of the first stage are (i) the use of nested coset codes to induce
non-uniform input distributions, (ii) the use of joint typical encoding and decoding that
enables us to analyze the probability of error over an \textit{arbitrary} MAC-DSTx that is
not constrained to be additive or symmetric, and (iii) an analysis of decoding of the sum
of the pair of transmitted codewords chosen from two dependent codebooks. Indeed, the
analysis of joint typical encoding and decoding of correlated codebooks with statistically
dependent codewords involves several new elements. The reader is encouraged to peruse
these in the proof of theorem \ref{Thm:AchievableRateRegionUsingNestedCosetCodes}.

The significance of the rate region proved achievable in the first stage is illustrated
through examples in section \ref{SubSec:Examples}.\footnote{The coding technique proposed
in the first stage reduces to that proposed in \cite{200906TIT_PhiZam} for BDD-MAC and
moreover Philosof and Zamir have proved strict sub-optimality of unstructured independent
coding for BDD-MAC. This in itself establishes significance of theorem
\ref{Thm:AchievableRateRegionUsingNestedCosetCodes}. Notwithstanding this, it is easy to
argue significance of our generalization by appealing to continuity. An additive channel
can be perturbed slightly to result in a non-additive channel for which the technique
proposed in \cite{200906TIT_PhiZam} may not be applicable as is. By continuity of the rate
regions as a function of the channel parameters, one can see why the proposed coding
scheme must perform strictly better than unstructured independent coding. Example
\ref{Ex:Example5} presented in section \ref{SubSec:Examples} corroborates
this.} In particular, we provide an example for which it is
necessary to induce non-uniform input distributions and is more efficient to decode the
sum of transmitted codewords. We also randomly perturb the BDD-MAC and demonstrate that
coding framework proposed herein can outperform unstructured independent codes. The
channels being non-additive, it is significantly harder to provide analytical comparisons,
and hence we resort to direct computation of rate regions achievable using unstructured
independent and nested coset codes. These examples illustrate that structured-code based
strategies do not hinge on the channel being additive but would benefit as long as the
optimizing test channel from the auxiliary inputs to the channel output is not far from
additive.

Does the rate region proved achievable using nested coset codes subsume the largest known
achievable rate region using unstructured independent codes? It is our belief that
strategies based on structured codes are not in lieu of their counterparts based on
unstructured codes. In most cases, structured codes enable efficient decoding of a
`compressive'\footnote{$f(U_{1},U_{2})$ is `compressive' if $H(f(U_{1},U_{2}))$ is
significantly lower than $H(U_{1},U_{2})$.} function of the two codewords.
However, for decoding both the codewords, it turns out the strategy of using a common
linear code to effect partition of the two codebooks is not optimal, instead one has to
employ two independent linear codes. The rate region achieved using the latter strategy is
equivalent to that achieved using unstructured independent codes.\footnote{Indeed, for the
problem of distributed reconstruction of modulo$-2$ sum of binary sources, K\"orner Marton
strategy \cite{197903TIT_KorMar} based on common linear codes is outperformed by
Slepian-Wolf \cite{197307TIT_SleWol} strategy (or equivalently the strategy of Csisz\'ar
based on independent linear codes \cite{198208TIT_Csi}.) for the class of source
distributions for which the modulo$-2$ sum is not sufficiently compressive.
 More precisely, if $H(X \oplus Y) > \frac{H(X,Y)}{2}$, then it is better to reconstruct $X \oplus Y$ using the technique of Slepian-Wolf or Csisz\'ar.} This leads us to the second stage of our 
coding scheme which is presented in section \ref{Sec:AUnifiedAchievableRateRegion}.
Following the approach of Ahlswede and Han \cite[Section VI]{198305TIT_AhlHan}, we glue
together structured and unstructured coding techniques to derive an achievable rate region for communicating over a MAC-DSTx that combines structured and
unstructured coding techniques. We present an example to
illustrate how the gluing of unstructured and structured coding techniques can yield a
rate region larger than either one, and their union. We remark that in spite of our
inability to compute the achievable rate region proposed in section
\ref{Sec:AUnifiedAchievableRateRegion}, we are able to demonstrate the significance of
the same through an example.

If the channel is far from additive, it may not be efficient to decode the sum, with
respect to a finite field, of codewords. For example, if the MAC-DSTx is doubly dirty with
field addition replaced by addition of an Abelian group, referred to as group
addition or group sum, then it is natural to decode group sum of codewords. In other
words, the technique of decoding sum of codewords must be generalized to decoding any
arbitrary bivariate function of the auxiliary inputs. In the third stage of our coding
scheme, presented in section \ref{Sec:EnlargingAchievableRateRegionUsingCodesOverGroups},
we consider decoding the group sum of the codewords. Specifically, codebooks are built
over Abelian group alphabets and each encoder is provided with codebooks that possess a
certain group structure. Analogous to the
first stage, we propose joint typical encoding and decoding of group codes. Though
essential elements of this analysis are similar 
to that of decoding sum of codewords chosen from nested coset codes over an
\textit{arbitrary} MAC-DSTx, the algebraic structure of a Abelian group being looser,
leads to several new elements.

The importance of (i) decoding an appropriate bivariate function of codewords, and (ii)
endowing codebooks with the appropriate algebraic structure is illustrated through an
example discussed in section \ref{Sec:EnlargingAchievableRateRegionUsingCodesOverGroups}.
Specifically, we indicate using numerical computation that for a quaternary doubly dirty
MAC-DSTx (QDD-MAC) wherein the
operation is mod$-4$ addition, decoding mod$-4$ sum, which is the group operation in the
quaternary alphabet, of the codewords strictly outperforms both independent unstructured
and
nested coset codes based strategies. In fact, significant gains for this problem are
achievable using Abelian group codes. The reader is encouraged to peruse details in
section \ref{Sec:EnlargingAchievableRateRegionUsingCodesOverGroups}.

Several findings in the context of multi-terminal communication problems point to
efficient strategies based on structured codes. Nazer and Gastpar \cite{200710TIT_NazGas}
propose a strategy based on linear codes for computing the sum of sources over additive
multiple access channels that outperforms earlier known strategies. Building on this technique, we develop a framework for computing sum of sources over an arbitrary multiple access channel in \cite{201307ISIT_PadPraComputation}. Sridharan et. al.
\cite{200809Allerton_SriJafVisJafSha} propose a coding technique based on lattices for
communicating over a $K-$user Gaussian interference channel ($K\geq 3$) that outperforms a
natural extension of Han-Kobayashi technique \cite{198101TIT_HanKob} under the Gaussian
input distribution. We propose an analogous coding technique based on nested linear codes
\cite{201207ISIT_PadSahPra} for the general discrete $3-$user interference channel and
identify an example for which the proposed technique outperforms the natural extension of
Han-Kobayashi technique \cite{198101TIT_HanKob}. Krithivasan and Pradhan
\cite{201103TIT_KriPra} propose a framework based on structured codes for the distributed
source coding problem that outperforms the best known strategy based on unstructured
independent codes due to Berger and Tung \cite{Berger-MSC}. We have employed the same
ensemble of nested
coset codes to strictly enlarge the largest known achievable rate region\footnote{This is
obtained by a natural extension of Marton's \cite{197905TIT_Mar} coding technique proposed
for $2-$user
broadcast channel.}
for the general $3-$user discrete broadcast channel in \cite{201207arXiv_PadPra}.

We summarize by stating the significance of our contribution. Nested coset codes is
currently the only ensemble of codes possessing an algebraic structure that has been
proven to be optimal for general point-to-point channels. We employ the same to derive the
largest known achievable rate region for a general discrete MAC-DSTx. Perhaps more
importantly, our findings hint at a general theory of structured codes. Thus far, linear
and nested linear codes have been employed to derive communication strategies for
particular
additive source and channel coding problems that outperform the best known
techniques based on independent unstructured codes. Our
findings indicate that strategies based on structured codes can be employed to analyze
more intelligent encoding and decoding techniques for an arbitrary multi-terminal
communication problem. This opens up the possibility of exploiting new degrees of freedom
to enlarge achievable rate regions for many multi-terminal communication problems that
have resisted a solution.

We begin with remarks on notation and state the form
of typicality employed herein.
\section{Preliminaries: Notation and Typicality}
\label{Sec:PreliminariesNotationAndTypicality}
\subsection{Notation}
\label{SubSec:Notation}
We employ notation that is now widely employed in information theory literature supplemented by the following.
\begin{itemize}
 \item We let $\naturals, \reals$ denote the set of natural
numbers and real numbers respectively. Calligraphic letters such as $\setX$, $\setY$ are employed exclusively to denote
finite sets. $\fieldq$ denotes the finite field of cardinality $q$. For any set $A$, $\cl\left(A\right),\cocl\left(A\right)$ denote closure of $A$ and closure of the convex hull of $A$ respectively. If $A$ is a finite set, $\left| A \right|$ denotes cardinality of $A$.
\item For positive integers $i\leq j$, $[i:j] \define \left\{ i,i+1,\cdots,j \right\}$.
We let $[j] \define [1:j]$.
\item While $+$ denotes addition in $\reals$, we let $\oplus$ denote addition in a finite
field. The particular finite field, which is uniquely determined
(up to an isomorphism) by
it's cardinality, is clear from context. When ambiguous, or to enhance clarity, we
specify addition in $\fieldq$ using $\oplus_{q}$. For $a,b \in \fieldq$, $a \ominus b
\define a \oplus (-b)$, where $(-b)$ is the additive inverse of $b$.
\item If $f: \setU \rightarrow \setX$ is a map, the $n$-letter
extension of $f$ denoted $f^{n} : \setU^{n}
\rightarrow \setX^{n}$ is defined $f^{n}\left(u^{n}\right) := \left(
f\left(u_{i}\right) : i=1,2,\cdots,n\right)$.
\item We employ the standard notation for probability mass functions (pmf). For example, if
$p_{UXSY}$ is a pmf on $\setU \times
\setX \times \StateAlphabet \times \setY$, then $p_{UY}$ is the corresponding
marginal on $\setU \times \setY$. $p_{UY}^{n}$ is the pmf
on $\setU^{n}\times\setY^{n}$ obtained as an $n-$fold product of $p_{UY}$ i.e.,
$p_{UY}^{n}(u^{n},y^{n})= \prod_{i=1}^{n}p_{UY}(u_{i},y_{i})$. We write $U \sim p_{U}$ if $p_{U}$ is the pmf of $U$.
\item The $\log$ and $\exp$ functions are taken with respect to base $2$.
\item For $a \in \naturals$, $\pi(a) \define \min \left\{ k \in \naturals: k \geq a, k
\mbox{ is a prime power}\right\}$.
\item For a pmf $p_{UXSY}$ defined on $\setU \times \setX \times \StateAlphabet \times \setY$, let
\[\mathcal{R}(p_{UXSY},U) \define \left\{ u \in \setU :  \exists
(x,s,y) \in \setX \times \StateAlphabet \times \setY  : p_{UXSY}(u,x,s,y) >
0\right\}\]
denote the essential range of $U$. When clear from context, we omit the
underlying pmf and let $\mathcal{R}(U)$ denote $\mathcal{R}(p_{UXSY},U)$.
\end{itemize}
\subsection{Typicality}
\label{SubSec:Typicality}
We adopt a slightly modified form \cite{200801LN-TMC} of the notion of robust typicality
as proposed by Orlitsky and Roche \cite{200103TIT_OrlRoc} and adopted in
\cite{201201NIT_ElgKim}. In the sequel, we provide definitions and state the results
employed in this article, in their simplest form. Since the following results have been
well documented in books such as \cite{CK-IT2011}, \cite{201201NIT_ElgKim},
\cite{2006EIT_CovTho} among others, we omit proofs, and allude to one of the above
references for the same.

Let $\setX_{1}$, $\setX_{2}$ be finite sets and $X \define (X_{1},X_{2})$ be a pair of
random variables taking values in $\setX \define \setX_{1} \times \setX_{2}$ with pmf
$p_{X} \define p_{X_{1}X_{2}}$. Let $X^{n} \define (X_{1}^{n}, X_{2}^{n})$ be $n$
independent and identically distributed copies of $X$. For a pair $a=(a_{1},a_{2}) \in
\setX$, and an $n-$tuple $x^{n} \define (x_{1}^{n},x_{2}^{n}) \in \setX^{n}$, let
$N(a|x^{n}) = \sum_{i=1}^{n} 1_{\left\{(x_{1i},x_{2i})=a\right\}}$ be the number of
occurrences of $a$ in $x^{n}$. Lastly, for $j \in \left\{ 1,2\right\}$, let
$\msout{j} \in \left\{ 1,2\right\} \setminus \left\{ j \right\}$ denote the element in it's complement. We are
now set to define typical set. For any $\delta > 0$, let
\begin{equation}
 \label{Eqn:TypicalSet}
T_{\delta} \define \left\{   x^{n}
\in \setX^{n}: \left| \frac{N(a\big|x^n)}{n} - p_{X}(a)
\right| \leq \frac{\delta p_{X}(a)}{\log |\setX|} \mbox{ for all }
a \in \setX \right\}\nonumber
\end{equation}
be the typical set on $\setX$ with respect to pmf $p_{X}$ and parameter $\delta
> 0$. For $j=1,2$, the projection
\begin{equation}
 \label{Eqn:ProjectedTypicalSet}
T_{\delta}(X_{j}) \define \left\{   x_{j}^{n}
\in \setX_{j}^{n}: \text{ there exists } x_{\msout{j}}^{n} \in
\setX_{\msout{j}}^{n} \text{ such that } (x_{1}^{n},x_{2}^{n}) \in T_{\delta}
\right\}\nonumber
\end{equation}
is the typical set on $\setX_{j}$ with respect to pmf $p_{X}$ and parameter
$\delta > 0$. For $j=1,2$ and any $x_{\msout{j}}^{n} \in \setX_{\msout{j}}^{n}$,
\begin{equation}
 \label{Eqn:ConditionalTypicalSet}
 T_{\delta}(X_{j}|x_{\msout{j}}^{n}) \define \left\{   x_{j}^{n}
\in \setX_{j}^{n} \text{ such that } (x_{1}^{n},x_{2}^{n}) \in T_{\delta}
\right\}\nonumber
\end{equation}
is the typical set on $\setX_{j}$ conditioned on $x_{\msout{j}}^{n}$ with respect to
distribution $p_{X}$ and parameter
$\delta > 0$. Before we state the basic results, the following remark is worth noting.

\begin{remark}
 \label{Rem:TheProbabilityOfALetterIs0ImplieItDoesNotAppearInATypicalSequence}
 If for any $a \in \setX$, $p_{X}(a) =0$, and $x^{n} \in T_{\delta}$, then
$N(a|x^{n}) = 0$.
\end{remark}

\begin{lemma}
 \label{Lem:BoundsOnProbabilityOfTypicalSequence}
If $x^{n} \in T_{\delta}$, then for every $n \in \naturals$, $|\frac{1}{n}\log p_{X^{n}}(x^{n})+H(X)| \leq \delta$, $|\frac{1}{n}\log p_{X_{j}^{n}}(x_{j}^{n})+H(X_{j})| \leq \delta$ for $j \in [2]$ and therefore $|\frac{1}{n}\log p_{X_{j}^{n}|X_{\msout{j}}^{n}}(x_{j}^{n}|x_{\msout{j}}^{n})+H(X_{j}|X_{\msout{j}})| \leq 2\delta$.
\end{lemma}

\begin{lemma}
 \label{Lem:ATypicalSetIsHighlyProbable}
 For every $\epsilon > 0$, $\delta > 0$, there exists $N(\epsilon, \delta) \in
\naturals$, such that for every $n \geq N(\epsilon, \delta)$, $\prob(X^{n} \in
T_{\delta}) \geq 1-\epsilon$, and therefore, $\prob(X_{j}^{n} \in
T_{\delta}(X_{j})) \geq 1-\epsilon$, for each $j \in [2]$. Moreover, \begin{equation}\label{Eqn:SanovBoundOnAtypicalSet}
\prob \left(  X^{n} \notin T_{\delta}\right) \leq \exp \left\{ -n\lambda \delta^{2}  \right\}\mbox{, where }\lambda = \dfrac{1}{(\log |\setX|)^{2}}\min_{a \in \setX} \left\{ p_{X}^{2}(a):a \in \setX,  p_{X}(a) > 0\right\}\nonumber
\end{equation}
\end{lemma}
While the first statement of lemma \ref{Lem:ATypicalSetIsHighlyProbable} can be proved using Cheybyshev inequality, the second statement, due to Hoeffding \cite{1965MMAMS_Hoe}, Sanov \cite{1957MMMS_San}, requires a finer analysis. The reader is referred to \cite[Problem 11 Pg 43]{CK-IT2011} for an idea of the proof.
\begin{lemma}
\label{Lem:BoundsOnSizeOfTypicalSet}
For every $\delta > 0$, there exists $N_{1}(\delta),
N_{2}(\delta) \in
\naturals$, such that,
\begin{enumerate}
\item for every $n \geq N_{1}(\delta)$, $\exp\left\{ 
n(H(X)-2\delta)\right\} \leq \left| T_{\delta}  \right| \leq \exp\left\{ 
n(H(X)+2\delta)\right\}$, and
\item for every $n \geq N_{2}(\delta)$, $\exp\left\{ 
n(H(X_{j})-2\delta)\right\} \leq \left| T_{\delta}(X_{j})  \right| \leq \exp\left\{ 
n(H(X_{j})+2\delta)\right\}$.
\end{enumerate}
\end{lemma}


\begin{lemma}
 \label{Lem:ConditionalTypicalSetOccursWithHighProbability}
 For every $\epsilon > 0$, $\delta > 0$, there exists $N(\epsilon, \delta) \in
\naturals$, such that for every $n \geq N(\epsilon, \delta)$, $x_{\msout{j}}^{n} \in
T_{\delta}(X_{\msout{j}})$, implies $\prob(X_{j}^{n} \in
T_{2\delta}(X_{j}|x_{\msout{j}}^{n})|X_{\msout{j}}^{n}=x_{\msout{j}}^{n}) \geq
1-\epsilon$ and therefore $\prob \left( X_{\msout{j}} \in T_{\delta}(X_{\msout{j}}), X^{n} \notin T_{2\delta} \right) \leq \epsilon$.
\end{lemma}

\begin{lemma}
\label{Lem:BoundsOnSizeOfConditionalTypicalSet}
For every $\delta > 0$, there exists $N(\delta) \in
\naturals$, such that,
for every $n \geq N(\delta)$, $x_{\msout{j}}^{n} \in
T_{\delta}(X_{\msout{j}})$ we have $\exp\left\{ 
n(H(X_{j}|X_{\msout{j}})-3\delta)\right\} \leq \left|
T_{2\delta}(X_{j}|x_{\msout{j}}^{n}) \right| \leq \exp\left\{ 
n(H(X_{j}|X_{\msout{j}})+3\delta)\right\}$.
\end{lemma}
\section{PTP-STx: Definitions and capacity}
\label{Sec:PTP-STxDefinitionsCapacityAndNestedCosetCodes}

We begin with a description of the PTP-STx \cite{1980MMPCT_GelPin} and relevant notions in section \ref{Subsec:DefinitionsPTP-STxAchievabilityAndCapacity}. In section \ref{Subsec:CapacityOfPTP-STx}, we state the capacity region of PTP-STx as derived by Gelfand and Pinsker.

\subsection{Definitions - PTP-STx, achievability and capacity}
\label{Subsec:DefinitionsPTP-STxAchievabilityAndCapacity}

Consider a point-to-point channel with knowledge of channel state at
transmitter (PTP-STx) studied by Gelfand and Pinsker \cite{1980MMPCT_GelPin}. Let
$\setX$ and $\setY$ denote finite input and output alphabet
sets respectively. Transition
probabilities depend on a random parameter, called state, that takes values in a finite
set $\setS$. The discrete time channel is (i) time invariant, i.e.,
pmf of $Y_{i}$, the output at time $i$, conditioned on $(X_{i}$,
$S_{i})$, the input and state at time $i$, is invariant with $i$, (ii)
memoryless, i.e., $Y_{i}$ is conditionally
independent of $(X_{t},S_{t}):1\leq t< i$ given $(X_{i},S_{i})$, and (iii) used without
feedback, i.e., encoder has no knowledge of outputs observed by decoder. Let
$W_{Y|XS}(y|x,s)$ be the probability of observing $y\in\setY$ at the
output given $x \in \setX$ is input to PTP-STx in state $s\in\setS$. The
state at time $i$, $S_{i}$ is (i) independent of $(X_{t},S_{t},Y_{t}):1\leq t < i$, and
(ii) identically distributed for all $i$. Let
$W_{S}(s)$ be probability of PTP-STx being in state $s\in\setS$. We assume
the sequence of states is non-causally available at the encoder. The input is constrained
with respect to a cost function $\kappa : \setX \times \setS \rightarrow
[0,\infty)$. We assume that the cost is time-invariant and additive i.e., cost of input $X^{n}$ to the
channel in state $S^{n}$ is $\bar{\kappa}^{n}(X^{n},S^{n}) \define
\frac{1}{n}\sum_{i=1}^{n}
\kappa(X_{i},S_{i})$. We refer to this channel
as PTP-STx $(\setS,W_{S},\setX,\kappa,\setY,W_{Y|XS})$.
\begin{definition}
\label{Defn:PTP-STxCode}
A PTP-STx code $\left( n, \mathscr{M}, e,d \right)$ consists of (i) an index set
$\setM$ of messages, of cardinality $\mathscr{M}$, 
(ii) an encoder map $e : \setM \times \setS^{n}
\rightarrow \setX^{n}$, and (iii) a decoder
map $d : \setY^{n} \rightarrow \setM$.
\end{definition}
Assuming a uniform pmf on the set of messages, we define the average
error probability and the cost of a PTP-STx code.
\begin{definition}
\label{Defn:PTP-STxErrorProbability}
The error probability of PTP-STx code $(n,\mathscr{M},e,d)$
conditioned on message $m \in \setM$ is 
\begin{equation}\xi(e,d|m) \define \sum_{s^{n} \in \setS^{n}} \sum_{\stackrel{y^{n}
: d( y^{n}  )}{ \neq m}} \!\!\!W_{S^{n}}(s^{n})W_{Y^{n}|X^{n},S^{n}}( y^{n} |
e(m,s^{n}),s^{n} ).\nonumber\end{equation} The average error probability of PTP-STx code
$(n,\mathscr{M},e,d)$ is $\bar{\xi}(e,d) \define \sum_{m=1}^{\mathscr{M}}
\frac{1}{\mathscr{M}}
\xi(e,d|m)$. The average cost of transmitting message $m \in \setM$ is $\tau(e|m) \define
\sum_{s^{n} \in
\setS^{n}}W_{S^{n}}(s^{n})\bar{\kappa}^{n}(e(m,s^{n}),s^{n})$ and the average cost of
PTP-STx code $(n,\mathscr{M},e,d)$ is $\tau(e) \define
\frac{1}{\mathscr{M}}\sum_{m=1}^{\mathscr{M}}\tau(e|m)$.
\end{definition}
\begin{definition}
\label{Defn:PTP-STxAchievabilityAndCapacity}
A rate cost pair $(R,\tau)\in [0,\infty)^{2}$ is achievable if for every $\eta > 0$,
there exists
$N(\eta)\in \naturals$ such that for all $n > N(\eta)$, there exists a
PTP-STx code $(n, \mathscr{M}^{(n)}, e^{(n)},d^{(n)})$ such that (i)
$\frac{\log\mathscr{M}^{(n)}}{n} \geq R-\eta$, (ii) $\bar{\xi}(e^{(n)},d^{(n)})
\leq \eta$, and (iii) average cost $\tau(e^{(n)}) \leq \tau+\eta$. The capacity region is
$\mathbb{C}(\tau) \define \cl{\left\{ R \geq 0: (R,\tau)\mbox{ is achievable} \right\}}$.
\end{definition}

In a celebrated result, Gelfand and Pinsker \cite{1980MMPCT_GelPin} derived a single
letter characterization of $\mathbb{C}(\tau)$. In the next section, we state
this characterization.
\subsection{Capacity of PTP-STx}
\label{Subsec:CapacityOfPTP-STx}

\begin{definition}
\label{Defn:CharacterizationOfTestChannelsForPTP-STx}
Let $\overline{\mathbb{D}}(\tau)$ be the collection of pmfs $p_{VXSY}$ on
$\setV\times\setX\times\setS\times\setY$ such that (i) $\setV$ is a finite set, (ii)
$p_{S} = W_{S}$, (iii)
$p_{Y|XSV}=p_{Y|XS}= W_{Y|XS}$, (iv)
$p_{X|SV}(x|s,v) \in \left\{0,1\right\}$ for all $(v,x,s)\in\setV \times \setX \times
\setS$ and (v) $\mathbb{E} \left\{ \kappa(X,S) \right\} \leq \tau$. Let
\[\mathbb{D}(\tau) = \left\{ p_{VXSY} \in 
\overline{\mathbb{D}}(\tau): |\mathcal{R}(p_{VXSY},V)| \leq \min \{ 
\left(|\setX|\cdot|\setS|\right)^{2},
\left(|\setX|+|\setS|+|\setY|-2\right)\cdot|\setX|\cdot|\setS| \}\right\}.\]For
any pmf $p_{VXSY}$ defined on
$\setV\times\setX\times\setS\times\setY$, let $\alpha(p_{VXSY}) \define
[0,I(V;Y)-I(V;S)]$, and
\begin{equation}\label{Eqn:CharacterizationOfCapacityRegionOfPTP-STxThatIsNotComputable}
\overline{\alpha}(\tau)
\define \cocl\left(
\underset{p_{VXSY} \in
\overline{\mathbb{D}}(\tau)}{\bigcup}\alpha(p_{VXSY})\right),\alpha(\tau)
\define \cocl\left( \underset{p_{VXSY} \in
\mathbb{D}(\tau)}{\bigcup}\alpha(p_{VXSY})\right).\nonumber \end{equation}
\end{definition}
\begin{thm}
\label{Thm:CapacityOfPTP-STx}
 $\mathbb{C}(\tau)=\alpha(\tau)=\overline{\alpha}(\tau)$.
\end{thm}
Gelfand and Pinsker \cite{1980MMPCT_GelPin} proved theorem \ref{Thm:CapacityOfPTP-STx} for
channels without a cost constraint. While the central elements of their proof can be
adopted for cost constrained channels, the sufficiency of restricting to test channels
$p_{VSXY}$ satisfying condition (iv) in definition
\ref{Defn:CharacterizationOfTestChannelsForPTP-STx} is established in \cite[Lemma
2]{200305TIT_BarCheWor}, which is attributed to Cohen. A cardinality bound on
$|\mathcal{V}|$ can be established using Fenchel-Eggleston strengthening of
Carath\'eodory's theorem \cite[Appendix C]{201201NIT_ElgKim} as done in \cite[Lemma
9]{201207arXiv_PadPra}. In particular, one can first prove the upper bound $\min \left\{ 
|\setX|\cdot|\setS|,|\setX|+|\setS|+|\setY|-2 \right\}$ on $|\mathcal{V}|$ for test
channels $p_{VSXY}$ that do not satisfy condition (iv) in definition
\ref{Defn:CharacterizationOfTestChannelsForPTP-STx}. Any such test channel $p_{VSXY}$ can
be mapped to a test channel $p_{\tilde{V}SXY}$ that 
satisfies condition (iv) in definition \ref{Defn:CharacterizationOfTestChannelsForPTP-STx}
without compromising on the achievable rate for which $|\tilde{\mathcal{V}}| \leq
|\setX|\cdot|\setS|\cdot |\mathcal{V}|$.

\section{Nested coset codes achieve capacity of point to point
channels}
\label{Sec:NestedCosetCodesAchieveCapacityOfArbitraryPointToPointChannels}

\subsection{Nested coset PTP-STx codes}
\label{Sec:NestedCosetPTP-STxCodes}

Gelfand and Pinsker prove
achievability of $\mathbb{C}(\tau)$ by averaging
error probability over an ensemble of PTP-STx codes. A code in this ensemble is specified
by a corresponding auxiliary code $\lambda_{O}$ built over an auxiliary
set and a mapping. An
ingenious technique of partitioning (binning) $\lambda_{O}$ into $\mathcal{M}$ bins, one
for each message $m \in \setM$, is the key feature of the coding technique. In the
following, we consider PTP-STx codes which are endowed with a nested coset code structure.
The distinguishing feature of a nested coset PTP-STx
code is that $\lambda_{O}$ is a coset code built over a finite field $\fieldq$ and
$\lambda_{O}$ is partitioned into bins by cosets of a sub coset code $\lambda_{I}
\subseteq \lambda_{O}$. In the sequel, we describe nested coset codes and define a nested
coset PTP-STx code.

We begin with a brief review of coset and nested coset codes. An $(n,k)$ coset code is a collection of vectors in $\fieldq^{n}$ obtained by adding a
bias vector to a $k-$dimensional subspace of $\fieldq^{n}$. If $\lambda_{O}
\subseteq \fieldq^{n}$ and $\lambda_{I}\subseteq \lambda_{O}$ are $(n,k+l)$ and
$(n,k)$ coset codes respectively, then $q^{l}$ cosets $\lambda_{O}/\lambda_{I}$ that
partition $\lambda_{O}$ is a nested coset code.We refer to this as nested coset code
$(n,k,l,g_{I},g_{O/I},b^{n})$ where $b^{n}$ is the bias vector, $g_{I}\in \fieldq^{k \times
n}$ and $g_{O}^{T} = \left[  g_{I}^{T} ~~ g_{O/I}^{T} \right ]\in
\fieldq^{(k+l)\times n}$ are generator matrices of $\lambda_{I}$ and $\lambda_{O}$ respectively.

An informed reader will begin to see the structure we are after. The bins are cosets of the smaller linear code $\lambda_{I}$. The
entire collection of bins forms a coset of the larger linear code $\lambda_{O}$. The message to be sent to the decoder indexes the bins. For this nested coset code, we let
$v^{n}(a^{k},m^{l}) \define a^{k}g_{I}\oplus m^{l}g_{O/I}\oplus b^{n}$ denote a generic
codeword in coset $c(m^{l}) \define \left\{  v^{n}(a^{k},m^{l}) \in \fieldq^{n}: a^{k}
\in \fieldq^{k}\right\}$. We refer to $c(m^{l})$ as the coset corresponding to message
$m^{l}$. The following is therefore a natural characterization of a nested coset PTP-STx
code.
\begin{definition}
 \label{Defn:CodeIsNestedLinear}
 A nested coset PTP-STx code $\left( n, \mathscr{M}, e,d \right)$ over $\fieldq$ is a PTP-STx code if there exists (i) a nested coset code
$\left( n,k,l,g_{I},g_{O/I},b^{n}\right)$ over $\fieldq$, ii) map $f:\fieldq
\times \setS \rightarrow
\setX$ and, (iii) a $1:1$ map $\iota : \setM \rightarrow \fieldq^{l}$ such that $e(m,
s^{n}) \in \left\{ f^{n} \left( a^{k}g_{I}\oplus \iota(m)g_{O
/ I}
\oplus b^{n},s^{n} \right) : a^{k} \in \fieldq^{k}\right\}$.
\end{definition}

\subsection{Achievability}
\label{SubSec:NestedCosetCodesAchieveCapacityOfArbitraryPointToPointChannels}
We now state and prove our first main finding  - nested coset PTP-STx codes
achieve $\mathbb{C}(\tau)$.

\begin{thm}
 \label{Thm:NestedCosetCodesAchieveCapacityOfPTP-STx}
For a PTP-STx $(\setS,W_{S},\setX,\kappa,\setY,W_{Y|XS})$, if $R \in \mathbb{C}(\tau)$, then there exists a sequence
$(n,\mathscr{M}^{(n)},e^{(n)},d^{(n)}):n \geq 1$ of nested coset PTP-STx codes over $\fieldq$ that
achieves $(R,\tau)$, where $q=\pi(\min \{  \left(|\setX|\cdot|\setS|\right)^{2},\left(|\setX|+|\setS|+|\setY|-2\right)\cdot|\setX|\cdot|\setS| \})$.
\end{thm}
\begin{proof}
 Consider any pmf $p_{VXSY} \in \mathbb{D}(\tau)$ and $\eta > 0$. We prove the existence
of a nested coset PTP-STx code $(n,\mathscr{M}^{(n)},e^{(n)},d^{(n)})$ of rate $\frac{\log
\mathscr{M}^{(n)}}{n} \geq I(V;Y)-I(V;S)-\eta$, average cost $\tau(e^{(n)}) \leq \tau +
\eta$ and average probability of error $\overline{\xi}(e^{(n)},d^{{n}}) \leq \eta$ for
every $n \in \naturals$ sufficiently large. The underlying finite field is of cardinality
$\pi(\min \{ 
\left(|\setX|\cdot|\setS|\right)^{2},
\left(|\setX|+|\setS|+|\setY|-2\right)\cdot|\setX|\cdot|\setS| \})$ referred to as $\pi$
for short.

We prove the existence by averaging the error probability over a specific ensemble of
nested coset PTP-STx codes. We begin with a description of a generic code in this
ensemble.

Consider a nested coset PTP-STx code $(n,k,l,g_{I},g_{O/I},b^{n})$, denoted $\lambda_{O}/\lambda_{I}$ with parameters
\begin{eqnarray}
 \label{Eqn:BinningRateOfNestedCosetPTP_STxCode}
 k &\define& \lceil n\left(1 -  \frac{H(V|S)}{\log \pi}+\frac{\eta}{8\log \pi}\right) \rceil \\
 \label{Eqn:SumRateOfNestedPTP_STxCosetCode}
 l &\define& \lfloor n\left(1 -  \frac{H(V|Y)}{\log \pi}-\frac{\eta}{8\log \pi}\right) \rfloor - k.
\end{eqnarray}
The reader is advised to bear in mind our notation is not reflective of $k$ and $l$
being functions of $n$. This abuse of notation reduces clutter. We specify encoding and decoding rules that map $\lambda_{O}/\lambda_{I}$ into a corresponding nested coset PTP-STx code.

The encoder is provided with nested coset code $\lambda_{O}/\lambda_{I}$. The message is
used to index one among $\pi^{l}$ cosets of $\lambda_{O}/\lambda_{I}$. For simplicity, we
assume that the set of messages
$\setM$ is $\setV^{l}$, and $M^{l} \in \setV^{l}$ to
be the uniformly distributed random variable representing user's message. The encoder observes the state sequence
$S^{n}$ and populates the list $L(M^{l},S^{n}) =
\left\{ v(a^{k},M^{l}):
(v(a^{k},M^{l}),S^{n}) \in
T_{\frac{\delta}{2}}(V,S),a^{k} \in \fieldq^{k}\right\}$ of codewords in the coset
corresponding
to the message that are jointly typical with the state sequence, where $\delta \define \frac{1}{2} \min \left\{ \frac{\eta}{48},\frac{\eta \log (|\setV||\setX||
\setS|| \setY|)}{\kappa_{\max}} \right\}$, $\kappa_{\max} \define \max \left\{ 
\kappa(x,s) : (x,s) \in \setX \times \setS\right\}$. If
$L(M^{l},S^{n})$ is empty, it picks a codeword uniformly
at random from coset $c(M^{l})$. Otherwise, it picks a codeword uniformly at
random from $L(M^{l},S^{n})$. Let
$V(A^{k},M^{l})$ denote the picked codeword in either case. The
encoder computes
$X^{n}(M^{l},S^{n}) \define
f^{n}(V^{n}(A^{k},M^{l}),S^{n})$, where $f : \setV \times
\setS \rightarrow \setX$ is any map that satisfies
$p_{X|VS}(f(v,s)|v,s)=1$ for all pairs $(v,s) \in
\setV \times \setS$.  $X^{n}(M^{l},S^{n})$ is fed as input to the channel.

The decoder observes the received vector $Y^{n}$ and populates the list \[D(Y^{n}) \define \left\{ m^{l} \in \setV^{l}: \exists v^{n}(a^{k},m^{l}) \mbox{ such that } (v^{n}(a^{k},m^{l}),Y^{n}) \in T_{\delta}(V,Y) \right\}.\] If $D(Y^{n})$ is a singleton, the decoder declares the content of $D(Y^{n})$ as the decoded message pair. Otherwise, it declares an error.

The above encoding and decoding rules map $\lambda_{O}/\lambda_{I}$ into a corresponding nested coset PTP-STx code $(n,\mathscr{M}^{n},e^{(n)},d^{(n)})$ of rate $\frac{\log \mathscr{M}^{(n)}}{n} = \frac{l\log \pi}{n}$. Observe that, for $n \geq N_{1}(\eta) \define  \lceil\frac{8\log \pi}{\eta}\rceil$, we have
\begin{eqnarray}
\label{Eqn:LowerBoundOnBinningRate}
n\left(1 -  \frac{H(V|S)}{\log \pi}+\frac{\eta}{8\log \pi}\right) \leq k &\leq&
n\left(1 -  \frac{H(V|S)}{\log \pi}+\frac{\eta}{8\log \pi}\right)+1\\
 \label{Eqn:BoundsOnBinningRate}
 &\leq& n\left(1 -  \frac{H(V|S)}{\log \pi}+\frac{\eta}{4\log \pi}\right), 
\end{eqnarray}
and similarly,
\begin{eqnarray}
\label{
Eqn:UpperAndLowerBoundsOnLargeCodeRate }
n\left(1 -  \frac{H(V|Y)}{\log \pi}-\frac{\eta}{8\log \pi}\right) \geq k+l &\geq&
n\left(1 - 
\frac{H(V|Y)}{\log \pi}-\frac{\eta}{8\log \pi}\right) -1\\
\label{Eqn:BoundLargeCode}
& \geq& n\left(1 - 
\frac{H(V|Y)}{\log \pi}-\frac{\eta}{4\log \pi}\right).
\end{eqnarray}
Combining the upper bound for $k$ in (\ref{Eqn:BoundsOnBinningRate}) and the lower bound
for $k+l$ in (\ref{Eqn:BoundLargeCode}), we get
\begin{equation}
\label{Eqn:LowerBoundOnRate}
\frac{l\log \pi}{n}\geq H(V|S)-H(V|Y)-\frac{\eta}{2} = I(V;Y) -
I(V;S)-\frac{\eta}{2}.
\end{equation}
Since $\lambda_{O}/\lambda_{I}$ was a generic nested coset code satisfying (\ref{Eqn:BinningRateOfNestedCosetPTP_STxCode}), (\ref{Eqn:SumRateOfNestedPTP_STxCosetCode}), we have characterized, through our encoding and decoding maps, an ensemble of nested coset PTP-STx codes, one for each $n \in \naturals$, $n \geq N_{1}(\eta)$ of rate at least $I(V;Y) -
I(V;S)-\frac{\eta}{2}$. It suffices to prove existence of a PTP-STx code $(n,\mathscr{M}^{(n)},e^{(n)},d^{(n)})$ in this ensemble, one for each $n\in \naturals$ sufficiently large, with average probability of error $\xi(e^{(n)},d^{(n)}) \leq \eta$ and average cost constraint $\tau(e^{(n)}) \leq \tau + \eta$. This is done by averaging $\xi(e^{(n)},d^{(n)})$ over the ensemble.

Consider a random nested coset code $(n,k,l,G_{I},G_{O/I},B^{n})$, denoted
$\Lambda_{O}/\Lambda_{I}$, with parameters $n,k,l$ satisfying
(\ref{Eqn:BinningRateOfNestedCosetPTP_STxCode}) and
(\ref{Eqn:SumRateOfNestedPTP_STxCosetCode}). Let $G_{I} \in \setV^{k \times n}$, $G_{O/I}
\in \setV^{l \times n}$ and bias vector $B^{n} \in \setV^{n}$ be mutually independent and
uniformly distributed on their respective range spaces. In the sequel, we study the
average probability of error $\xi(e^{(n)},d^{(n)})$ of the corresponding random nested
coset PTP-STx code. Towards this end, we begin with a few remarks on notation. Let
$V^{n}(a^{k},m^{l}) \define a^{k}G_{I} \oplus m^{l}G_{O/I} \oplus B^{n}$ denote a generic
codeword in coset $C(m^{l}) \define\left\{ V^{n}(a^{k},m^{l}) : a^{k} \in
\setV^{k}\right\}$ corresponding to message $m^{l}$.

In order to study $\xi(e^{(n)},d^{(n)})$, we need to characterize the error events
associated with the random nested coset PTP-STx code corresponding to
$\Lambda_{O}/\Lambda_{I}$. If $\epsilon_{1} \define \{ S^{n}
\notin T_{\frac{\delta}{4}}(S) \}$, $\epsilon_{2} \define \{ \phi_{\frac{\delta}{2}}(S^{n},M^{l}) =0
\}$, where $\phi_{\frac{\delta}{2}}(s^{n},m^{l}) \define \sum_{a^{k} \in
\setV^{k}}
1_{\{\left(V^{n}(a^{k},m^{l}), s^{n}\right) \in
T_{\frac{\delta}{2}}^{n}(VS)\}}$, then the error event at the encoder is contained in $\epsilon_{1} \cup \epsilon_{2}$. The error event at the decoder is contained in $\epsilon_{3}^{c} \cup \epsilon_{4}$, where $\epsilon_{3}\define \cup_{a^{k} \in \setV^{k}}\{ (V^{n}(a^{k},M^{l}),Y^{n}) \in T_{\delta}^{n}(V,Y) \}$ and $\epsilon_{4}\define\cup_{\hatm^{l} \neq M^{l}}\cup_{a^{k} \in \setV^{k}}\left\{\left(V^{n}(a^{k},\hatm^{l}), Y^{n}\right) \in T_{\delta}^{n}(V,Y)\right\}$. It suffices to derive an upper bound on $P(\epsilon_{1}) + P(\epsilon_{1}^{c} \cap \epsilon_{2} ) + P((\epsilon_{1} \cup \epsilon_{2})^{c} \cap \epsilon_{3}^{c}) + P(\epsilon_{4})$. In the sequel, we derive an upper bound on each term of the above sum.

Lemma \ref{Lem:ATypicalSetIsHighlyProbable} guarantees the existence of $N_{2}(\eta)
\in \naturals$\footnote{Since $\delta$ is a function of $\eta$, the dependence of
$N_{2}(\eta)$ on $\delta$ is captured through $\eta$.} such that
$\forall n \geq N_{2}(\eta)$, $P(\epsilon_{1}) \leq \frac{\eta}{16}$. In appendix
\ref{Sec:AnUpperBoundOnProbabilityofEpsilon1ComplementIntersectionEpsilon2}, we prove the
existence of $N_{3}(\eta) \in \naturals$, such that $\forall n \geq
N_{3}(\eta)$,
\begin{equation}
 \label{Eqn:UpperBoundOnEncoderErrorProbabilityDerivedInAppendix}
P(\epsilon_{1}^{c} \cap \epsilon_{2}) \leq \exp\left\{-n \log \pi
\left( \frac{k}{n} -
\left(1-\frac{H\left(V|S \right)}{\log \pi} +\frac{3\delta}{4\log \pi}\right)
\right)\right\}.
\end{equation}
Substituting the lower bound in (\ref{Eqn:LowerBoundOnBinningRate}) for $k$ in (\ref{Eqn:UpperBoundOnEncoderErrorProbabilityDerivedInAppendix}), for all $n \geq \max \left\{ N_{1}(\eta),N_{3}(\eta) \right\}$, we have
\begin{equation}
\label{Eqn:UpperBoundOnEncoderErrorProbabilityAfterSubstitutingForkByn}
P(\epsilon_{1}^{c} \cap \epsilon_{2}) \leq \exp\left\{-n \left( \frac{\eta}{8} -\frac{3\delta}{4}\right)\right\} \leq \exp \left\{ -n\left(\frac{7\eta}{64}\right) \right\},
\end{equation}
where the last inequality follows from the choice of $\delta$.

We now consider $P((\epsilon_{1} \cup \epsilon_{2})^{c} \cap \epsilon_{3}^{c})$. An informed reader will recognize that an upper bound on this term can be derived using a typical application of conditional frequency typicality lemma \ref{Lem:ConditionalTypicalSetOccursWithHighProbability}. For the sake of completeness we state the arguments.
The encoding rule ensures, $ \left(\epsilon_{1} \cup \epsilon_{2}\right)^{c}
\subseteq \{  (V^{n}(M^{l},S^{n}),S^{n}) \in
T_{\frac{\delta}{2}}^{n}(V,S) \}$, and thus
\begin{eqnarray} 
\lefteqn{P((\epsilon_{1} \cup \epsilon_{2})^{c} \cap \epsilon_{3}^{c}) \leq P\left(\left\{(V^{n}(M^{l},S^{n}),S^{n}) \in T_{\frac{\delta}{4}}^{n}(V,S)\right\}\cap \epsilon_{3}^{c}\right)}\nonumber\\
&\leq&\!\!\!\!\!\!\!\!\!\!\!\!\sum_{(v^{n},s^{n}) \in
T_{\frac{\delta}{2}}^{n}(V,S)}\!\!\!\!
P((V^{n}(M^{l},S^{n}),S^{n})=(v^{n},s^{n}))P\left(\epsilon_{3}^{c}|(V^{n}(M^{l},S^{n}),S^{
n})=(v^{n},s^{n})\right)\nonumber\\
\label{Eqn:ProbOfYnNotbeingTypicalExpresedUsingLawOfTotalProbability}
&\leq&\!\!\!\!\!\!\!\!\!\!\!\! \sum_{(v^{n},s^{n}) \in
T_{\frac{\delta}{2}}^{n}(V,S)}\!\!\!\!\!\!\!\!\!\!\!\!\!
P((V^{n}(M^{l},S^{n}),S^{n})=(v^{n},s^{n}))P\left( Y^{n} \notin
T_{\delta}(Y|v^{n},s^{n})|(V^{n}(M^{l},S^{n}),S^{n})=(v^{n},s^{n}) \right).
\end{eqnarray}
For any $(v^{n},s^{n}) \in T_{\frac{\delta}{2}}^{n}(V,S) $, note
that, 
\begin{eqnarray}
P\left(\substack{Y^{n}=y^{n},\\X^{n}(M^{l},S^{n})=x^{n}}|\substack{(V^{n}(M^{l},S^{n}),S^{
n } )\\=(v^ { n } ,
s^ {n} )}\right) &=& \prod_{i=1}^{n}
P\left(X_{i}=x_{i},Y_{i}=y_{i}|V_{i}=v_{i},S_{i}=s_{i}  \right)\nonumber
\end{eqnarray}
where the second equality follows from Markov chain $V-(X,S)-Y$. By lemma \ref{Lem:ConditionalTypicalSetOccursWithHighProbability}, there exists
 $N_{4}(\eta) \in \naturals$ such
that for all $ n \geq N_{4}(\eta)$
\begin{equation}
\label{Eqn:ProbOfInputAndOutputBeingTypicalConditionedOnNoErrorAtEncoder}
P((Y^{n},X^{n}(M^{l},S^{n})) \notin T_{\delta}^{n}(X,Y|v^{n},s^{n})
|(V^{n}(M^{l},S^{n}),S^{n})=(v^{n},
s^ { n } )) \leq \frac{\eta}{8}.
\end{equation}
Substituting (\ref{Eqn:ProbOfInputAndOutputBeingTypicalConditionedOnNoErrorAtEncoder}) in (\ref{Eqn:ProbOfYnNotbeingTypicalExpresedUsingLawOfTotalProbability}), we have $P((\epsilon_{1} \cup \epsilon_{2})^{c} \cap \epsilon_{3}) \leq \frac{\eta}{8}$ for all $n \geq N_{4}(\eta)$. It remains to provide an upper bound on $P(\epsilon_{4})$. In appendix
\ref{Sec:AnUpperBoundOnProbabilityofEpsilon4}, we prove the existence of $N_{5}(\eta)
\in \naturals$ such that $\forall n \geq
N_{5}(\eta)$,
$P((\epsilon_{1}\cup\epsilon_{2}\cup\epsilon_{3})^{c}\cap\epsilon_{4}) \leq
\exp \left\{ 
 -n \log \pi \left( 1- \frac{H(V|Y)}{\log \pi} - \frac{3\delta}{2 \log \pi} -
\frac{k+l}{n}\right) \right\}$. For $n \geq \max\left\{ N_{1}(\eta),N_{5}(\eta) \right\}$,
the upper bound for $k+l$ derived in (\ref{Eqn:BoundLargeCode}) is substituted to yield,
$P((\epsilon_{1}\cup\epsilon_{2}\cup\epsilon_{3})^{c}\cap\epsilon_{4}) \leq \exp \left\{
-n \left( \frac{\eta}{8}- \frac{3\delta}{2}\right) \right\}\leq \exp \left\{ -n \left(
\frac{3\eta}{32}\right) \right\}$.

We have therefore proved that for every $n \geq \max \left\{ N_{i}(\eta):i \in [5]
\right\}$, there exists at least one nested coset PTP-STx code $(n,\pi^{l},e,d)$ over $\fieldpi$ for which
$\bar{\xi}(e,d) \leq \frac{\eta}{8}+\exp \left\{ -n \frac{7\eta}{64} \right\}+\frac{\eta}{8}+\exp \left\{ -n\frac{3\eta}{32} \right\}$. For $n \geq \max \left\{ N_{i}(\eta):i \in [6] \right\}$, where $N_{6}(\eta)=\lceil \frac{32}{3\eta}\log \frac{8}{\eta}\rceil$, $\bar{\xi}(e,d) \leq \frac{\eta}{2}$. It only remains to prove this code satisfies the
average cost constraint. It can be verified that $\tau(e) \leq \frac{\eta}{2}
\kappa_{\max}+(1-\frac{\eta}{2})(\tau + \frac{\delta \kappa_{\max}}{2 \log
(|\setX||\setS|)})$. The choice of $\delta$ ensures that $\tau(e) \leq \frac{\eta}{2}
\kappa_{\max} + (\tau + \frac{\eta}{2})$. Since $\kappa_{\max} \in \reals$ is bounded,
this proves the existence of a sequence $(n,\pi^{l(n)},e^{(n)},d^{(n)}):n\geq 1$ of nested
coset PTP-STx codes that achieve $(R,\tau)$ for every $R \in \mathbb{C}(\tau)$.
\end{proof}

The codewords of $\Lambda_{O}$ being uniformly distributed over $\mathcal{F}_{\pi}^{n}$
(c.f. Lemma
\ref{Lem:CodewordsOfRandomLinearCodeUniformlyDistributedAndPairwiseIndependent}(i)), the
probability of it being jointly typical with a typical state sequence $s^{n}$
is $\frac{|T_{\delta}(U|S)|}{\pi^{n}} = \exp\{ n(H(U|S)-\log \pi) \}$. This
indicates that each coset must contain roughly $\frac{q^{n}}{|T_{\delta}(U|S)|}
= \frac{q^{n}}{q^{n(H(U|S))}}= q^{n(\log\pi-H(U|S))}$ codewords. Indeed, it
suffices to partition $\Lambda_{O}$ with a coset of rate $\frac{k}{n}
> 1 -\frac{H(U|S)}{\log\pi}$. $1-\frac{H(U|S)}{\log \pi}$ being in general larger than
$\frac{I(U;S)}{\log \pi}$, we conclude that the constraint of linearity forces us to
increase the rate of the binning code.

However, the sparsity of typical vectors in a random linear code comes to our rescue when
we attempt to pack cosets. The decoder looks for all vectors in the auxiliary code that
are jointly typical with the received vector $Y^{n}$. In unstructured random coding, since
each codeword is individually typical with high probability, the rate of auxiliary code
is bounded from above by $\frac{I(U;Y)}{\log \pi}$. The typical vectors being sparse in
random linear
code, a similar argument as above enables us to enlarge the auxiliary code to a rate
$1-\frac{H(U|Y)}{\log\pi}$. The rate of the code is thus $(1-\frac{H(U|Y)}{\log\pi}) -
(1-\frac{H(U|S)}{\log \pi}) = \frac{I(U;Y)-I(U;S)}{\log\pi}$.

We have thus proved nested coset codes achieve the capacity of arbitrary PTP-STx. The
interested reader is referred to \cite{201207ISIT_SahPraLattices} wherein nested lattice
codes are
proved to achieve capacity of arbitrary continuous point to point channels. In order to
achieve capacity of arbitrary continuous PTP-STx, it is necessary to construct lattices
which induce arbitrary test channels when employed for source quantization.
In a related work, Gariby and Erez \cite{200807ISIT_GarEre} construct lattices for source
coding of continuous sources that yield a family of quantization error distributions.

\section{MAC-DSTx: Definitions, largest known achievable rate region}
\label{Sec:MAC-DSTxDefinitionsLargestKnownAchievableRateRegion}
The rest of the article is aimed at deriving achievable rate regions for the MAC-DSTx. In this section, we lay the necessary groundwork. In particular, we describe MAC-DSTx and precisely state relevant notions such as code, achievability in section \ref{Subsec:DefinitionsMAC-DSTxAchievabilityAndCapacityRegion}. In section \ref{Subsec:LargestKnownAchievableRateRegionUsingUnstructuredCodes}, we provide a characterization of the rate region based on independent unstructured codes. We illustrate this rate region for BDD-MAC in section \ref{SubSec:UnstructuredCodesForBinaryDoublyDirtyMAC} and highlight the reasons for it's sub-optimality. This will set the stage for it's enlargement in subsequent sections.

\subsection{Definitions : MAC-DSTx, code and achievability}
\label{Subsec:DefinitionsMAC-DSTxAchievabilityAndCapacityRegion}
Consider the two user multiple access analogue of PTP-STx \cite{1980MMPCT_GelPin}. Let $\setX_{1}$
and $\setX_{2}$ denote finite input alphabet sets and $\setY$, the output
alphabet set. Transition probabilities depend on a random vector parameter
$\boldS\define(S_{1},S_{2})$, called state, that takes
values in a finite set $\setS\define \setS_{1} \times \setS_{2}$. The
discrete time channel is (i) time invariant, i.e., pmf of
$Y_{i}$, the output at time $i$, conditioned on inputs $\boldsymbol{X_{i}}\define
(X_{1i},X_{2i})$ and state $\mathbf{S_{i}}\define (S_{1i},S_{2i})$ at time $i$, is
invariant with $i$, (ii) memoryless, i.e., $Y_{i}$ is conditionally independent of
$(\boldsymbol{X_{t}},\boldsymbol{S_{t}}):1\leq t <i$ given
$\boldsymbol{X_{i}},\boldsymbol{S_{i}}$, and (iii) used without feedback. Let
$W_{Y|\boldX \boldS}(y|\boldx,\bolds)$ be the probability of observing $y \in \setY$ at the
output given $\boldx \define (x_{1},x_{2}) \in \setX\define\setX_{1} \times \setX_{2}$ is
input to the channel in state $\bolds \define (s_{1},s_{2}) \in \setS$.
The state at time $i$, $\boldSi$ is (i) independent of $(\boldSt,\boldXt,Y_{t}):1\leq t <
i$, and (ii) identically distributed for all $i$. Let
$W_{\boldS}(\bolds)$ be the probability of MAC-DSTx being in state
$\bolds\in\setS$. We
assume $S_{j}^{n}$ is non-causally known to encoder $j$. Input $X_{j}$ is constrained
with respect to a cost function $\kappa_{j}:\setX_{j}  \times \setS_{j} \rightarrow
[0,\infty)$. We assume that the cost is time-invariant and additive i.e., cost of input $X_{j}^{n}$
at input $j$ to the
channel in state $\boldS^{n}$ is $\bar{\kappa_{j}}^{n}(X_{j}^{n},S_{j}^{n}) \define
\frac{1}{n}\sum_{i=1}^{n}
\kappa_{j}(X_{ji},S_{ji})$. We refer to this channel
as MAC-DSTx $(\setS,W_{\boldS},\setX,\kappa,\setY,W_{
Y|\boldX ,\boldS } )$. Towards characterizing a new inner bound for the capacity region
of a MAC-DSTx, we begin with definitions of relevant notions such as achievability
and capacity.
\begin{definition}
\label{Defn:MAC-DSTxCode}
A MAC-DSTx code $\left( n, \mathscr{M}_{1},\mathscr{M}_{2}, e_{1},e_{2},d \right)$
consists of (i) index sets
$\mathcal{M}_{j}$ of messages, of cardinality $\mathscr{M}_{j}$ for $j=1,2$
(ii) encoder maps $e_{j} : \setM_{j} \times \setS_{j}^{n}
\rightarrow \setX_{j}^{n}$ for $j=1,2$, and (iii) a decoder
map $d : \setY^{n} \rightarrow \setM_{1}\times \setM_{2}$.
\end{definition}

We let $\mathscr{M}\define (\mathscr{M}_{1},\mathscr{M}_{2})$,
$\bolde\define(e_{1},e_{2})$ and refer to above as MAC-DSTx code
$(n,\mathscr{M},\bolde,d)$. Assuming the pair of messages to be uniformly distributed, we
define the average error probability and the cost of a MAC-DSTx code as follows.
\begin{definition}
\label{Defn:MAC-DSTxErrorProbability}
The average error probability of MAC-DSTx code $(n,\mathscr{M},\bolde,d)$
conditioned on message $\boldm \define (m_{1},m_{2}) \in \setM\define\setM_{1}
\times \setM_{2}$ is
\begin{eqnarray}\xi(\bolde,d|\boldm) \define \sum_{\bolds^{n} \in
\setS^{n}}W_{\boldS^{n}}(\bolds^{n}) \sum_{y^{n}
: d( y^{n}  )\neq \boldm} W_{Y^{n}|\boldX^{n},\boldS^{n}}( y^{n} |
e_{1}(m_{1},s_{1}^{n}),e_{2}(m_{2},s_{2}^{n}),\bolds^{n}
).\nonumber\end{eqnarray} The average error probability is $\bar{\xi}(\bolde,d)
\define \sum_{\boldm \in \setM}
\frac{1}{\mathscr{M}_{1}\mathscr{M}_{2}}
\xi(\bolde,d|\boldm)$. The average cost of transmitting message pair $\boldm$ is
$\tau(\bolde|\boldm) \define (\tau_{1}(e_{1}|m_{1}),\tau_{2}(e_{2}|m_{2}))$, where
\[
 \tau_{j}(e_{j}|m_{j})\define \sum_{s_{j}^{n} \in
\setS_{j}^{n}}W_{S_{j}^{n}}(s_{j}^{n})\bar{\kappa_{j}}^{n}(e_{j}(m_{j},s_{j}^{n}),s_{j}^{n
} ).
\]
The average cost of the code is $\tau(\bolde) \define \sum_{\boldm \in \setM}
\frac{1}{\mathscr{M}_{1}\mathscr{M}_{2}}
\tau(\bolde|\boldm)$, where $\tau(\bolde) = (\tau(e_{1}),\tau(e_{2}))$.
\end{definition}
\begin{definition}
\label{Defn:MAC-DSTxAchievabilityAndCapacity}
A rate cost quadruple $(\boldR,\boldTau) \in [0,\infty)^{4}$ is
achievable if for
every $\eta > 0$, there exists $N(\eta)\in \mathbb{N}$ such that for
all $n > N(\eta)$, there exists a MAC-DSTx code
$(n,\mathcal{M}^{(n)},\bolde^{(n)},d^{(n)})$
such that (i)
$\frac{\log\mathcal{M}_{j}^{(n)}}{n} \geq R_{j}-\eta$ for $j=1,2$, (ii)
$\bar{\xi}(\bolde^{(n)},d^{(n)})
\leq \eta$, and (iii)
$ \tau_{j}(e_{j}^{(n)}) \leq \tau_{j}+ \eta$, for $j=1,2$. The
capacity region $\mathbb{C}(\boldTau) \define \cocl\left(\left\{ \boldR \in
[0,\infty)^{2}: (\boldR,\boldTau)\mbox{ is achievable}\right\}\right)$.
\end{definition}
The coding technique that achieves capacity of PTP-STx \cite{1980MMPCT_GelPin} can
be generalized to obtain an achievable rate region for MAC-DSTx. For a general MAC-DSTx
this is the largest known inner bound to $\mathbb{C}(\boldTau)$. We provide a
characterization of the same in the following section.

\subsection{Largest known achievable rate region using unstructured codes}
\label{Subsec:LargestKnownAchievableRateRegionUsingUnstructuredCodes}

\begin{definition}
\label{Defn:CharacterizationOfTestChannelsForMAC-DSTx}
Let
$\mathbb{D}(\boldTau)$ be collection of pmfs
$p_{\boldU\boldX\boldS Y}$ on
$\setU^{2}\times\setX\times\setS\times\setY$, where $\boldU$
denotes $U_{1},U_{2}$ and $\setU^{2}$ is a two fold Cartesian product of a finite set
$\setU$, such that
(i) $p_{\boldS} = W_{\boldS}$, (ii) $p_{Y|\boldX\boldS\boldU}=p_{Y|\boldX\boldS}=
W_{Y|\boldX\boldS}$, (iii) $p_{U_{j}|\boldS U_{\msout{j}}}=p_{U_{j}|\boldS} =
p_{U_{j}|S_{j}}$ and $p_{X_{j}|\boldS\boldU X_{\msout{j}}}=p_{X_{j}|\boldS
\boldU}=p_{X_{j}|S_{j}U_{j}}$ for any distinct elements $j,\msout{j} \in \{ 1,2 \}$, (iv)
$p_{X_{j}|S_{j}U_{j}}(x_{j}|s_{j},u_{j}) \in \left\{0,1\right\}$ for all
$(u_{j},s_{j},x_{j}), j=1,2$ and (v)
$\mathbb{E}\left\{{\kappa}_{j}(X_{j},S_{j})\right\}
\leq \tau_{j}$ for $j=1,2$. For
$p_{\boldU\boldX\boldS Y} \in \mathbb{D}(\boldTau)$, let $\alpha(p_{\boldU\boldX\boldS Y})$ be
defined as the set
\begin{eqnarray}
\label{Eqn:LargestSingleLetterAchievableRegionForMAC-DSTx}
\left\{  \begin{array}{ll} (R_{1},R_{2}) \in [0,\infty)^{2}:&R_{1}
\leq I(U_{1};YU_{2})-I(U_{1};S_{1}), R_{2} \leq
I(U_{2};YU_{1})-I(U_{2};S_{2}), \\&R_{1}+R_{2}\leq
I(\boldU;Y)+ I(U_{1};U_{2}) -
\sum_{j=1}^{2}I(U_{j};S_{j})
\end{array}\right\}
\nonumber\end{eqnarray} and
\begin{equation}
\label{Eqn:CharacterizationOfCapacityRegionOfMAC-DSTxThatIsNotComputable}
\alpha(\boldTau)
\define \cocl\left(
\underset{p_{\boldU\boldX\boldS Y}\in
\mathbb{D}(\tau)}{\bigcup}\alpha(p_{\boldU\boldX\boldS
Y})\right).\nonumber \end{equation}
\end{definition}

\begin{thm}
\label{Thm:CapacityOfMAC-DSTx}
$\alpha(\boldTau)\subseteq\mathbb{C}(\boldTau)$.
\end{thm}
Achievability of $\alpha(p_{\boldU \boldX \boldS Y})$ can be proved by employing the
encoding technique proposed by Gelfand and Pinsker \cite{1980MMPCT_GelPin} at each encoder
and joint decoding proposed by Ahlswede \cite{197109ISIT_Ahl}, Liao \cite{1972MMISIT_Lia}.
In the sequel, we provide an illustration of this coding technique for BDD-MAC.
\subsection{Rate region achievable using unstructured codes for BDD-MAC}
\label{SubSec:UnstructuredCodesForBinaryDoublyDirtyMAC}
Philosof and Zamir characterize $\mathbb{C}(\boldTau)$ for BDD-MAC using PZ-technique
and prove $\alpha(\boldTau) \subsetneq \mathbb{C}(\boldTau)$ for the same.
In order to identify the key elements of PZ-technique, we briefly analyze unstructured
coding (this section), PZ-technique (section
\ref{SubSec:NestedLinearCodesForBinaryDoublyDirtyMAC}) and set the stage for a new coding
scheme.

BDD-MAC is a MAC-DSTx with binary alphabets $\StateAlphabet_{j}=\InputAlphabet_{j}=\OutputAlphabet=\{ 0,1\}$, $j=1,2$. The state sequences are independent Bernoulli-$\frac{1}{2}$ processes, i.e., $W_{\boldS}(\bolds) = \frac{1}{4}$ for all $\bolds \in \StateAlphabet$. The channel transition is described by the relation $Y=X_{1} \oplus_{2} S_{1} \oplus_{2} X_{2} \oplus_{2} S_{2}$. An additive Hamming cost is assumed on the input, i.e., $\kappa_{j}(1,s_{j})=1$ and $\kappa_{j}(0,s_{j})=0$ for any $s_{j} \in \StateAlphabet_{j}$, $j=1,2$ and the input is subject to a symmetric cost constraint $\boldTau=(\tau,\tau)$.

We describe the test channel $p_{\boldU \boldS \boldX Y} \in \mathbb{D}(\boldTau)$ that
achieves $\alpha(\boldTau)$. For each user $j$, consider the test channel that achieves
the Gelfand-Pinsker capacity treating the other user as noise i.e.,
$p_{U_{j}S_{j}X_{j}}(0,1,1)=p_{U_{j}S_{j}X_{k}}(1,0,1)=\frac{\tau}{2},~p_{U_{j}S_{j}X_{j}}
(0,0,0)=p_{U_{j}S_{j}X_{j}}(1,1,0)=\frac{1-\tau}{2}$. Philosof and Zamir prove $p_{\boldU
\boldS \boldX}=p_{U_{1}S_{1}X_{1}}p_{U_{2}S_{2}X_{2}}$ achieves $\alpha(\boldTau) =
\left\{ \boldR: R_{1}+R_{2}\leq |2h_{b}(\tau)-1|^{+} \right\}$, where $|\cdot|^{+}$
denotes upper convex envelope.

Let us take a closer look at achievability of the vertex $(2h_{b}(\tau)-1,0)$ using the
above test channel. Since user $2$ has no message to transmit, it picks a single bin with
roughly $2^{nI(U_{2};S_{2})}=2^{n(1-h_{b}(\tau))}$ codewords independently and uniformly
from the entire space of binary vectors. User $1$ picks $2^{nR_{1}}$ bins each with
roughly $2^{nI(U_{1};S_{1})}=2^{n(1-h_{b}(\tau))}$ independently and uniformly distributed
binary vectors. Encoder $2$ observes $S_{2}^{n}$ and chooses a codeword, say $U_{2}^{n}$,
that is within a Hamming distance of roughly $n\tau$ from $S_{2}^{n}$ and transmits
$X_{2}^{n}=U_{2}^{n} \oplus_{2} S_{2}^{n}$. Encoder $1$ performs a similar encoding,
except that it restricts the choice of $U_{1}^{n}$ to the bin indexed by user 1's message,
and transmits $X_{1}^{n}=U_{1}^{n} \oplus_{2} S_{1}^{n}$.

What is the maximum rate $R_{1}$ at which user $1$ can transmit it's message? Decoder
receives $Y^{n}=U_{1}^{n}\oplus_{2}U_{2}^{n}$ and looks for all pairs of codewords that
are jointly typical with $Y^{n}$. Since any pair of binary $n-$length vectors are jointly
typical ($U_{1}$ and $U_{2}$ are independent and uniform), the decoding rule reduces to
finding all pairs of binary $n-$length vectors in the pair of codebooks that sum to the
received vector $Y^{n}$. All bins chosen independently without structure imply that any
bin of user $1$'s codebook when added to the user $2$'s codebook (a single bin) results in
roughly $2^{n(2-2h_{b}(\tau))}$ distinct vectors. Therefore, we cannot hope to pack more
than roughly $\frac{2^{n}}{2^{n(2-2h_{b}(q))}}=2^{n(2h_{b}(q)-1)}$ bins in user $1$'s
codebook. We remark that \textit{an explosion in the range of sum of transmitted codewords
severely limits achievable rate}.

We make a few observations. Effectively, communication occurs over the $(U_{1},U_{2})-Y$
channel and the test channel induces the Markov chain
$(U_{1},U_{2})-U_{1}\oplus_{2}U_{2}-Y$. It would therefore be more efficient to
communicate information over the $U_{1}\oplus_{2}U_{2}-Y$ channel which suggests an
efficient utilization of $U_{1}\oplus_{2}U_{2}-$space. Having chosen codewords in each bin
independently and moreover the two users' bins independently, each message pair utilizes
$2^{n(2-2h_{b}(\tau))}$ vectors in the $U_{1}\oplus_{2}U_{2}-$space. In section
\ref{SubSec:NestedLinearCodesForBinaryDoublyDirtyMAC}, we summarize PZ-technique, wherein
the algebraic structure in the codebooks is exploited for more efficient utilization of
$U_{1}\oplus_{2}U_{2}-$space.

\section{An achievable rate region using nested coset codes}
\label{Sec:AnAchievableRateRegionUsingNestedCosetCodes}

\subsection{Nested linear codes for BDD-MAC}
\label{SubSec:NestedLinearCodesForBinaryDoublyDirtyMAC}

We present PZ-technique proposed for BDD-MAC. The encoding and decoding techniques are
similar to that stated in \ref{SubSec:UnstructuredCodesForBinaryDoublyDirtyMAC} except for
one key difference. The bins of user $1$ and $2$'s codebooks are cosets of a common linear
code. In particular, let $\lambda_{I}$ denote a linear code of rate roughly equal to
$1-h_{b}(\tau)$ that can quantize a uniform source, state $S_{j}^{n}$ in our case, within
an average Hamming distortion of $\tau$. Since user $2$ has no message to transmit, it
employs $\lambda_{I}$ as it's only bin. Encoder $1$ employs $2^{nR_{1}}$ cosets of
$\lambda_{I}$ within a larger linear code, called $\lambda_{O}$, as it's bins. Note that
rate of $\lambda_{O}$ is roughly $R_{1}+1-h_{b}(\tau)$. Encoding rule is as described in
section \ref{SubSec:UnstructuredCodesForBinaryDoublyDirtyMAC}.

The codebook of user $2$ when added to any bin of user $1$'s code results in a coset of
$\lambda_{I}$, and therefore contains approximately at most $2^{n(1-h_{b}(\tau))}$
codewords. Moreover, since $U_{1}^{n}$ lies in $\lambda_{I}$, user $2$'s codeword
$U_{2}^{n}$ and the received vector $Y^{n}=U_{1}^{n} \oplus_{2} U_{2}^{n}$ lie in the
same coset.\footnote{This is also because the channel is noiseless.} Since the channel is
noiseless, user $1$ may employ all cosets of $\lambda_{I}$ and therefore communicate at
rate $h_{b}(\tau)$ which is larger than $2h_{b}(\tau)-1$ for all $\tau \in
(0,\frac{1}{2})$.

Let us identify key elements of PZ-technique. Each message pair corresponds to roughly
$2^{n(1-h_{b}(\tau))}$ vectors in $U_{1}\oplus_{2} U_{2}-$space, resulting in a more
efficient utilization of this space. This indeed is the difference in the sum
rate achievable using independent unstructured codes and PZ-technique. We also note the
\textit{decoder does not attempt to disambiguate the pair $(U_{1}^{n},U_{2}^{n})$ and
restricts to decoding $U_{1}^{n}\oplus_{2}U_{2}^{n}$. This is motivated by the Markov
chain $(U_{1},U_{2})-U_{1}\oplus_{2}U_{2}-Y$ induced by the test channel and the use of
structured codebooks that contain the sum}.

It is instructive to investigate the efficacy of this technique if users $1$ and $2$
employ distinct linear codes $\lambda_{I1},\lambda_{I2}$ of rate $1-h_{b}(\tau)$ instead
of a common linear code $\lambda_{I}$. In this case, each message of user $1$ can result
in $2^{2-2h_{b}(\tau)}$ received vectors which restricts user $1$'s rate to $2h_{b}(q)-1$
and provides no improvement over the unstructured coding technique. We conclude that
\textit{if the bins of the MAC channel code are nontrivial, as in this case due to the
presence of a state, then it maybe beneficial to endow the bins with an algebraic
structure that restricts the range of a bivariate function, and enable the decoder decode
this function of chosen codewords}.

\subsection{Stage I : An achievable rate region for MAC-DSTx using nested coset codes}
\label{SubSec:AnAchievableRateRegionForArbitraryMAC-DSTxUsingNestedCosetCodes}

In this section, we present the first stage of our coding scheme that uses joint typical
encoding and decoding and nested coset codes over an
arbitrary MAC-DSTx. The technique proposed by Philosof and Zamir is specific to the
binary doubly dirty MAC - Hamming cost constraint that induces additive test channels
between the auxiliary and state random variables, and additive and symmetric nature of the
channel. Moreover, linear codes only
achieve the symmetric capacity, and therefore if the output were obtained by passing
$(X_{1}^{n}\oplus_{2}S_{1}^{n},X_{2}^{n}\oplus_{2}S_{2}^{n})$ through an asymmetric MAC,
linear codes though applicable, might not be optimal.

We begin with a characterization of test channels followed by achievability.

\begin{definition}
\label{Defn:TestChannelsRestrictedToFiniteFields}
Let $\mathbb{D}_{f}(\boldTau)
\subseteq \mathbb{D}(\boldTau)$ be the collection of
distributions $p_{\boldV \boldS\boldX  Y}$ on $\mathcal{V}^{2}\times\mathcal{S}\times
\mathcal{X}\times\mathcal{Y}$ where $\mathcal{V}$ is a finite field. For $p_{\boldV \boldX
\boldS Y} \in
\mathbb{D}_{f}(\boldTau)$, let $\beta_{f}(p_{\boldV \boldX \boldS Y})$
be defined as the set
\begin{eqnarray}
\label{Eqn:CharacterizationOfAchievableRateRegionForMACDSTx}
\left\{
\begin{array}{l}\!\!\!\!(R_{1},R_{2}) \in [0,\infty)^{2}\!:\!
R_{1}\!+\!R_{2} \!\leq\! \min\left\{
H(V_{1}|S_{1}),H(V_{2}|S_{2})\right\}\!\!
-H(V_{1}\oplus
V_{2}|Y)\end{array}\!\!
\right\}.\end{eqnarray} Let
\begin{equation}
 \beta_{f}(\boldTau) \define \cocl \left( \underset{p_{\boldV\boldX\boldS Y} \in
\mathbb{D}_{f}(\boldTau)}{\bigcup} \beta_{f}(p_{\boldV\boldX\boldS
Y})  \right)
 \nonumber
\end{equation}
\end{definition}

\begin{thm}
 \label{Thm:AchievableRateRegionUsingNestedCosetCodes}
$\beta_{f}(\boldTau) \subseteq \mathbb{C}(\boldTau)$.
\end{thm}
Before we provide a proof, we state the coding technique and indicate achievability of
promised rates. As stated in section
\ref{SubSec:NestedLinearCodesForBinaryDoublyDirtyMAC}, the key aspect
is to employ cosets of a common linear code as bins for quantizing the state. We
employ three nested coset codes -one each for the two encoders and the decoder- that share
a common inner (sparser) code. We begin by describing the encoding rule. The nested coset
code provided to encoder $j$ is described through a pair of generator matrices $g_{I} \in
\mathcal{V}^{k \times n}$ and $g_{Oj/I} \in \mathcal{V}^{l_{j} \times n}$ where (i) $g_{I}$
and $g_{Oj}^{T}\define \left[ 
g_{I}^{T}~~  g_{Oj/I}^{T}\right]$ are generator matrices of inner (sparser) and complete
(denser) codes respectively, (ii)
\begin{eqnarray}
\label{Eqn:BinningRate}
\frac{k}{n}&>& 1-\frac{\min\left\{H(V_{1}|S_{1}),H(V_{2}|S_{2}) \right\}}{\log \pi}\\
\label{Eqn:ChannelCodingRate}
\frac{k+l_{1}+l_{2}}{n} &<& 1-\frac{H(V_{1} \oplus V_{2}|Y)}{\log \pi}.
\end{eqnarray}
with $\pi \define |\mathcal{V}|$ and (iii) bias vector $b_{j}^{n}$. Let $\lambda_{I}$ and
$\lambda_{Oj}$ denote linear codes corresponding to generator matrices $g_{I}$ and
$g_{Oj}$ respectively. User $j$'s message $M_{j}^{l_{j}} \in \mathcal{V}^{l_{j}}$ indexes
the coset $( a^{k}g_{I}\oplus M_{j}^{l_{j}}g_{Oj/I}\oplus b_{j}^{n}:a^{k} \in
\mathcal{V}^{k})$. Encoder $j$ observes state $S_{j}^{n}$ and looks for a codeword
in the coset indexed by the message that is jointly typical with the state sequence
$S_{j}^{n}$ according to $p_{S_{j}V_{j}}$. If it finds one such codeword, say $V_{j}^{n}$,
a vector $X_{j}^{n}$ is generated according
$\prod_{t=1}^{n}p_{X_{j}|S_{j}V_{j}}(\cdot|S_{jt}V_{jt})$ and $X_{j}^{n}$ is fed as input
to the channel. Otherwise, it declares an error.

Now to the decoding rule. Let $\lambda_{O}$ denote the complete code provided to
the decoder, i.e., the coset code whose (i) generator matrix is $g_{O}^{T} \define \left[
g_{I}^{T}~~ g_{O/I}^{T} \right]$, where $g_{O/I}^{T} \define \left[
g_{O1/I}^{T} ~~ g_{O2/I}^{T} \right]$ and (ii) bias vector $b_{1}^{n}\oplus
b_{2}^{n}$. Having received $Y^{n}$, it lists all codewords in $\lambda_{O}$ that are
jointly typical with $Y^{n}$ with respect to $p_{V_{1}\oplus V_{2},Y}$. If all such
codewords belong to a unique coset (of $\lambda_{I}$ in $\lambda_{O}$) say
$(a^{k}g_{I}\oplus m_{1}^{l_{1}}g_{O1/I} \oplus m_{2}^{l_{2}}g_{O2/I}\oplus
b_{1}^{n} \oplus b_{2}^{n}:a^{k} \in
\mathcal{V}^{k})$, it declares $(m_{1}^{l_{1}},m_{2}^{l_{2}})$ as the pair of decoded
messages. Otherwise, it declares an error.

We pick entries of each of the constituent generator matrices $g_{I},g_{O1/I},g_{O2/I}$
independently and uniformly from $\mathcal{V}$. Lower bound (\ref{Eqn:BinningRate}) enable
us to drive down the probability of encoder not finding a jointly typical codeword in the
indexed coset. This bound can be interpreted easily. If we picked codewords according to
$\prod_{t=1}^{n}p_{V}$, then we need the bin to be of rate
roughly $H(V_{1})-H(V_{1}|S_{1})$. Since we average uniformly over the ensemble of coset
codes, each codeword of a linear code is uniformly distributed over $\mathcal{V}^{n}$.
Hence the bin must of rate at least $\log \pi-H(V_{1}|S_{1})$. The decoder makes an error
with 
arbitrarily small probability if (\ref{Eqn:ChannelCodingRate}) is satisfied. This bound can also be interpreted intuitively. If the codewords were picked according to $p_{V_{1}\oplus V_{2}}$, the upper bound would have been $H(V_{1}\oplus V_{2})-H(V_{1}\oplus V_{2}|Y)$. In this case, the codewords in the sum of nested linear codes are also uniformly distributed over $\mathcal{V}^{n}$, and this explains the bound in (\ref{Eqn:ChannelCodingRate}). From (\ref{Eqn:BinningRate}), (\ref{Eqn:ChannelCodingRate}) it can be verified that $R_{1}+R_{2}=\frac{l_{1}+l_{2}}{n} \leq \min \left\{ H(V_{1}|S_{1}),H(V_{2}|S_{2})-H(V_{1}\oplus V_{2}|Y)\right\}$ is achievable.

We emphasize that joint typical encoding and decoding enables us to decode the sum over an
arbitrary MAC-DSTx. The informed reader will recognize
the need to prove statistical independence of a codeword in a competing sum coset and the
pair
of cosets indexed by the messages. The dependence built across the codewords and cosets as
a consequence of the algebraic structure exemplifies the interplay of algebra and
probability. The following proof details these elements.
\begin{proof}
Let pmf $p_{\boldV \boldX \boldS Y} \in \mathbb{D}_{f}(\boldTau)$,
rate pair $\boldR \in \beta_{f}(p_{\boldV \boldX \boldS Y})$ and $\eta > 0$. We prove
existence of a
MAC-DSTx code $(n,\mathscr{M},\bolde, d)$ whose rate $\frac{\log \mathscr{M}_{j}}{n} \geq
R_{j}-\eta$, average error probability $\overline{\xi}(\bolde,d) \leq \eta$, and average
cost $\tau(e_{j}) \leq \tau_{j}+\eta$ for $j=1,2$.

We begin with a description of the structure of the MAC-DSTx code whose existence we seek to prove. Let $\pi \define |\setV|$ and we assume $H(V_{1}|S_{1}) \geq H(V_{2}|S_{2})$ without loss of generality. Consider a pair of nested coset codes $(n,k_{j},l_{j},g_{I_{j}},g_{O_{j}/I_{j}},b_{j}^{n}):j=1,2$ built over $\setV$, denoted
$\lambda_{O_{j}}/\lambda_{I_{j}}:j=1,2$ with parameters
\begin{eqnarray}
 \label{Eqn:Parameterk1OfNestedCosetCode}
 k_{1} &\define& \lceil n\left(1-\frac{H(V_{1}|S_{1})}{\log \pi}+\frac{\eta_{1}(\eta)}{\log \pi}\right) \rceil,\\
\label{Eqn:Parameterk2OfNestedCosetCode}
k_{2}=k_{1}+k_{+}, \mbox{ where }k_{+} &\define& \lceil n \left(1 -
\frac{H(V_{2}|S_{2})}{\log \pi}+\frac{\eta_{1}(\eta)}{\log \pi}\right) \rceil -k_{1},  \\
\label{Eqn:Parameterl1OfNestedCosetCode}
l_{1} &\define& \lfloor n\left( \frac{R_{1}}{\log \pi}-\frac{\eta_{2}(\eta)}{\log \pi} \right) \rfloor
\\
\label{Eqn:Parameterl2OfNestedCosetCode}
l_{2} &\define& \lfloor n \left( 1+ \frac{R_{2}}{\log \pi}  - \frac{H(V_{2}|S_{2})}{\log \pi}
- \frac{\eta_{3}(\eta)}{\log \pi}\right)\rfloor-k_{2}, \mbox{ and,} \\
\label{Eqn:Identicalk1RowsOfgI1AndgI2}
\mbox{the first $k_{1}$ rows of }g_{I_{1}} \mbox{ and } &&\!\!\!\!\!\!\!\!\!\!\!\!\!\!\!\!\!\!\!g_{I_{2}} \mbox{ are identical i.e., }g_{I_{1},t}=g_{I_{2},t} \mbox{ for } t \in [k_{1}].
\end{eqnarray}
A few remarks on
the structure of $\lambda_{O_{j}}/\lambda_{I_{j}}:j=1,2$ and the relationship between their parameters are in order. For $n \geq N_{1}(\eta)\define \max \left\{ \frac{\log \pi}{\eta_{1}(\eta)},\frac{\log \pi}{\eta_{2}(\eta)},\frac{\log \pi}{\eta_{3}(\eta)} \right\}$, we have
\begin{eqnarray}
\label{Eqn:MACDSTxBoundsOnBinningRate}
\frac{n}{\log \pi}(\log \pi -  H(V_{j}|S_{j})+\eta_{1}(\eta)) \leq& k_{j} &\leq
\frac{n}{\log \pi}(\log \pi - 
H(V_{j}|S_{j})+2\eta_{1}(\eta))\\
\label{Eqn:MACDSTxUpperAndLowerBoundsOnR1}
\frac{n}{\log \pi}(R_{1}-2\eta_{2}(\eta)) \leq &l_{1} &\leq \frac{n}{\log \pi}(R_{1}-\eta_{2}(\eta))\\
\label{Eqn:MACDSTxBoundOnUser2ChannelCode}
\frac{n}{\log \pi}(R_{2}+\log \pi - H(V_{2}|S_{2})-2\eta_{3}(\eta)) \leq& k_{2}+l_{2}& \leq \frac{n}{\log \pi}(R_{2}+\log \pi - H(V_{2}|S_{2})-\eta_{3}(\eta))
\end{eqnarray}
Combining the lower bound in (\ref{Eqn:MACDSTxBoundOnUser2ChannelCode}) and the upper
bound for $k_{2}$ in (\ref{Eqn:MACDSTxBoundsOnBinningRate}), we have
\begin{equation}
 \label{Eqn:MACDSTxLowerBoundOnUser2Rate}
\frac{l_{2} \log \pi}{n} \geq R_{2}-2\eta_{3}(\eta)-2\eta_{1}(\eta)
\end{equation}
and similarly, combining the upper bound for $k_{2}+l_{2}$ in
(\ref{Eqn:MACDSTxBoundOnUser2ChannelCode}) and the upper bound for $l_{1}$ in
(\ref{Eqn:MACDSTxUpperAndLowerBoundsOnR1}), we have
\begin{eqnarray}
 \label{Eqn:MACDSTxUpperBoundOnTheUnionCodeInTermsOfR1AndR2}
k_{2}+l_{1}+l_{2} &\leq& \frac{n}{\log \pi} \left( R_{1}+R_{2} + \log \pi -H(V_{2}|S_{2}) - \eta_{3}(\eta)-\eta_{2}(\eta)  \right)\nonumber\\
\label{Eqn:MACDSTxUpperBoundOnTheUnionCode}
&\leq& \frac{n}{\log \pi} \left( \log \pi -H(V_{1}\oplus V_{2}|Y) - \eta_{3}(\eta)-\eta_{2}(\eta)  \right),
\end{eqnarray}
where (\ref{Eqn:MACDSTxUpperBoundOnTheUnionCode}) follows from $\boldR \in
\beta_{f}(p_{\boldV \boldX \boldS Y})$.

We now specify encoding and decoding rules that map this pair
$\lambda_{O_{j}}/\lambda_{I_{j}}:j=1,2$ of nested coset codes into a MAC-DSTx code. User
$j$ is provided with the nested
coset code $\lambda_{O_{j}}/\lambda_{I_{j}}$.
User $j$'s message is used to index one among $\pi^{l_{j}}$ cosets of
$\lambda_{O_{j}}/\lambda_{I_{j}}$. We assume that the set of messages
$\setM_{j}\define \setV^{l_{j}}$, and $M^{l_{j}}_{j} \in \setV^{l_{j}}$ to
be the uniformly distributed random variable representing user $j$'s message. We let
$v_{j}^{n}(a_{j}^{k_{j}},m_{j}^{l_{j}})\define
a_{j}^{k_{j}}g_{I_{j}}\oplus m_{j}^{l_{j}}g_{O_{j}/I_{j}}\oplus b_{j}^{n}$ denote
a generic codeword in $\lambda_{O_{j}}/\lambda_{I_{j}}$ and $c_{j}(m^{l_{j}}_{j}) \define
(v_{j}^{n}(a_{j}^{k_{j}},m_{j}^{l_{j}}):a_{j}^{k_{j}} \in \setV^{k_{j}})$ denote the coset
corresponding to
message $m_{j}^{l_{j}}$. Encoder $j$ observes the state sequence
$S_{j}^{n}$ and populates the list $L_{j}(M_{j}^{l_{j}},S_{j}^{n}) =
\left\{ v_{j}(a_{j}^{k_{j}},M_{j}^{l_{j}}):
(S_{j}^{n},v_{j}(a_{j}^{k_{j}},M_{j}^{l_{j}})) \in
T_{\eta_{4}(\eta)}(S_{j},V_{j})\right\}$ of codewords in the coset corresponding
to the message that are jointly typical with the state sequence. If
$L_{j}(M_{j}^{l_{j}},S_{j}^{n})$ is empty, it picks a codeword uniformly
at random from coset $c_{j}(M_{j}^{l_{j}})$. Otherwise, it picks a codeword uniformly at
random from $L_{j}(M_{j}^{l_{j}},S_{j}^{n})$. Let
$V_{j}(A_{j}^{k_{j}},M_{j}^{l_{j}})$ denote the picked codeword in either case. The
encoder computes
$X_{j}^{n}(M_{j}^{l_{j}},S_{j}^{n}) \define
f_{j}^{n}(V_{j}^{n}(A_{j}^{k_{j}},M_{j}^{l_{j}}),S_{j}^{n})$, where $f_{j} : \setV_{j} \times
\setS_{j} \rightarrow \setX_{j}$ is any map that satisfies
$p_{X_{j}|V_{j}S_{j}}(f_{j}(v_{j},s_{j})|v_{j},s_{j})=1$ for all pairs $(v_{j},s_{j}) \in
\setV_{j} \times \setS_{j}$.  $X_{j}^{n}(M_{j}^{l_{j}},S_{j}^{n})$ is fed as input to the
channel.

We now specify the decoding rule. The decoder is provided with nested coset code
$(n,k,l,g_{I},g_{O/I},b^{n})$ denoted $\lambda_{O}/\lambda_{I}$ where
$k=k_{2}$, $l=l_{1}+l_{2}$, $g_{I}=g_{I_{2}}$,
$g_{O/I}^{T} \define \left[ g_{O_{1}/I_{1}}^{T} ~~ g_{O_{2}/I_{2}}^{T} \right ]$ and
$b^{n} \define b_{1}^{n}\oplus b_{2}^{n}$. With a slight abuse
of
notation, we let $m^{l} \define (m_{1}^{l_{1}},m_{2}^{l_{2}}) \in
\setV^{l} \define \setV^{l_{1}} \times \setV^{l_{2}}$ represent a
pair of messages and analogously random variable $M^{l} \define (M_{1}^{l_{1}},M_{2}^{l_{2}})$
denote the pair of user messages. For $a^{k} \in \setV^{k}$
and $m^{l} \in \setV^{l}$, let $v^{n}(a^{k},m^{l}) \define
a^{k}g_{I} \oplus m^{l}g_{O/I}\oplus b^{n}$ and $c(m^{l})\define
(v^{n}(a^{k},m^{l}):a^{k} \in \setV^{k})$ denote a generic codeword in
$\lambda_{O}/\lambda_{I}$ and the coset corresponding to the message
pair $m^{l}$ respectively. The decoder observes the received vector $Y^{n}$ and populates
$D(Y^{n}) \define \left\{ m^{l} \in \setV^{l}: \exists v^{n}(a^{k},m^{l}) \mbox{ such that
} (v^{n}(a^{k},m^{l}),Y^{n}) \in T_{\eta_{5}(\eta)}(V_{1}\oplus V_{2},Y) \right\}$. If
$D(Y^{n})$ is a singleton, the decoder declares the content of $D(Y^{n})$ as the decoded
message pair. Otherwise, it declares an error.

The above encoding and decoding rules map every pair
$\lambda_{O_{j}}/\lambda_{I_{j}}:j=1,2$ of nested coset codes that satisfy
(\ref{Eqn:Parameterk1OfNestedCosetCode})-(\ref{Eqn:Identicalk1RowsOfgI1AndgI2}) into a
corresponding MAC-DSTx code $(n,\mathscr{M}^{(n)},\bolde^{(n)},d^{(n)})$ of rate
$\frac{\log \mathscr{M}_{j}^{(n)}}{n} \geq R_{j}-2\eta_{1}(\eta)-2\eta_{2}(\eta)$, thus
characterizing an ensemble, one for each $n$, of MAC-DSTx codes. We average the error
probability over
this ensemble of MAC-DSTx codes by letting the bias vectors
$B_{j}^{n}:j=1,2$ and generator matrices $G_{I_{2}}, G_{O_{j}/I_{j}}:j=1,2$ mutually independent and
uniformly distributed over their respective range spaces. Let
$\Lambda_{O_{j}}/\Lambda_{I_{j}}:j=1,2$ and $\Lambda_{O}/\Lambda_{I}$ denote the random
nested coset codes $(n,k_{j},l_{j},G_{I_{j}},G_{O_{j}/I_{j}},B_{j}^{n}):j=1,2$ and
$(n,k,l,G_{I},G_{O/I},B^{n})$ respectively. For $a_{j}^{k_{j}} \in \setV^{k_{j}}$, $m_{j}^{l_{j}} \in \setV^{l_{j}}$, $a^{k} \in \setV^{k}$,
$m^{l} \in \setV^{l}$, let $V_{j}^{n}(a_{j}^{k_{j}},m_{j}^{l_{j}})\define a_{j}^{k_{j}}G_{I_{j}} \oplus m_{j}^{l_{j}}G_{O_{j}/I_{j}}\oplus B_{j}^{n}:j=1,2$,
$V^{n}(a^{k},m^{l}) \define a^{k}G_{I}\oplus m^{l}G_{O/I}\oplus B^{n}$ denote corresponding random codewords in
$\Lambda_{O_{j}}/\Lambda_{I_{j}}:j=1,2$ and $\Lambda_{O}/\Lambda_{I}$ respectively. Let
$C_{j}(m^{l_{j}}_{j}) \define
(V_{j}^{n}(a_{j}^{k_{j}},m_{j}^{l_{j}}):a_{j}^{k_{j}} \in \setV^{k_{j}})$ and $C(m^{l})\define
(V^{n}(a^{k},m^{l}):a^{k} \in \setV^{k})$ denote random  cosets in
$\Lambda_{O_{j}}/\Lambda_{I_{j}}:j=1,2$ and $\Lambda_{O}/\Lambda_{I}$ corresponding to
message $m_{j}^{l_{j}}:j=1,2$ and $m^{l}$ respectively.

Our next goal is to derive an upper bound on the probability of error. Towards this end,
we begin with a characterization of related events. Let 
\begin{eqnarray}
 \label{Eqn:StateSequencesNotTypical}
\epsilon_{1j} &\define& \{
S_{j}^{n}
\notin T_{\frac{\eta_{4}(\eta)}{2}}(S_{j}) \},~~~~ \epsilon_{1} \define \left\{ \boldS^{n} \notin T_{\frac{\eta_{4}(\eta)}{2}}(\boldS) \right\}\nonumber\\
\label{Eqn:NoCodewordIsJointlyTypicalWithStateSequence}
\epsilon_{2j} &\define& \{
\phi_{j}(S_{j}^{n},M_{j}^{l_{j}}) =0
\}, \mbox{ where } \phi_{j}(s_{j}^{n},m_{j}^{l_{j}}) \define
\sum_{a_{j}^{k_{j}} \in
\setV^{k_{j}}} 1_{\left\{\left(V_{j}^{n}(a_{j}^{k_{j}},m_{j}^{l_{j}}), s_{j}^{n}\right) \in
T_{\eta_{4}(\eta)}(V_{j},S_{j})\right\}}\nonumber\\
\label{Eqn:LegitimateCodewordIsNotTypicalWithReceivedVector}
\epsilon_{4}&\define& \underset{a^{k} \in \setV^{k}}{\bigcup}\left\{ (V^{n}(a^{k},M^{l}),Y^{n}) \in
T_{\eta_{5}(\eta)}(V_{1}\oplus V_{2},Y) \right\}\nonumber\\
\label{Eqn:IllegitimateCodewordsAreTypicalWithReceivedVector}
\epsilon_{5}&\define&\underset{{\hatm}^{l} \neq M^{l}}{\bigcup}\underset{a^{k} \in
\setV^{k}}{\bigcup}\left\{\left(V^{n}(a^{k},\hatm^{l}), Y^{n}\right) \in
T_{\eta_{5}(\eta)}(p_{V_{1} \oplus V_{2},Y})\right\}.\nonumber
\end{eqnarray}
Note that $\epsilon_{1} \cup \epsilon_{21} \cup \epsilon_{22} \cup \epsilon_{4}^{c} \cup
\epsilon_{5}$ contains the error event and hence $P(\epsilon_{1})+P(\epsilon_{11}^{c} \cap
\epsilon_{21})+P(\epsilon_{12}^{c} \cap \epsilon_{22})+P((\epsilon_{1} \cup \epsilon_{21}
\cup \epsilon_{22})^{c}\cap \epsilon_{4}^{c} )+P(\epsilon_{5})$ is an upper bound on the
probability of error. In the sequel, we provide an upper bound on each of the above terms.

Lemma \ref{Lem:ATypicalSetIsHighlyProbable} guarantees the existence of $N_{2}(\eta) \in
\naturals$ such that $P(\epsilon_{1}) \leq \frac{\eta}{8}$ for all $n \geq N_{1}(\eta)$.
Lemma \ref{Lem:WhenIsACosetCodeAGoodSourceCode?}(3) in appendix
\ref{Sec:AnUpperBoundOnProbabilityofEpsilon1ComplementIntersectionEpsilon2} implies
the existence of $N_{3}(\eta) \in \naturals$ such that for all $n\geq N_{3}(\eta)$
\begin{eqnarray}
\label{Eqn:MACDSTxUpperBoundingProbabilityOfEncodingError}
P(\epsilon_{1j}^{c}\cap \epsilon_{2j}) \leq \exp \left\{ -n\log \pi \left( \frac{k_{j}}{n} - \left( 1-\frac{H(V_{j}|S_{j})}{\log \pi}+\frac{3\eta_{4}(\eta)}{2\log \pi}  \right) \right) \right\}.\nonumber
\end{eqnarray}
Substituting the lower bound in (\ref{Eqn:MACDSTxBoundsOnBinningRate}) for $\frac{k_{j}}{n}$, we obtain
\begin{equation}
\label{Eqn:MACDSTxUpperBoundingProbabilityOfEncodingErrorUsingLowerBoundOnKj}
P(\epsilon_{1j}^{c}\cap \epsilon_{2j})\leq \exp \left\{ -n  \left(\eta_{1}(\eta)- \frac{3\eta_{4}(\eta)}{2} \right) \right\}.
\end{equation}
for all $n \geq \max \left\{ N_{1}(\eta),N_{3}(\eta) \right\}$. We now derive an upper bound on $P((\epsilon_{1} \cup \epsilon_{21} \cup \epsilon_{22})^{c}\cap \epsilon_{4}^{c})$. The encoding rule ensures $(\epsilon_{1} \cup \epsilon_{21} \cup \epsilon_{22})^{c} \subseteq (\epsilon_{1} \cup \epsilon_{2})^{c}$, where
\begin{equation}
 \label{Eqn:EventAtleastOneEncoderMakesAMistake}
\epsilon_{2} = \bigcup_{j=1}^{2} \left\{ \left( S_{j}^{n},V_{j}^{n}(A_{j}^{k_{j}},M_{j}^{l_{j}}) \right) \notin T_{\eta_{4}(\eta)}(S_{j},V_{j})  \right\},\nonumber
\end{equation}
and $V_{j}^{n}(A_{j}^{k_{j}},M_{j}^{l_{j}})$ denotes codeword in $L_{j}(M_{j}^{l_{j}},S_{j}^{n})$ chosen by encoder $j$. Our first step is to provide an upper bound on $P((\epsilon_{1} \cup \epsilon_{2})^{c} \cap \epsilon_{3})$ for sufficiently large $n$, where
\begin{equation}
\label{Eqn:EventStateAndVCodewordsAreNOTJointlyTypical}
\epsilon_{3} = \left\{ \left( S_{j}^{n},V_{j}^{n}(A_{j}^{k_{j}},M_{j}^{l_{j}}):j=1,2 \right) \notin T_{\frac{\eta_{5}(\eta)}{2}}(S_{1},V_{1},S_{2},V_{2})  \right\}.\nonumber
\end{equation}
In the second step, we employ the result of conditional frequency typicality to
provide an upper bound on $P((\epsilon_{1} \cup \epsilon_{2} \cup \epsilon_{3})^{c}\cap
\epsilon_{4}^{c})$.

As an astute reader might have guessed, the proof of first step will employ the Markov chain $V_{1}-S_{1}-S_{2}-V_{2}$. The proof is non-trivial because of statistical dependence of the codebooks. We begin with the definition
\begin{equation}
 \label{Eqn:PairwiseTypicalButQuadrapleNotTypical}
\Theta(\bolds^{n}) \define \left\{ \boldv^{n} \in \setV^{n} : (s_{j}^{n},v_{j}^{n}) \in T_{\eta_{4}(\eta)}(S_{j},V_{j}):j=1,2, (\bolds^{n},\boldv^{n}) \notin T_{\frac{\eta_{5}(\eta)}{2}}(\boldS,\boldV) \right\}\nonumber
\end{equation}
for any $\bolds^{n} \in \setS^{n}$. Observe that,
\begin{eqnarray}
 \label{Eqn:}
P((\epsilon_{1} \cup \epsilon_{2})^{c} \cap \epsilon_{3}) &=& \underset{\bolds^{n} \in T_{\frac{\eta_{4}(\eta)}{2}}(\boldS)}{\sum} \underset{\boldv^{n} \in \Theta(\bolds^{n})}{\sum} P(\boldS^{n} = \bolds^{n}, V_{j}^{n}(A_{j}^{k_{j}},M_{j}^{l_{j}})=v_{j}^{n}:j=1,2)
\nonumber\\
&=& \underset{\bolds^{n} \in T_{\frac{\eta_{4}(\eta)}{2}}(\boldS)}{\sum} \underset{\boldv^{n} \in \Theta(\bolds^{n})}{\sum} P\left(\underset{a_{1}^{k_{1}} \in \setV_{1}^{k_{1}}}{\bigcup}\underset{{a_{2}^{k_{2}} \in \setV_{2}^{k_{2}}}}{\bigcup}\left\{ \boldS^{n} = \bolds^{n},\substack{ V_{j}^{n}(A_{j}^{k_{j}},M_{j}^{l_{j}})=v_{j}^{n}:j=1,2,\\V_{j}^{n}(a_{j}^{k_{j}},M_{j}^{l_{j}})=v_{j}^{n}:j=1,2}\right\}\right)
\nonumber\\
&\leq& \underset{\bolds^{n} \in T_{\frac{\eta_{4}(\eta)}{2}}(\boldS)}{\sum}
\underset{\boldv^{n} \in \Theta(\bolds^{n})}{\sum} \underset{a_{1}^{k_{1}} \in
\setV_{1}^{k_{1}}}{\sum}\underset{{a_{2}^{k_{2}} \in
\setV_{2}^{k_{2}}}}{\sum}P\left(\left\{ \boldS^{n} =
\bolds^{n},\substack{V_{1}^{n}(a_{1}^{k_{1}},M_{1}^{l_{1}})=v_{1}^{n},\\V_{2}^{n}(a_{2}^{
k_ { 2}},M_{2}^{l_{2}})=v_{2}^{n}}\right\} \right)
\nonumber\\
\label{Eqn:MarkovLemmaErrorEventIndependenceOfVAndS}
&=&\!\!\!\!\!\! \underset{\bolds^{n} \in T_{\frac{\eta_{4}(\eta)}{2}}(\boldS)}{\sum}
\underset{\boldv^{n} \in \Theta(\bolds^{n})}{\sum} \underset{a_{1}^{k_{1}} \in
\setV_{1}^{k_{1}}}{\sum}\underset{{a_{2}^{k_{2}} \in \setV_{2}^{k_{2}}}}{\sum}P\left(
\boldS^{n} =
\bolds^{n}\right)P\left(\substack{V_{1}^{n}(a_{1}^{k_{1}},M_{1}^{l_{1}})=v_{1}^{n},\\V_{2}
^{n}(a_{2}^{
k_ { 2}},M_{2}^{l_{2}})=v_{2}^{n}}\right)\\
\label{Eqn:MarkovLemmaErrorEventBeforeUsingSanovlemma}
&=& \underset{\bolds^{n} \in T_{\frac{\eta_{4}(\eta)}{2}}(\boldS)}{\sum} \underset{\boldv^{n} \in \Theta(\bolds^{n})}{\sum} P\left( \boldS^{n} = \bolds^{n}\right)\frac{1}{\pi^{n-k_{1}}}\frac{1}{\pi^{n-k_{2}}}
\end{eqnarray}
where $V_{j}^{n}(A_{j}^{k_{j}},M_{j}^{l_{j}})$ is defined as the random codeword chosen
by the encoder, (\ref{Eqn:MarkovLemmaErrorEventIndependenceOfVAndS}) follows from
independence of random variables $(M^{l},G_{I},G_{O/I},B_{1}^{n},B_{2}^{n})$ that
characterize $V_{j}^{n}(a_{j}^{k_{j}},M_{j}^{l_{j}})$ and $\boldS^{n}$.
We now employ the upper bound on $k_{j}$ in (\ref{Eqn:MACDSTxBoundsOnBinningRate}) to
substitute for $\frac{1}{\pi^{n-k_{j}}}$. For $n \geq N_{1}(\eta)$, we have $k_{j} \leq n
- \frac{H(V_{j}|S_{j})}{\log \pi}+\frac{2\eta_{1}(\eta)}{\log \pi}$ and hence
\begin{equation}
\label{Eqn:UsingUpperBoundOnkjInMarkovLemma}
\frac{1}{\pi^{n-k_{j}}} \leq \exp \left\{ -n \left( H(V_{j}|S_{j})-2\eta_{1}(\eta) \right)\right\}.
\end{equation} 
Furthermore, by Lemma \ref{Lem:BoundsOnProbabilityOfTypicalSequence}, for every
$\bolds^{n} \in T_{\frac{\eta_{4}(\eta)}{2}}(\boldS)$ and $\boldv^{n} \in
\Theta(\bolds^{n})$, 
\begin{equation}
\label{Eqn:LowerBoundOnCoditionalProbabilityOfTypicalSeqeuneces}
\exp \left\{ -n \left( H(V_{j}|S_{j})-2\eta_{4}(\eta)  \right)  \right\} \leq p_{V_{j}^{n}|S_{j}^{n}}(v_{j}^{n}|s_{j}^{n}) =  p_{V_{j}^{n}|\boldS^{n}}(v_{j}^{n}|\bolds^{n}) = p_{V_{j}^{n}|\boldS^{n}V_{\msout{j}}^{n}}(v_{j}^{n}|\bolds^{n},v_{\msout{j}}^{n}),\end{equation}
where the last equalities is a consequence of Markov chain $V_{1}-S_{1}-S_{2}-V_{2}$.
Substituting the upper bounds in (\ref{Eqn:UsingUpperBoundOnkjInMarkovLemma}) and
(\ref{Eqn:LowerBoundOnCoditionalProbabilityOfTypicalSeqeuneces}) for
$\frac{1}{\pi^{n-k_{j}}}$ in (\ref{Eqn:MarkovLemmaErrorEventBeforeUsingSanovlemma}), we
obtain
\begin{eqnarray}
 P((\epsilon_{1} \cup \epsilon_{2})^{c} \cap \epsilon_{3}) &\leq& \exp \left\{ n(4\eta_{1}(\eta)+4\eta_{4}(\eta)) \right\}\cdot \underset{\bolds^{n} \in T_{\frac{\eta_{4}(\eta)}{2}}(\boldS)}{\sum} \underset{\boldv^{n} \in \Theta(\bolds^{n})}{\sum} p_{\boldS^{n} \boldV^{n}}(\bolds^{n},\boldv^{n})\nonumber\\
\label{Eqn:UsingSanovLemmaInMarkovLemma}
&\leq & \exp \left\{ n(4\eta_{1}(\eta)+4\eta_{4}(\eta)) \right\} \cdot \sum_{(\bolds^{n},\boldv^{n}) \notin T_{\eta_{5}(\eta)}(\boldS,\boldV)}p_{\boldS^{n} \boldV^{n}}(\bolds^{n},\boldv^{n})
\end{eqnarray}
for all $n \geq N_{1}(\eta)$. We now employ the exponential upper bound provided in Lemma
\ref{Lem:ATypicalSetIsHighlyProbable}. In particular, Lemma
\ref{Lem:ATypicalSetIsHighlyProbable} guarantees the existence of $N_{4}(\eta) \in
\naturals$
such that for every $n \geq N_{4}(\eta)$,
\begin{equation}
 \label{Eqn:ApplyingSanovToProbabilityInQuestion}
\sum_{\substack{(\bolds^{n},\boldv^{n}) \in \\T_{\eta_{5}(\eta)}(\boldS,\boldV)}}\!\!\!p_{\boldS^{n} \boldV^{n}}(\bolds^{n},\boldv^{n}) \leq \exp \left\{ -n\lambda \eta_{5}^{2}(\eta) \right\} \mbox{, where } \lambda \define \frac{\min_{(\bolds,\boldv) \in \setS \times \setV} \left\{ p^{2}_{\boldS\boldV}(\bolds,\boldv): p_{\boldS\boldV}(\bolds,\boldv) > 0\right\}}{\left(\log |\setS||\setV|\right)^{2}}.
\end{equation}
Substituting (\ref{Eqn:ApplyingSanovToProbabilityInQuestion}) in (\ref{Eqn:UsingSanovLemmaInMarkovLemma}), we conclude
\begin{equation}
\label{Eqn:ConclusionOfStep1OfMarkovlemma}
 P((\epsilon_{1} \cup \epsilon_{2})^{c} \cap \epsilon_{3}) \leq \exp \left\{  -n\left( \lambda\eta_{5}^{2}(\eta)-4\eta_{1}(\eta)-4\eta_{4}(\eta)\right)\right\}
\end{equation}
for every $n \geq \max\left\{ N_{1}(\eta),N_{4}(\eta) \right\}$.
This gets us to the second step. We begin with two observations. Firstly, note that $V(a_{1}^{k_{1}}0^{k_{+}}\oplus a_{2}^{k_{2}},m_{1}^{l_{1}}m_{2}^{l_{2}}) = V_{1}(a_{1}^{k_{1}},m_{1}^{l_{1}}) \oplus V_{2}(a_{2}^{k_{2}},m_{2}^{l_{2}})$. This follows from the definition of the codewords involved. Secondly, 
\begin{eqnarray}
&P\left(\substack{V(A_{1}^{k_{1}}0^{k_{+}}\oplus A_{2}^{k_{2}},M_{1}^{l_{1}}M_{2}^{l_{2}})=v^{n},\\X_{j}^{n}(M_{j}^{l_{j}},S_{j}^{n})=x_{j}^{n}:j=1,2,Y^{n}=y^{n}}\middle|\substack{V_{j}^{n}(A_{j}^{k_{j}},M_{j}^{l_{j}})=v_{j}^{n},\\:j=1,2,\boldS^{n}=\bolds^{n}}\right) = P\left(\substack{V_{1}(A_{1}^{k_{1}},M_{1}^{l_{1}})\oplus V_{2}(A_{2}^{k_{2}},M_{2}^{l_{2}})=v^{n},\\X_{j}^{n}(M_{j}^{l_{j}},S_{j}^{n})=x_{j}^{n}:j=1,2,Y^{n}=y^{n}}\middle|\substack{V_{j}^{n}(A_{j}^{k_{j}},M_{j}^{l_{j}})=v_{j}^{n},\\:j=1,2,\boldS^{n}=\bolds^{n}}\right)\nonumber\\
\label{Eqn:ChannelInputAndOutputAreGeneratedAccordingToTheCorrectCondDist}
\lefteqn{= \displaystyle \prod_{t=1}^{n}\left[p_{V_{1} \oplus V_{2}|V_{1}V_{2}}(v_{t}|v_{1t},v_{2t})\left(\prod_{j=1}^{2}p_{X_{j}|V_{j}S_{j}}(x_{jt}|v_{jt},s_{jt})\right)W_{Y|\boldX \boldS}(y_{t}|\boldx_{t},\bolds_{t}) \right]}\\
\label{Eqn:ChannelInputAndOutputAreGeneratedAccordingToTheCorrectCondDistStep2}
\lefteqn{=\displaystyle\prod_{t=1}^{n}P(V_{1} \oplus V_{2} = v_{t}, \boldX = \boldx_{t}, Y_{t}=y_{t}|\boldS_{t}=\bolds_{t},\boldV_{t}=\boldv_{t}),}
\end{eqnarray}
where we have employed 1) encoding rule and Markov chains $\boldU - (\boldX,\boldS)- Y$ in arriving at (\ref{Eqn:ChannelInputAndOutputAreGeneratedAccordingToTheCorrectCondDist}) and 2) the identity $p_{X_{j}|\boldS\boldU X_{\msout{j}}}=p_{X_{j}|\boldS
\boldU}=p_{X_{j}|S_{j}U_{j}}$ for any distinct elements $j,\msout{j} \in \{ 1,2 \}$ in arriving at (\ref{Eqn:ChannelInputAndOutputAreGeneratedAccordingToTheCorrectCondDistStep2}). Since
\begin{eqnarray}
 \label{Eqn:UsingConditionalFrequencyTypicalityToUpperBoundFirstErrorEventAtDecoder}
&P((\epsilon_{1} \cup \epsilon_{2} \cup \epsilon_{3})^{c} \cap \epsilon_{4}^{c}) \leq P\left((\epsilon_{1} \cup \epsilon_{2} \cup \epsilon_{3})^{c} \cap \left\{ (V(\substack{A_{1}^{k_{1}}0^{k_{+}}\\\oplus A_{2}^{k_{2}}},M_{1}^{l_{1}}M_{2}^{l_{2}}),Y^{n}) \notin T_{\eta_{5}(\eta)}(V_{1}\oplus V_{2},Y) \right\}\right)\nonumber\\
\label{Eqn:SecondStepOfMarkovlemmaEventReformulated}
&\leq P\left((S_{j}^{n},V_{j}^{n}(A_{j}^{k_{j}},M_{j}^{l_{j}}):j=1,2) \in T_{\frac{\eta_{5}(\eta)}{2}}(\boldS,\boldV),  (V(\substack{A_{1}^{k_{1}}0^{k_{+}}\\\oplus A_{2}^{k_{2}}},M_{1}^{l_{1}}M_{2}^{l_{2}}),Y^{n}) \notin T_{\eta_{5}(\eta)}(V_{1}\oplus V_{2},Y) \right),\nonumber
\end{eqnarray}
and the above two observations imply that $(V(A_{1}^{k_{1}}0^{k_{+}}\oplus
A_{2}^{k_{2}},M_{1}^{l_{1}}M_{2}^{l_{2}}),\boldX^{n},Y^{n})$ is distributed according to
$\prod_{t=1}^{n}P(V_{1} \oplus V_{2} = v_{t}, \boldX = \boldx_{t},
Y_{t}=y_{t}|\boldS_{t}=\bolds_{t},\boldV_{t}=\boldv_{t})$. Lemma
\ref{Lem:ConditionalTypicalSetOccursWithHighProbability} guarantees the existence of
$N_{5}(\eta) \in \naturals$, such that for all $n \geq N_{5}(\eta)$, the term on the right
hand side of (\ref{Eqn:SecondStepOfMarkovlemmaEventReformulated}) is bounded from above by
$\frac{\eta}{8}$. Therefore, for all $n \geq N_{5}(\eta)$
\begin{equation}
 \label{Eqn:FinalStepOfSecondStepOfMarkovlemma}
P((\epsilon_{1} \cup \epsilon_{2} \cup \epsilon_{3})^{c} \cap \epsilon_{4}^{c}) \leq \frac{\eta}{8}.
\end{equation}

It remains to provide an upper bound on $P((\epsilon_{1} \cup \epsilon_{21} \cup
\epsilon_{22}\cup \epsilon_{4}^{c})^{c} \cap \epsilon_{5})$. In appendix
\ref{Sec:AnUpperBoundOnProbabilityofEpsilon4ForMAC-DSTx}, we prove the existence of
$N_{6}(\eta) \in \naturals$ such that $P(\epsilon_{5}) \leq \exp \left\{ -n \left(
3\eta_{5}(\eta)-\eta_{2}(\eta)-\eta_{3}(\eta)\right) \right\} \text{ for all }n \geq
\max\left\{ N_{1}(\eta),N_{6}(\eta) \right\}$. The informed reader will recognize that
deriving an upper bound on $P(\epsilon_{5})$ will involve proving statistical independence
of the pair $(C_{j}(M_{j}^{l_{j}}):j=1,2)$ of cosets corresponding to the legitimate
message pair $M_{j}^{l}$
and any codeword $V^{n}(\hata^{k},\hatm^{l})$ corresponding to a competing message pair
$\hatm^{l} \neq M^{l}$. This is considerably simple for a coding technique based on
classical unstructured codes wherein codebooks and codewords in every codebook are
independent. The coding technique proposed herein involves correlated codebooks - the
first $k_{1}$ rows of $G_{I_{j}}:j=1,2$ are identical\footnote{If
$H(V_{1}|S_{1})=H(V_{2}|S_{2})$, users 1 and 2 share the same generator matrix $G_{I}$.
Indeed, channel codes of users' 1 and 2 are partitioned into cosets of the same linear
code.} - and codewords in each codebook are correlated.

To conclude, we put together the upper bounds derived on the probability of events that comprise the error event. For $n \geq N_{2}(\eta)$, $P(\epsilon_{1}) \leq \frac{\eta}{8}$. In (\ref{Eqn:MACDSTxUpperBoundingProbabilityOfEncodingErrorUsingLowerBoundOnKj}), we proved $P(\epsilon_{1j}^{c}\cap \epsilon_{2j})\leq \exp \left\{ -n  \left(\eta_{1}(\eta)- \frac{3\eta_{4}(\eta)}{2} \right) \right\}$ for all $n \geq N_{3}(\eta)$. Combining (\ref{Eqn:ConclusionOfStep1OfMarkovlemma}) and (\ref{Eqn:FinalStepOfSecondStepOfMarkovlemma}), we have \[P((\epsilon_{1} \cup \epsilon_{2})^{c} \cap \epsilon_{4}^{c}) \leq \exp \left\{  -n\left( \lambda\eta_{5}^{2}(\eta)-4\eta_{1}(\eta)-4\eta_{4}(\eta)\right)\right\} + \frac{\eta}{8} \]
for all $n \geq \max\left\{ N_{1}(\eta),N_{4}(\eta),N_{5}(\eta) \right\}$. And finally
$P(\epsilon_{5}) \leq \exp \left\{ -n \left(
\eta_{2}(\eta)+\eta_{3}(\eta)-3\eta_{5}(\eta)\right) \right\}$ for all $n \geq
\max\left\{ N_{1}(\eta),N_{6}(\eta) \right\}$ follows from
(\ref{Eqn:FinalUpperBoundOnProbabilityOfEpsilon5}). By choosing
\begin{equation}
 \label{Eqn:ChoiseOfEtaParameters}
\eta_{2}(\eta) = \eta_{3}(\eta) = \frac{\eta}{16}, \eta_{5}(\eta) = \frac{\eta}{48}, \eta_{1}(\eta) = \min \left\{ \frac{\eta}{16},\frac{\lambda\eta_{5}^{2}(\eta)}{10} \right\} \mbox{ and }\eta_{4}(\eta) = \frac{\eta_{1}(\eta)}{4}
\end{equation}
it can be verified that for $n \geq \overline{N}(\eta)\define\max  \left\{ N_{i}(\eta):i\in [6] \right\}$,
\begin{itemize}
 \item $2\eta_{1}(\eta)+2\eta_{3}(\eta) < \frac{\eta}{2}$ and thus $\frac{l_{2} \log \pi}{n} \geq R_{2}-\frac{\eta}{2}$ from (\ref{Eqn:MACDSTxLowerBoundOnUser2Rate}),
\item $\eta_{2}(\eta) < \frac{\eta}{2}$ and thus $\frac{l_{1}\log \pi}{n} > R_{1}-\frac{\eta}{2}$ from (\ref{Eqn:MACDSTxUpperAndLowerBoundsOnR1}),
\item $\eta_{1}(\eta)- \frac{3\eta_{4}(\eta)}{2}=\frac{5 \eta_{1}(\eta)}{8}$ and thus $P(\epsilon_{1j}^{c}\cap \epsilon_{2j})\leq \exp \left\{ -n  \left(\frac{5 \eta_{1}(\eta)}{8} \right) \right\}$,
\item $\lambda\eta_{5}^{2}(\eta)-4\eta_{1}(\eta)-4\eta_{4}(\eta) \geq \frac{\lambda\eta_{5}^{2}(\eta)}{2}$ and thus $P((\epsilon_{1} \cup \epsilon_{2})^{c} \cap \epsilon_{4}^{c}) \leq \exp \left\{  -n\left( \frac{\lambda\eta_{5}^{2}(\eta)}{2}\right)\right\} + \frac{\eta}{8}$, and
\item $\eta_{2}(\eta)+\eta_{3}(\eta)-3\eta_{5}(\eta) = \frac{\eta}{16}$ and therefore $P(\epsilon_{5}) \leq \exp \left\{ -n \left(\frac{\eta}{16}\right) \right\}$.
\end{itemize}
For $n\geq \overline{N}(\eta)$, $P(\epsilon_{1})+P(\epsilon_{11}^{c} \cap \epsilon_{21})+P(\epsilon_{12}^{c} \cap \epsilon_{22})+P((\epsilon_{1} \cup \epsilon_{21} \cup \epsilon_{22})^{c}\cap \epsilon_{4}^{c} )+P(\epsilon_{5}) \leq \frac{\eta}{4}+3\exp\left\{ -n(\frac{5\eta_{1}}{8}) \right\}$. Thus for $n \geq N(\eta)\define \max\left\{\overline{N}(\eta),\frac{1}{\eta_{1}(\eta)}\log \lceil \frac{4}{\eta}\rceil  \right\}$, the error event has probability at most $\eta$.
\end{proof}
We conclude this section with two remarks.
\begin{remark}
 \label{Rem:TheoremAchievesCapacityForBDDMAC}
 For BDD-MAC described in section
\ref{SubSec:AnAchievableRateRegionForArbitraryMAC-DSTxUsingNestedCosetCodes},
$\beta_{f}(\boldTau) = \mathbb{C}(\boldTau)$. Indeed, the test channel
$p_{\boldV\boldS\boldX Y} \in \mathbb{D}_{f}(\boldTau)$ defined as
$p_{\boldV\boldS\boldX}=\prod_{j=1}^{2}p_{V_{j}S_{j}X_{j}}$ where $V_{j}$ takes values
over $\mathcal{V}_{j}=\left\{ 0,1 \right\}$ with
 \begin{equation}
  \label{Eqn:TestChannelThatAchievesPhilosofZamir}
  p_{V_{j},X_{j}|S_{j}}(x_{j}\oplus_{2}s_{j},x_{j}|s_{j}) = \left\{ \begin{array}{lr} 1-\tau&\mbox{if }x_{j}=0\\\tau &\mbox{otherwise} \end{array} \right. \nonumber
 \end{equation}
for each $j=1,2$ and $s_{j} \in \{0,1\}$ achieves $\mathbb{C}(\boldTau)=\left\{ (R_{1},R_{2}):R_{1}+R_{2} \leq h_{b}(\tau) \right\}$.
\end{remark}

We have thus presented a coding technique based on decoding the sum of codewords chosen by
the encoders and analyzed the same to derive an achievable rate region for an arbitrary
MAC-DSTx. One might
attempt a generalization of PZ-technique along the lines of modulo lattice transformation
proposed by Haim, Kochman and Erez \cite{201207ISIT_HaiKocEre}. The rate region proposed
herein subsumes that achievable through modulo-lattice transformation using test channels
identified through the virtual channel in a natural way.
\subsection{Examples}
\label{SubSec:Examples}

A key element of the coding framework proposed herein lies in characterizing achievable
rate regions for arbitrary test channels, i.e., test channels that are not restricted to
be uniform or additive in nature using structured codes.

A few remarks on our study of the following examples are in order. The examples needing to
be non-additive lends it considerably hard to
provide analytical upper bounds for the rate region achievable using unstructured
codes.\footnote{We recognize that the analytical upper bound derived in
\cite{200906TIT_PhiZam} is a key element of the findings therein.} We therefore resort to
computation. It can be noted that the problem of computing the sum rate bound achievable
using unstructured codes is a non-convex optimization problem. The only approach is
direct enumeration, i.e., sampling the probability matrix of the auxiliary
random variables.\footnote{This holds even for the case of multiple access without states
for which a computable characterization of the capacity region is known.} Sampling the
probability matrix with any reasonable step size beyond the auxiliary alphabets of size
$2$ is infeasible with currently available computation resources. The sum rate bound for
the unstructured coding technique projected below is therefore obtained through 
computation involving binary auxiliary alphabet sets followed by convexification
(time sharing between different costs). The
resulting space of probability distributions that respect the cost constraints is sampled
with a step size of $0.015$ in each dimension. The resulting bound on the sum rate
achievable using unstructured codes (without time sharing) is marked with blue crosses
(denoted $\alpha$ in the legend) in the plots. The resulting upper bound is obtained as
an upper convex envelope. Similarly, sum rate achievable using nested coset codes is
marked with red circles (denoted $\beta$ in the legend) in the plots.

\begin{figure}[h]
\includegraphics[height=2.2in,width=4.4in]{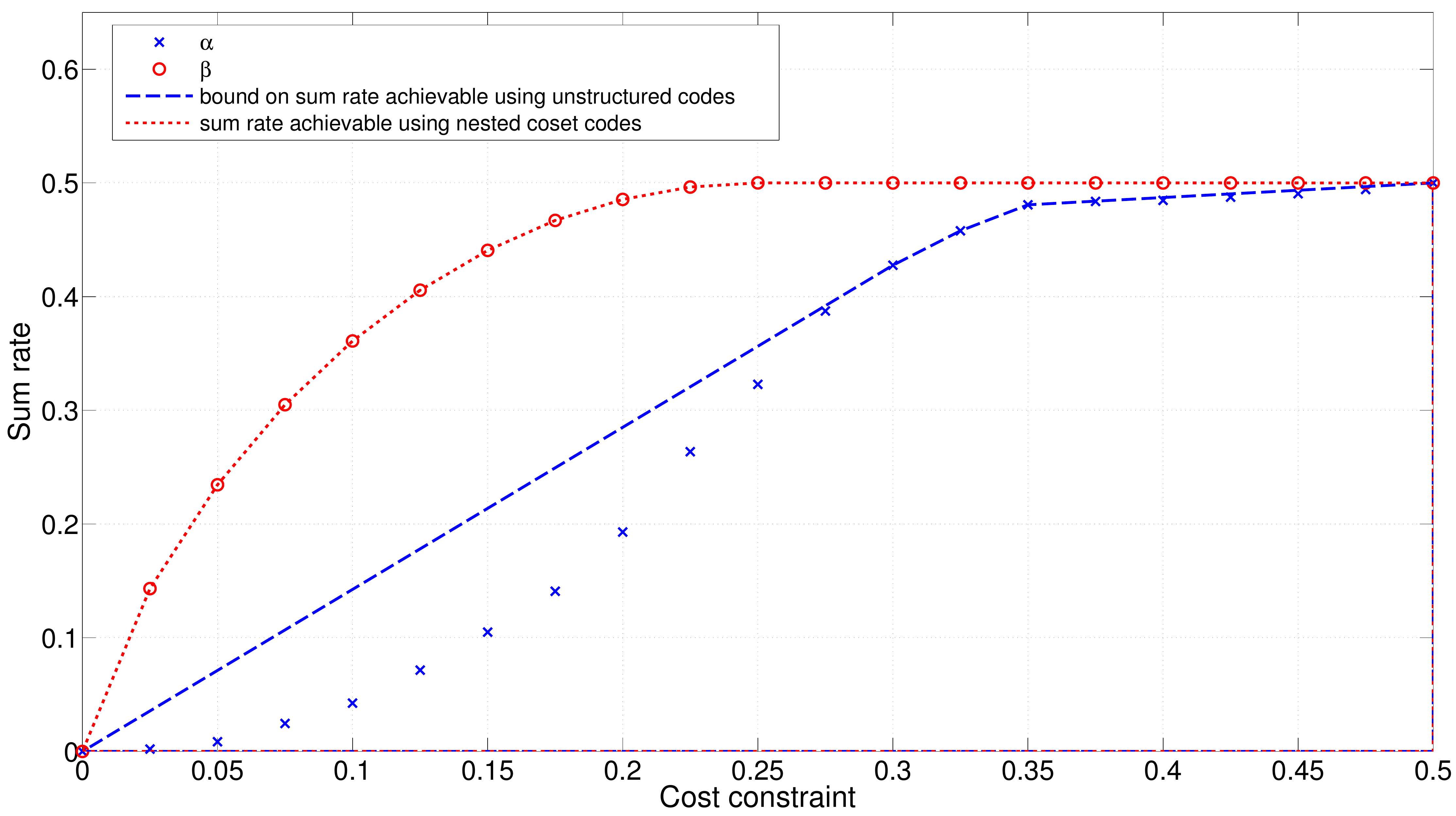}
\caption{Bounds on sum rate for example \ref{Ex:DoublyDirtyMACReplacedWithLogicalOr}}
\label{Fig:PlotsForDoublyDirtyMACReplacedWithLogicalOr}
\end{figure}
For examples \ref{Ex:DoublyDirtyMACReplacedWithLogicalOr} and \ref{Ex:Example5}, we
assume the alphabet sets to be binary $\setS_{j}=\setX_{j}=\{ 0,1\}$, $j=1,2$, (ii)
uniform
and independent states, i.e.,
$W_{\boldS}(\bolds) = \frac{1}{4}$ for all $\bolds \in \setS$, (iii) a
Hamming cost function $\kappa_{j}(1,s_{j})=1$ and
$\kappa_{j}(0,s_{j})=0$ for any $s_{j} \in \setS_{j}$, $j=1,2$.

\begin{example}
\label{Ex:DoublyDirtyMACReplacedWithLogicalOr}
Let $Y = (X_{1}\vee S_{1})\oplus (X_{2}\vee S_{2})$, where $\vee$ denotes logical OR
operator. Having studied the BDD-MAC it is natural to conjecture that the test channel
that
optimizes the sum rate achievable using linear codes to be
$p_{U_{j}X_{j}|S_{j}}(0,0|0)=1-2\tau,
p_{U_{j}X_{j}|S_{j}}(1,1|0)=2\tau,p_{U_{j}X_{j}|S_{j}}(1,0|1)=1,$ for $j=1,2$ when the
cost constraint $\tau \in [0,\frac{1}{4}]$. Indeed, our numerical computation asserts
this. In other words, the sum rate
achievable using linear codes for a cost $\tau \in
(0,\frac{1}{4})$ is $\frac{h_{b}(2\tau)}{2}$ and $0.5$ for $\tau \in [0.25,0.5]$. The sum
rate achievable using unstructured codes and nested
coset codes are plotted in figure
\ref{Fig:PlotsForDoublyDirtyMACReplacedWithLogicalOr}. We
highlight significant gains achievable using nested coset codes.

A preliminary look at this channel may lead the reader to conclude that
PZ-technique appropriately modified can achieve the same sum rate as that achievable
using nested coset codes, since the above test channel is additive, i.e.,
$U_{j}=S_{j}\oplus X_{j}$ for $j=1,2$ and $Y=U_{1} \oplus U_{2}$.However, a careful analysis will reveal the significance of the coding
framework proposed herein. The induced pmf on $U_{j}$, $p_{U_{j}}(1)=\frac{1}{2}+2\tau
$ for $\tau \in (0,\frac{1}{4})$ is \textit{not} uniform, and the PZ-technique of
choosing a codeword in the indexed bin with an average Hamming distance of $\tau$ does
not yield the sum rate guaranteed by nested coset codes. Nesting of codes enables
achieving
non-uniform distributions that are necessary as exemplified herein.
\end{example}

\begin{example}
\label{Ex:Example5}
The channel transition matrix is given in table
\ref{Table:ChannelTransitionMatrixExample5}. 1) An upper bound on sum rate achievable
using unstructured codes and 2) sum rate achievable using nested coset codes are
plotted in figure \ref{Fig:SumRatePlotsForIntractableChannel}. This channel is obtained by
randomly perturbing the BDD-MAC.\footnote{The reader is referred to \cite[Section VI.C]{201301arXivMACDSTx_PadPra} wherein we have presented results for a few more channels obtained by a random perturbation of the BDD-MAC.} In the
space of channel transition probability matrices, this channel is in a neighborhood of the
BDD-MAC. Since the rate regions are continuous functions over this space of channels, the
coding technique proposed herein outperforms unstructured coding technique in this
neighborhood. This example validates the same. As in the previous example, we note that
the optimizing distribution of the auxiliary random variables is non-uniform for
certain cost values. Furthermore, note that $\beta_{f}(\boldTau)$ does not contain
$\alpha(\boldTau)$ and therefore it helps to incorporate both unstructured and structured
coding techniques as will be studied in the following section.
\end{example}
\begin{figure}
\includegraphics[height=2.2in,width=4.4in]{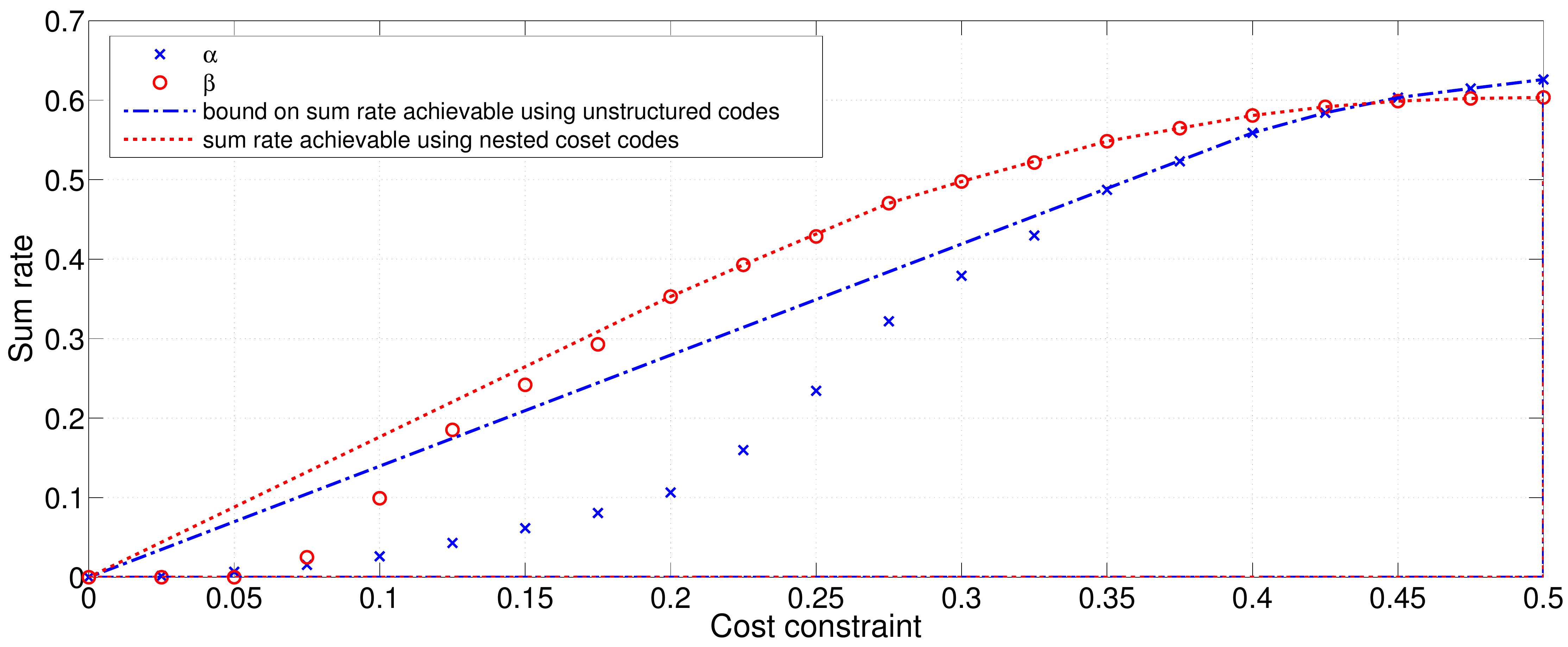}
\caption{Bounds on sum rate for example \ref{Ex:Example5}}
\label{Fig:SumRatePlotsForIntractableChannel}
\end{figure}
\begin{table} \begin{center}
\begin{tabular}{|c|c|c|c|c|c|c|c|c|c|} \hline
$\scriptstyle S_{2}X_{2}S_{1}X_{1}$& $\scriptscriptstyle W_{ Y| \boldS\boldX}(0|\cdot)$&
$\scriptstyle S_{2}X_{2}S_{1}X_{1}$ &$\scriptscriptstyle W_{ Y|
\boldS\boldX}(0|\cdot)$&$\scriptstyle S_{2}X_{2}S_{1}X_{1}$& $\scriptscriptstyle W_{ Y|
\boldS\boldX}(0|\cdot)$&
$\scriptstyle S_{2}X_{2}S_{1}X_{1}$ &$\scriptscriptstyle W_{ Y| \boldS\boldX}(0|\cdot)$\\
\hline\hline
0000 & 0.92 & 1000 & 0.07  &0100 & 0.10 & 1100 &  0.88  \\
\hline
0001 & 0.08 & 1001 &  0.92 &0101 & 0.92 & 1101 &   0.08\\
\hline
0010 & 0.06 & 1010 & 0.96 &0110 & 0.95 & 1110 &  0.11\\
\hline
0011 & 0.94 & 1011 &  0.10&0111 & 0.06 & 1111 &  0.91\\
\hline
\end{tabular} \end{center}
\caption{Channel transition matrix Example \ref{Ex:Example5}}
\label{Table:ChannelTransitionMatrixExample5} 
\end{table} 

\begin{example}
 \label{Ex:OrExample}
 Consider the channel $Y=(S_{1} \oplus X_{1}) \vee (S_{2} \oplus X_{2})$. Observe
that the information available at the encoders is fused through a logical OR operation by
the channel. Moreover, $(U_{1},U_{2})- U_{1}\oplus_{3} U_{2}-U_{1}\vee U_{2}$ is a
Markov chain and hence, although channel input, state and output alphabets are binary, we
expect that for certain choice of auxiliary distributions, the sum rate achievable using
codes over $\mathcal{F}_{3}$ is larger than that achievable using unstructured codes.
Through an exhaustive search, we have identified such distributions, an example of which
is given in table \ref{Table:DistributionForWhichLinearCodeDoesBetter}.
\end{example}
\begin{table} \begin{center}
\begin{tabular}{|c|c|c|c|c|c|c|c|c|c|} \hline
$\scriptstyle U_{1}S_{1}X_{1}$& $\scriptscriptstyle p_{U_{1}S_{1}X_{1}}$&
$\scriptstyle  U_{1}S_{1}X_{1}$ &$\scriptscriptstyle p_{U_{1}S_{1}X_{1}}$&$\scriptstyle
U_{2}S_{2}X_{2}$& $\scriptscriptstyle p_{U_{2}S_{2}X_{2}}$&
$\scriptstyle U_{2}S_{2}X_{2}$ &$\scriptscriptstyle p_{U_{2}S_{2}X_{2}}$\\
\hline
\hline
000 & 0.1472 & 101 & 0.3528  &000 & 0.1472 & 101 & 0.3528  \\
\hline
011 & 0.50 &  &   &011 & 0.50 &  &   \\
\hline
\end{tabular} \end{center}
\caption{Test channel for example \ref{Ex:OrExample} for which nested coset code
over $\mathcal{F}_{3}$ performs
better than unstructured code}
\label{Table:DistributionForWhichLinearCodeDoesBetter} 
\end{table} 

For the distribution in table \ref{Table:DistributionForWhichLinearCodeDoesBetter}, the rate achievable using nested coset codes over
$\mathcal{F}_{3}$ is $0.0017$, while that
achievable using unstructured code is negative. For an appropriate choice of cost
function, the above might be the optimizing distribution for the
unstructured coding scheme thus resulting in larger sum rate using nested coset codes over
$\mathcal{F}_{3}$. We do not as of yet have a precise analytical characterization of such
a cost function\footnote{Such a characterization of cost function is available for
point-to-point channels with state available at both encoder and decoder \cite{CK-IT2011},
\cite{200305TIT_GasRimVet}, \cite{200607TIT_PraRam}.} and we are in pursuit of the
same. Nevertheless, the above lends credence to the use of nested coset codes for
arbitrary channels.

\section{Stage II: Combining unstructured and structured coding techniques}
\label{Sec:AUnifiedAchievableRateRegion}
In this section, we put together the techniques of unstructured and structured random
coding to derive a larger achievable rate region for a general MAC-DSTx. Our
approach is similar to that proposed by Ahlswede and Han
\cite[Section VI]{198305TIT_AhlHan} for the problem of reconstructing mod$-2$ sum
of distributed binary sources. We begin with a characterization of valid test channels.
\begin{definition}
 \label{Defn:CollectionOfDistributionsRequiredToUnifyUnstructuredAndStructuredTechniques}
 Let $\mathbb{D}_{sf}(\boldTau) \subseteq \mathbb{D}(\boldTau)$ be the
collection of distributions $p_{\boldU\boldV\boldS\boldX Y}$ on
$(\mathcal{U}\times\mathcal{V})^{2}\times \mathcal{S}\times\mathcal{X}\times\mathcal{Y}$
where
$\mathcal{U}$ is a finite set and $\mathcal{V}$ is a finite field. For
$p_{\boldU\boldV\boldS\boldX Y} \in \mathbb{D}_{sf}(\boldTau)$, let
$\mathcal{R}_{sf}(p_{\boldU \boldV\boldX \boldS Y})$ be defined as the set of rate pairs
$(R_{1},R_{2}) \in [0,\infty)^{2}$ that satisfy
\begin{eqnarray}
\begin{array}{l}R_{1} \leq I(U_{1};U_{2}Y)-I(U_{1};S_{1}) +
\min \left\{H(V_{1}|U_{1},S_{1}),H(V_{2}|U_{2},S_{2})\right\}-H(V_{1}\oplus
V_{2}|U_{1},U_{2},Y), \\
R_{2}\leq I(U_{2};U_{1}Y)-I(U_{2};S_{2}) +
\min \left\{
H(V_{1}|U_{1},S_{1}),H(V_{2}|U_{2},S_{2})\right\}-H(V_{1}\oplus
V_{2}|U_{1},U_{2},Y), \\
R_{1}+R_{2} \leq
I(U_{1}U_{2};Y)+I(U_{1};U_{2})-\sum_{j=1}^{2}I(U_{j};S_{j})+
\min \left\{
H(V_{1}|U_{1},S_{1}),H(V_{2}|U_{2},S_{2})\right\}
\\~~~~~~~~~~~~~~~~~~~~~~~~~~~-H(V_{1}\oplus V_{2}|U_{1},U_{2},Y),
\end{array}\nonumber\end{eqnarray} where $\oplus$ is addition in $\mathcal{V}$. Let
\begin{equation}
 \label{Eqn:LargestKnownAchievableRateForMACDSTxUsingStructuredAndUnstructuredCodes}
 \mathcal{R}_{sf}(\boldTau) \define \cocl \left( \underset{p_{\boldU\boldV\boldX\boldS Y}
\in
\mathbb{D}_{sf}(\boldTau)}{\bigcup} \mathcal{R}_{sf}(p_{\boldU \boldV\boldX \boldS Y}) 
\right)
\end{equation}
\end{definition}

\begin{thm}
 \label{Thm:AchievableRateRegionUsingNestedCosetCodesAndUnstructuredCodes}
$\mathcal{R}_{sf}(\boldTau) \subseteq \mathcal{C}(\boldTau)$.
\end{thm}
\begin{remark}
 \label{Rem:TheLargestKnownAchievableRateRegion}
$\alpha(\boldTau) \subsetneq \mathcal{R}_{sf}(\boldTau)$.
\end{remark}

\begin{proof}
 Achievability of $\mathcal{R}_{sf}(\boldTau)$ is proved by gluing
together unstructured and structured coding techniques. Each encoder splits it's message
$M_{j}$ into two parts $M_{j,1}$ and $M_{j}^{l_{j}}$. $M_{j,1}$ is communicated to the
decoder using an unstructured random code built over $\setU^{n}$. $M_{j}^{l_{j}}$ is
communicated to the decoder using a nested coset code identical to that proposed in proof
of theorem \ref{Thm:AchievableRateRegionUsingNestedCosetCodes}. With regard to nested
coset codes, we employ the notation proposed in the proof of theorem
\ref{Thm:AchievableRateRegionUsingNestedCosetCodes} and do not restate the same. 

Encoder $j$ is provided a codebook built over $\setU^{n}$ that contains
$2^{n\bar{R_{j}}}$ bins each with $2^{nB_{j}}$
codewords. For $1 \leq b_{j} \leq 2^{nB_{j}}$, let $u_{j}(r_{j},b_{j})$ denote a
generic codeword in bin $r_{j}$ ($1\leq r_{j}\leq 2^{n\bar{R_{j}}}$). Encoder $j$ is also
provided with the nested coset code $\lambda_{O_{j}/I}$. Without loss of generality, we
assume $M_{j}^{l_{j}} \in
\setV^{l_{j}}$. Encoder $j$ observes state sequence
$S_{j}^{n}$ and declares error if $S_{j}^{n}
\notin T_{\frac{\delta}{8}}(W_{S_{j}})$. Otherwise it looks for a pair
$(u_{j}^{n}(M_{j,1},b_{j}),v^{n}(a^{k},M_{j}^{l_{j}})) \in
T_{\frac{\delta}{4}}({U_{j}V_{j}}|S_{j}^{n})$. If it finds at least one such
pair, one of them say, $(u_{j}^{n}(M_{j,1},b_{j}),v^{n}(a^{k},M_{j}^{l_{j}}))$ is chosen
uniformly at random and $e_{j}^{n}(M_{j},S_{j}^{n})$ is transmitted, where $e_{j}^{n}(M_{j},S_{j}^{n})$ is a function of $u_{j}^{n}(M_{j,1},b_{j}),v^{n}(a^{k},M_{j}^{l_{j}}),S_{j}^{n}$ that is determined upfront.
Otherwise, an error is declared.

We now specify the decoding rule. The decoder receives $Y^{n}$ and declares error if
$Y^{n} \notin T_{\frac{\delta}{2}}({Y})$. Otherwise, decoding is performed in two
stages. In the first stage it lists all codewords $(u^{n}_{j}(m_{j,1},b_{j}):j=1,2) \in 
 T_{\delta}^{n}({U_{1},U_{2}}|y^{n})$. If it finds exactly one such pair, say
$(u^{n}_{j}(m_{j,1},b_{j}):j=1,2)$, then the decoding proceeds to the next stage.
Otherwise, an error is declared and decoding halts. In the second stage, the decoder looks
for all
codewords $v^{n}(a^{k},\bold{m^{l}}) \in \lambda_{O}$ such that
$(u^{n}_{j}(m_{j,1},b_{j}):j=1,2,v^{n}(a^{k},\bold{m^{l}}),Y^{n}) \in
T_{\delta}^{n}({U,V_{1}\oplus V_{2},Y})$. If it finds all such codewords in a unique bin, say
corresponding to $\bold{m^{l}}$, then it declares $m_{j,1},m_{j}^{l_{j}}:j=1,2$ as the
decoded pair of messages. Otherwise, an error is declared. We derive an upper bound on
probability of error by
averaging the error probability over the ensemble codes. A pmf is induced over
the ensemble of codes by letting $U_{j}^{n}(r_{j},b_{j}):1\leq r_{j \leq
2^{n\bar{R_{j}}}}, 1 \leq b_{j} \leq 2^{nB_{j}},j=1,2$ be mutually independent and
distributed according to $\prod_{t=1}^{n}p_{U_{j}}$. The pmf induced on the
ensemble of nested coset codes is identical to that in proof of theorem
\ref{Thm:AchievableRateRegionUsingNestedCosetCodes}. Moreover, $(G_{I},
G_{O_{j}/I},B_{j}^{n}:j=1,2)$ is independent of the unstructured random code on
$\setU^{n}$. Analyzing the error events, we obtain the following sufficient
conditions for the average probability of error to decay exponentially.
\begin{eqnarray}
 \label{Eqn:BoundsOnRatesForUnifiedRateRegion}
 B_{1} \geq I(U_{1};S_{1})&B_{2} \geq I(U_{2};S_{2})\nonumber\\
 \bar{R_{1}}+B_{1} \leq I(U_{1};U_{2}Y)&\bar{R_{2}}+B_{2} \leq
I(U_{2};U_{1}Y)\nonumber\\
\frac{k}{n} \geq 1-H(V_{1}|U_{1}S_{1})&\frac{k}{n} \geq 1-H(V_{2}|U_{2}S_{2})\nonumber\\
\sum_{j=1}^{2}\bar{R_{j}}+B_{j} \leq
I(\boldU;Y)+I(U_{1};U_{2})&\frac{l_{1}+l_{2}}{n} \leq 1 - H(V_{1}+V_{2}|\boldU
Y).\nonumber
\end{eqnarray}
For each $j=1,2$, substituting $R_{j}-\frac{l_{j}}{n}$ for $\bar{R_{j}}$ in the above
bounds and eliminating $B_{j},\frac{k}{n},\frac{l_{j}}{n}:j=1,2$ using the technique of
Fourier-Motzkin \cite[Appendix D]{201201NIT_ElgKim}, $\mathcal{R}_{sf}(\boldTau)$ is
proved achievable.
\end{proof}
\begin{remark}
 \label{Eqn:JointTypicalityEncodingAndDecodingMayPerformBetter}
The above rate region is obtained by analyzing sequential typicality encoding and
decoding, i.e., encoding and decoding of unstructured codes precedes that of structured
codes. The informed reader will recognize that performing joint typicality encoding and
decoding of unstructured and structured codes might enlarge the achievable rate region.
While this might be true, Fourier-Motzkin elimination of the resulting bounds does not
yield a compact description of the resulting achievable rate region. We therefore chose to
present the above rate region.
\end{remark}
We conclude with an illustrative example.
\begin{example}
\label{Ex:Blackwell}
For $ j=1,2$, let $\mathcal{S}_{j}=\mathcal{X}_{j}=\mathcal{Y}=\{ 0,1\}$. The channel
transition is described as $W_{Y|\boldX \boldS}(y|\boldx,\bolds) =
W^{*}_{Y|g(\boldX,\boldS)}(y|g(\boldx,\bolds))$, where $g(\boldx,\bolds) =
[(s_{2}\wedge \bar{x_{2}})\wedge(\bar{s_{1}}\vee x_{1})]\vee [(s_{1}\wedge
\bar{x_{1}})\wedge(\bar{s_{2}}\vee x_{2})]$, $\wedge$ denotes logical AND, and
$W^{*}_{Y|g(\boldX,\boldS)}(1|0)=0.02$,
$W^{*}_{Y|g(\boldX,\boldS)}(0|1)=0.04$. The function $g(\cdot,\cdot)$ can be
alternatively described as $g(\boldX,\boldS)=[S_{1}\wedge (S_{1} \oplus X_{1})]\oplus
[S_{2}\wedge (S_{2} \oplus X_{2})]$.

This channel is inspired by Blackwell's
broadcast channel and in particular the coding technique proposed by Gelfand
\cite{197707PIT_Gel}.\footnote{Analogous to the defect masking the written bits, here the
states mask the corresponding channel.} The bounds on the sum rate achievable with
unstructured and
nested coset codes are plotted in figure \ref{Fig:PlotsForBlackwellChannel}. The above
plots unequivocally indicate $\mathcal{R}_{sf}(\boldTau)$ to be strictly larger than
$\alpha(\boldTau) \cup \beta_{f}(\boldTau)$ and in particular either one of
$\alpha(\boldTau)$, $\beta_{f}(\boldTau)$. It is therefore desirable to compute
$\mathcal{R}_{sf}(\boldTau)$, however the presence of two additional auxiliary random
variables lends computation infeasible with current computational resources. We remark
that the structure of this example enables us to argue the strict
containment $\alpha(\boldTau) \cup \beta_{f}(\boldTau)
\subsetneq\mathcal{R}_{sf}(\boldTau)$ in spite of not being able to compute
$\mathcal{R}_{sf}(\boldTau)$.
\end{example}
\begin{figure}
\includegraphics[height=2.2in,width=4.4in]{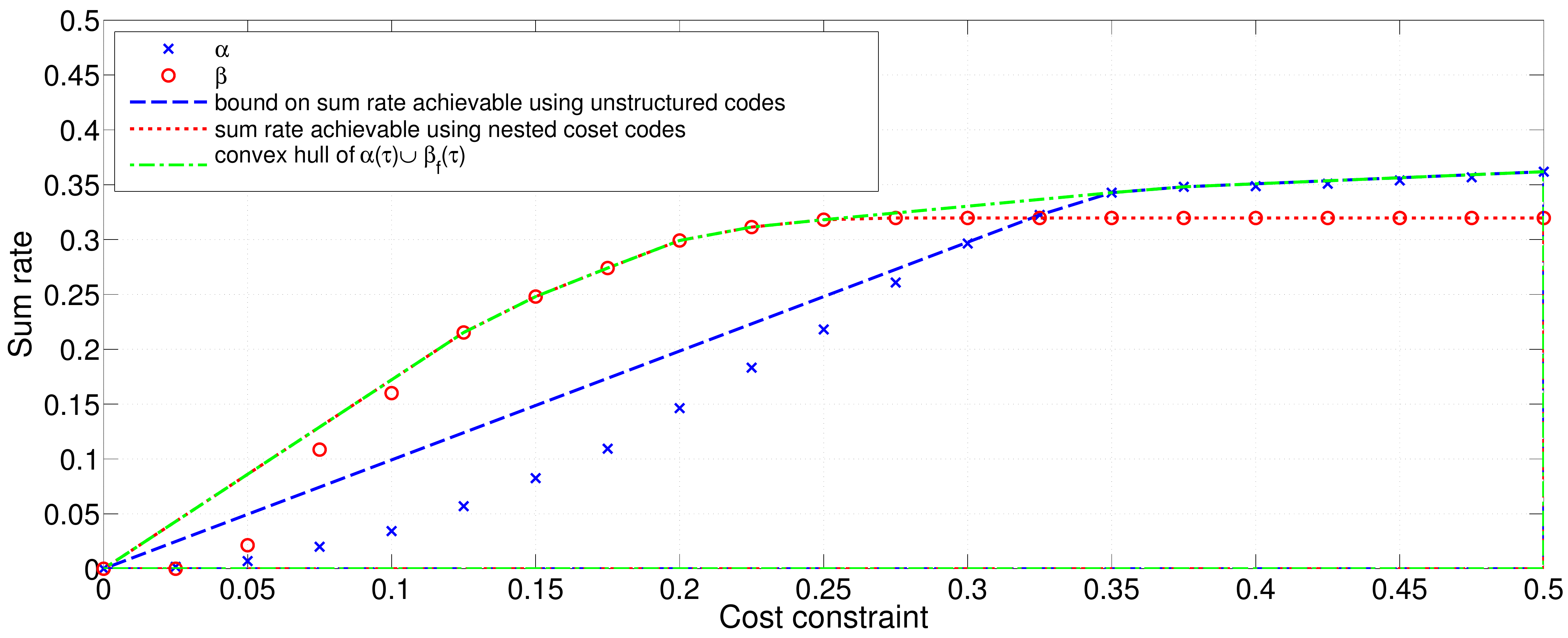}
\caption{Bounds on sum rate for example \ref{Ex:Blackwell}}
\label{Fig:PlotsForBlackwellChannel}
\end{figure}

\section{Stage III: Achievable rate region using codes over Abelian groups}
\label{Sec:EnlargingAchievableRateRegionUsingCodesOverGroups}

Consider a quaternary doubly dirty MAC-DSTx (QDD-MAC), with
$\StateAlphabet_{j}=\InputAlphabet_{j}=\OutputAlphabet=\{ 0,1,2,3\}$, $j=1,2$. The state
sequences are independent and uniformly distributed, i.e., $W_{\boldS}(\bolds) =
\frac{1}{16}$ for all $\bolds \in \StateAlphabet$. The channel transition is described by
the relation $Y=X_{1} \diamondplus S_{1} \diamondplus X_{2} \diamondplus S_{2}$, where
$\diamondplus$ denotes addition mod$-4$. All nonzero symbols have equal cost, i.e.,
$\kappa_{j}(x,s_{j})=1$ for all $x \in \left\{1,2,3  \right\}$ and $\kappa_{j}(0,s_{j})=0$
for all $s_{j} \in \StateAlphabet_{j}$, $j=1,2$ and the input is subject to a symmetric
cost constraint $\boldTau=(\tau,\tau)$.

What would be the achievable rate region for QDD-MAC using unstructured codes? It is natural to guess the optimizing test channel to be
\begin{equation}
\label{Eqn:OptimizingTestChannelForUnstructuredCodesOnQDD-MAC}
p_{X_{j}V_{j}|S_{j}}(x_{j},x_{j} \diamondplus s_{j}|s_{j})=
\begin{cases}
1-\tau &\mbox{ for }x_{j}=0\\
\frac{\tau}{3} &\mbox{ otherwise.}
\end{cases}
\end{equation}
In appendix D of \cite{201301arXivMACDSTx_PadPra}, with the aid of numerical computation,
we argue that this is indeed the case. The sum rate achievable using unstructured codes
can be evaluated to be the upper convex envelope of the function $\alpha:[0,\frac{3}{4}]
\rightarrow
[0,\infty)$ defined as $\alpha(\tau)=\max \left\{
-2\tau\log(\frac{\tau}{3})-2(1-\tau)\log(1-\tau))-2,0 \right\}$. Since $4$ is a prime
power, there exists a unique field $\mathcal{F}_{4}$ of cardinality $4$. Do nested coset
codes built over $\mathcal{F}_{4}$ achieve a larger sum rate?

We are unable to characterize the sum rate achievable using nested coset codes and the
dimensionality of the space of probability distributions lends computation infeasible. We
conjecture that the above test channel optimizes the sum rate achievable using nested
coset
codes. In any case, computing the sum rate achievable using nested coset codes for the
above test channel is instructive. It can be verified that the sum rate achievable using
the above test channel with nested coset codes is the upper convex envelope of the
function $\beta_{f}:[0,\frac{3}{4}] \rightarrow [0,\infty)$ defined as
$\beta_{f}(\tau)=\max
\left\{ -\tau\log(\frac{\tau}{3})-(1-\tau)\log(1-\tau))-\frac{1}{2},0 \right\}$.

\begin{figure}
\includegraphics[height=2.2in,width=4.4in]{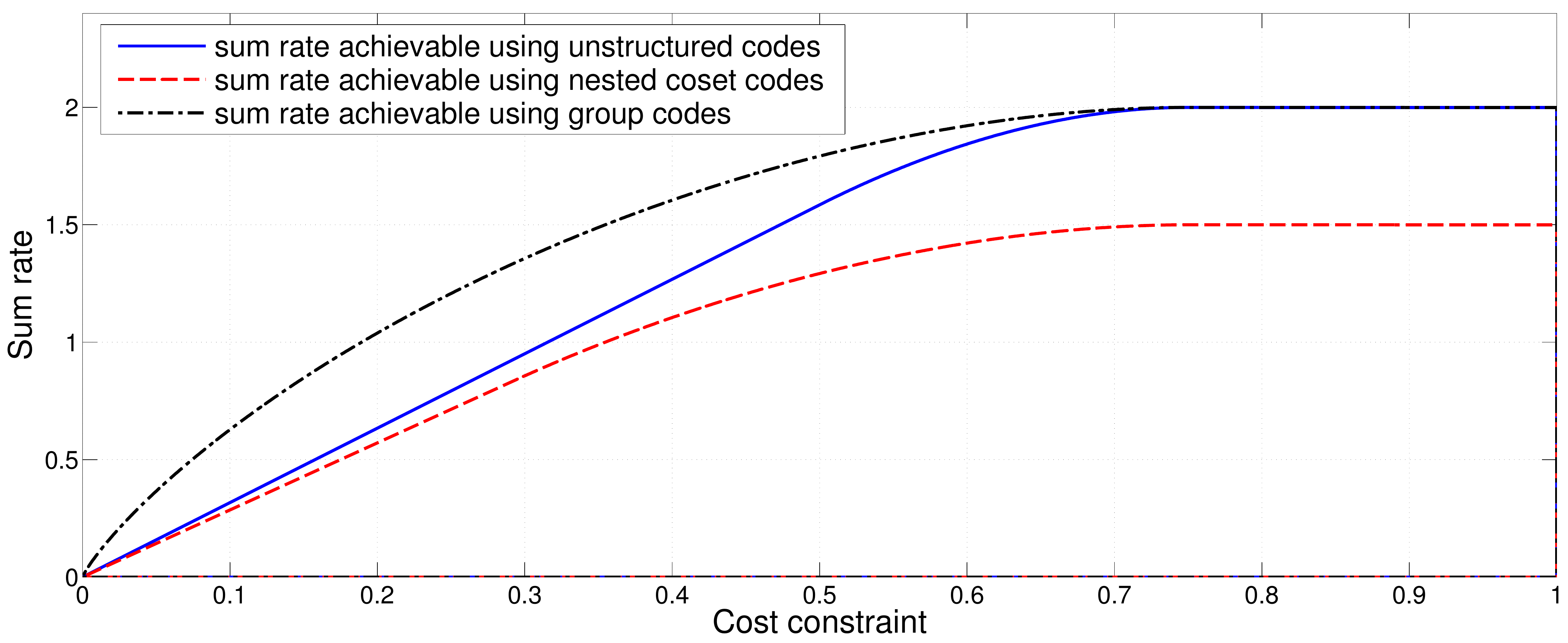}
\caption{Sum rate achievable using unstructured, nested coset and Abelian group codes for
test
channel (\ref{Eqn:OptimizingTestChannelForUnstructuredCodesOnQDD-MAC})}
\label{Fig:SumRateAchievableForQDD-MAC}
\end{figure}
The sum rate achievable for the above test channel using unstructured and nested coset
codes are plotted in figure \ref{Fig:SumRateAchievableForQDD-MAC}. It is no surprise that
nested coset codes perform poorly. The channel operation is \textit{not} the field
addition $\oplus_{4}$ in $\mathcal{F}_{4}$. Instead, $\diamondplus$ is the group
addition\footnote{We refer to group operation of an Abelian group
as group addition.}
in the Abelian group $\integers_{4}$. This suggests that we build codes over Abelian
groups
that are closed under group addition and decode the group sum $\diamondplus$ of codewords.

Linear codes are kernels of field homomorphisms. This lends
them the property of closure under field addition. We build \textit{Abelian group codes}
that are kernels of group homomorphisms. Abelian group codes are closed under group
addition. As was proposed with nested coset codes, we employ bins of each user's code to
be cosets of a common Abelian group code. The encoder chooses a codeword from the bin
indexed by the message and the decoder attempts to localize the group sum of chosen
codewords. The bins of each users' codebook is chosen such that the decoder can decode the
pair of messages by identifying the group sum of transmitted codewords.

In the interest of brevity, we only describe the results and omit proofs. Recall that any
Abelian group $\mathcal{V}$ can be decomposed as sum of $\mathbb{Z}_{p^{r}}-$cyclic
groups, i.e.,
\begin{equation}
 \label{Eqn:AbelianGroupAsACyclicSumOfZprGroups}
 \mathcal{V} = \underset{i=1}{\overset{I}\bigoplus} ~\mathbb{Z}_{p_{i}^{r_{i}}},
\end{equation}
where $p_{i}$ is a prime and $r_{i}$ is a positive integer for each $i=1,\cdots,I$. We
therefore state our findings in two stages. The first stage, described in section
\ref{SubSec:GroupCodesZprGroups} describes the coding technique and achievable rate
region for a $\mathbb{Z}_{p^{r}}-$ group. This is extended to an arbitrary Abelian
group in section \ref{SubSec:GroupCodesAbelianGroups}

\subsection{Achievable rate region for MAC-DSTx using group codes : The
$\mathbb{Z}_{p^{r}}$-case}
\label{SubSec:GroupCodesZprGroups}

In the discussion following proof of theorem
\ref{Thm:NestedCosetCodesAchieveCapacityOfPTP-STx}, we noted that if the auxiliary
alphabet $\mathcal{V}$ is a field and the bins are constrained to be closed under
field addition then with respect to a test channel $p_{V|S}$, the bins need to be of rate
at least $\log |\mathcal{V}| - H(V|S) $. This enlargement of the bins was compensated by
the ability
to pack more bins. In particular, the rate of the composite code could be as large as
$\log |\mathcal{V}| - H(V|Y)$ with respect to the induced distribution $p_{V|Y}$, and this
enabled us to achieve the capacity of PTP-STx.

If the
auxiliary alphabet $\mathcal{V}=\integers_{p^{r}}$ is an Abelian group of order 
$p^{r}$, and the bins are restricted to be closed under group addition, then with respect
to a test channel $p_{V|S}$, using the results of \cite{201209Allerton_SahPra}, the bins
have to be of rate at least
\begin{eqnarray}
\label{Eqn:SourceCodingZprGroupMutualInformation}
 \overline{I}_{s}^{\mathcal{V}}(V;S)= \max^{r}_{\theta=1} \left[ r\log p -
\frac{r}{\theta}H([V]_{\theta}|S) \right] 
= \max^{r}_{\theta=1} \frac{r}{\theta}I([V]_{\theta};S),
\end{eqnarray}
where $\mathcal{H}_{\theta}$ is the sub-group $p^{\theta}\mathbb{Z}_{p^{r}}$ and
$[V]_{\theta}
\define  V \diamondplus \mathcal{H}_{\theta}$ is the random variable taking values from
cosets of
subgroup $\mathcal{H}_{\theta}$ of $\mathcal{V}$, denoted $\mathcal{H}_{\theta} \preceq
\mathcal{V}$. We note that
$\overline{I}_{s}^{\mathcal{V}}(V;S) \geq \log q - H(V|S) \geq I(V;S)$. The natural
question to ask is whether we can pack sufficient number of bins to achieve capacity of
PTP-STx. It turns out that if we constrain the composite code, i.e., the union of bins, to
be a coset of a group code, then the rate of this union can be at most
\begin{eqnarray}
\overline{I}_{c}^{\mathcal{V}}(V;Y)= \min^{r-1}_{\theta=0} \left[ r\log p -
\frac{r}{r-\theta}H(V|Y[V]_{\theta}) \right] = \min^{r-1}_{\theta=0}
\frac{r}{r-\theta}I(V;Y|[V]_{\theta}).\nonumber
\end{eqnarray}
with respect to the induced distribution $p_{V|Y}$. Since $\log |\mathcal{V}|-H(V|Y)$
corresponds to $\theta=0$ in the above expression, $\overline{I}_{c}^{\mathcal{V}}(V;Y)$
is in general smaller than $\log |\mathcal{V}|-H(V|Y)$. Therefore,
$\overline{I}_{c}^{\mathcal{V}}(V;Y)-\overline{I}_{s}^{\mathcal{V}}(V;S)$ is in
general strictly smaller than the capacity of PTP-STx, implying the constraint of
closure under group addition results in a rate penalty. This indicates that the use of
group
codes will in general result in rate penalties for multi-terminal communication
problems.\footnote{The interested reader is referred to \cite{197109IC_Ahl-I},
\cite{197109IC_Ahl-II}, \cite{197102TAMS_Ahl} for early work on rates achievable using
group codes for point-to-point channels. \cite{201305TITarXiv_SahPra} provides bounds on
rates achievable using Abelian group codes for point-to-point source and channel coding
problems.}

With the objective of increasing $\overline{I}_{c}^{\mathcal{V}}(V;Y)$ and therefore
minimizing the rate penalty, we take a closer look at the coding technique proposed in
section \ref{SubSec:AnAchievableRateRegionForArbitraryMAC-DSTxUsingNestedCosetCodes}.
While we
exploited the property of bins being closed under field addition, we did not need the
union of bins to be a coset. We therefore relax this and only require the bins
to have an algebraic structure, i.e., a coset of a group code, but the composite code of
each user is not required to be a coset of a group code. While this relaxation does not
yield gains in achievable
rate for the field case, we do obtain larger achievable rates while coding over groups.
In particular, the rate of the composite code, or the union of bins can be as large as
$\log |\mathcal{V}|-H(V|Y)$ which is in general larger than
$\overline{I}_{c}^{\mathcal{V}}(V;Y)$. Therefore, if we were to communicate over a
PTP-STx $(\setS,W_{S},\setX,\kappa,\setY,W_{Y|XS})$ using codes over an Abelian
$\mathbb{Z}_{p^{r}}-$group $\mathcal{V}=\mathbb{Z}_{p^{r}}$ and we constrained the bins
to be closed under group addition, then the test channel $p_{VSXY} \in
\overline{\mathbb{D}}(\tau)$ yields an achievable rate $\log |\mathcal{V}|-H(V|Y)
-(\overline{I}_{s}^{\mathcal{V}}(V;S)) = \overline{H}_{s}^{\mathcal{V}}(V|S)-H(V|Y)$,
where 
\begin{equation}
 \label{Eqn:SourceCodingGroupEntropy}
 \overline{H}_{s}^{\mathcal{V}}(V|S)= \log
|\mathcal{V}|-\overline{I}_{s}^{\mathcal{V}}(V;S),
\end{equation}
is defined as \textit{source coding group
entropy} of group $\mathcal{V}=\mathbb{Z}_{p^{r}}$ and
$\overline{H}_{s}^{\mathcal{V}}(V) = \overline{H}_{s}^{\mathcal{V}}(V|0)$.

The diligent reader will now be able to characterize an achievable rate region for a
MAC-DSTx based on group codes. As mentioned earlier, the encoding and decoding techniques
are identical to that proposed in section
\ref{SubSec:AnAchievableRateRegionForArbitraryMAC-DSTxUsingNestedCosetCodes} except for
group addition replacing field addition. Consider a distribution $p_{\boldV \boldS
\boldX Y} \in \mathbb{D}(\tau)$ defined over $\mathcal{V}^{2}\times\mathcal{S}\times
\mathcal{X}\times\mathcal{Y}$ where $\mathcal{V}$ is an Abelian group of order $p^{r}$.
Cosets of a common group code is employed as bins of each user's code. Following an
analysis similar to that performed in proof of theorem
\ref{Thm:AchievableRateRegionUsingNestedCosetCodes}, one can prove the probability of the
encoders not finding a codeword jointly typical with the state sequence decays
exponentially with block length if the bins are of rate at least
$\max\left\{\log
|\mathcal{V}|-\overline{H}_{s}^{\mathcal{V}}(V_{j}|S_{j}):j=1,2\right\}$. The decoder
decodes the group sum of chosen codewords from the group sum of the two users' codebooks.
The codebooks of the two users are chosen to be union of arbitrary cosets of a common
group code
and
therefore the the group sum of the two users codebooks will also be a union of
arbitrary cosets of this group code. The probability of error at the decoders decays
exponentially
if the rate of the group sum of the two users' codebooks is at most $\log
|\mathcal{V}|-H(V_{1} \diamondplus V_{2}|Y)$. We conclude that a rate pair $(R_{1},R_{2})$
is
achievable if $R_{1}+R_{2} \leq
\min\left\{\overline{H}_{s}^{\mathcal{V}}(V_{j}|S_{j}):j=1,2\right\}-H(V_{1}
\diamondplus V_{2}|Y)$. The following is a formal characterization of achievable rate
region for MAC-DSTx using group codes over a $\mathbb{Z}_{p^{r}}-$group.

\begin{definition}
 \label{Defn:GroupTestChannels}
Let $\mathbb{D}_{sg}(\boldTau)
\subseteq \mathbb{D}(\boldTau)$ be the collection of
distributions $p_{\boldU \boldV\boldS\boldX  Y}$ on
$(\mathcal{U} \times \mathcal{V})^{2}\times\mathcal{S}\times
\mathcal{X}\times\mathcal{Y}$ where $\mathcal{U}$ is a finite set and $\mathcal{V}$ is an
Abelian group of order $p^{r}$,
where $p$ is a prime. For $p_{\boldU \boldV \boldS\boldX  Y} \in
\mathbb{D}_{sg}(\boldTau)$, let
$\mathcal{R}_{sg}(p_{\boldU \boldV \boldS\boldX  Y})$ be defined as the set of rate pairs
$(R_{1},R_{2}) \in [0,\infty)^{2}$ that satisfy
\begin{eqnarray}\label{Eqn:AchievableRateRegionForParticularTestChannelUsingGroupCodes}
\begin{array}{l}R_{1} \leq I(U_{1};U_{2}Y)-I(U_{1};S_{1}) +
\min \left\{ 
\overline{H}_{s}^{\mathcal{V}}(V_{1}|U_{1},S_{1}),\overline{H}_{s}^{\mathcal{V}}(V_{2}
|U_{2},S_{2 }
) \right\}-H(V_{1}\diamondplus
V_{2}|U_{1},U_{2},Y), \\
R_{2}\leq I(U_{2};U_{1}Y)-I(U_{2};S_{2}) +
\min \left\{ 
\overline{H}_{s}^{\mathcal{V}}(V_{1}|U_{1},S_{1}),\overline{H}_{s}^{\mathcal{V}}(V_{2}
|U_{2},S_{2 }
) \right\}-H(V_{1}\diamondplus
V_{2}|U_{1},U_{2},Y), \\
R_{1}+R_{2} \leq
I(U_{1}U_{2};Y)+I(U_{1};U_{2})-\sum_{j=1}^{2}I(U_{j};S_{j})+
\min \left\{ 
\overline{H}_{s}^{\mathcal{V}}(V_{1}|U_{1},S_{1}),\overline{H}_{s}^{\mathcal{V}}(V_{2}
|U_{2},S_{2 }
) \right\}
\\~~~~~~~~~~~~~~~~~~~~-H(V_{1}\diamondplus
V_{2}|U_{1},U_{2},Y), \end{array}\end{eqnarray}where $\diamondplus$ denotes group
addition in group $\mathcal{V}=\mathbb{Z}_{p^{r}}$, and
\begin{equation}
\label{Eqn:AchievableRateRegionUsingGroupCodes}
  \mathcal{R}_{sg}(\boldTau) \define \cocl \left( \underset{p_{\boldU\boldV\boldX\boldS Y}
\in
\mathbb{D}_{sg}(\boldTau)}{\bigcup} \mathcal{R}_{sg}(p_{\boldU \boldV\boldX \boldS Y}) 
\right)
\end{equation}
\end{definition}
\begin{thm}
 \label{Thm:AchievableRateRegionForMAC-DSTxUsingGroupCodes}
 $\mathcal{R}_{sg}(\boldTau) \subseteq \mathbb{C}(\boldTau)$.
\end{thm}
\begin{example}
 Let us now compute the achievable rate region using group codes for QDD-MAC. We only compute
$\mathcal{R}_{sg}(p_{\boldU \boldV\boldX \boldS Y})$ where $\mathcal{U} = \phi$, the empty
set and $\mathcal{V}=\left\{ 0,1,2,3 \right\}$ and $p_{VSX}$ is given in
(\ref{Eqn:OptimizingTestChannelForUnstructuredCodesOnQDD-MAC}).
$\mathcal{V}=\left\{ 0,1,2,3 \right\}$ has two sub-groups - the group itself, $\left\{
0,2 \right\}$. It can be verified that \[\overline{I}_{s}^{\mathcal{U}}(U;S)= \max
\left\{ \log_{2}4-2h_{b}(\frac{2\tau}{3}),\log_{2} 4
+\tau\log_{2}(\frac{\tau}{3})+(1-\tau)\log_{2}(1-\tau)  \right\}\] yielding
$\mathcal{R}_{sg}(p_{\boldU \boldV\boldX \boldS Y})=\left\{ (R_{1},R_{2}) \in
[0,\infty)^{2}:R_{1}+R_{2}  \leq
|\beta_{g}(\tau)|^{+}  \right\}$, where \[\beta_{g}(\tau)=\max
\left\{
\min\left\{-\tau\log_{2}(\frac{\tau}{3})-(1-\tau)\log_{2}(1-\tau),2h_{b}(\frac{2\tau}{3}
)\right\}
, 0 \right\}.\] In figure \ref{Fig:SumRateAchievableForQDD-MAC}, the sum rate achievable
using group codes for the above test channel is plotted. We highlight significant gains
achievable using group codes for QDD-MAC thus emphasizing the need to build codes with
appropriate algebraic structure that matches the channel.
\end{example}

\subsection{Achievable rate region for MAC-DSTx using group codes : The
general Abelian group}
\label{SubSec:GroupCodesAbelianGroups}

We now let the auxiliary alphabet $\mathcal{V}$ be a general Abelian group and build group
codes over $\mathcal{V}$ to enable the decoder to reconstruct the group sum of chosen
codewords. The discussion in section \ref{SubSec:GroupCodesZprGroups} indicates that we
only need to characterize the minimum rate of a bin in the code with respect to a generic
test channel $p_{V|S}$ under the constraint that the bin has to be a coset of a group
code. Essentially, this will involve characterizing fundamental group information
theoretic quantity $\overline{I}_{s}^{\mathcal{V}}(V;S)$ and the related source coding
group entropy $\overline{H}_{s}^{\mathcal{V}}(V|S)$ in the context of a general Abelian
group $\mathcal{V}$. 

Let $\mathcal{V}$ be the Abelian group in (\ref{Eqn:AbelianGroupAsACyclicSumOfZprGroups}).
Let $\theta = (\theta_{1},\cdots, \theta_{r})$ be such that $0\leq \theta_{i} \leq r_{i}$
for $i=1,2,\cdots,I$ and let $\mathcal{H}_{\theta}$ be a subgroup of $\mathcal{V}$ defined
as

\begin{equation}
 \label{Eqn:AbelianSubGroupAsACyclicSumOfZprGroups}
 \mathcal{H}_{\theta} = \underset{i=1}{\overset{I}\bigoplus}~
p^{\theta_{i}}\mathbb{Z}_{p_{i}^{r_{i}}},\nonumber
\end{equation}
and random variable $[V]_{\theta}$ taking values from cosets of $\mathcal{H}_{\theta}$ in
$\mathcal{V}$ as $[V]_{\theta}=V\diamondplus\mathcal{H}_{\theta}$. If the state has a pmf
$p_{S}$ and the bins
over $\mathcal{V}$ are constrained to be cosets of a group code, then for a test channel
$p_{V|S}$, the rate of a bin has to be at least
\begin{eqnarray}
\label{Eqn:SourceCodingAbelialGroupMutualInformation}
 \overline{I}_{s}^{\mathcal{V}}(V;S)&\define&
\underset{\substack{w_{1},\cdots,w_{I}\\w_{1}+\cdots + w_{I}=1}}{\min}~
\underset{\substack{\mathcal{H} \preceq \mathcal{V}\\\mathcal{H} \neq \mathcal{V}}}{\max}~
\frac{1}{1-w_{\theta}}I([V]_{\theta};S) 
\end{eqnarray}
where
\begin{equation}
 \label{Eqn:WTheta}
 w_{\theta}=\sum_{i=1}^{I} \frac{r_{i}-\theta_{i}}{r_{i}}w_{i}.\nonumber
\end{equation}
Alternatively, one might express the minimum rate of the bin as $\log
|\mathcal{V}|-\overline{H}_{s}^{\mathcal{V}}(V|S)$, where, as before 
\begin{equation}
 \label{Eqn:SourceCodingGroupEntropyForAbelianCodes}
  \overline{H}_{s}^{\mathcal{V}}(V|S)= \log
|\mathcal{V}|-\overline{I}_{s}^{\mathcal{V}}(V;S),
\end{equation}
is defined as the \textit{source coding group
entropy} of an Abelian group $\mathcal{V}$ and
$\overline{H}_{s}^{\mathcal{V}}(V) = \overline{H}_{s}^{\mathcal{V}}(V|0)$. We note
that definitions (\ref{Eqn:SourceCodingAbelialGroupMutualInformation}) and
(\ref{Eqn:SourceCodingGroupEntropyForAbelianCodes}) defined for an arbitrary Abelian group
reduces to that in (\ref{Eqn:SourceCodingZprGroupMutualInformation}) and
(\ref{Eqn:SourceCodingGroupEntropy}) for a $\mathbb{Z}_{p^{r}}-$group. This enables us
to characterize an achievable rate region for MAC-DSTx based on Abelian group codes using
$\mathcal{R}_{sg}(\boldTau)$.

\begin{definition}
\label{Defn:TestChannelsForAbelianGroups}
Let $\mathbb{D}_{sg}(\boldTau)
\subseteq \mathbb{D}(\boldTau)$ be the collection of
distributions $p_{\boldU \boldV\boldS\boldX  Y}$ on
$(\mathcal{U} \times \mathcal{V})^{2}\times\mathcal{S}\times
\mathcal{X}\times\mathcal{Y}$ where $\mathcal{U}$ is a finite set and $\mathcal{V}$ is an
Abelian group. For $p_{\boldU \boldV\boldS\boldX  Y} \in \mathbb{D}_{sg}(\boldTau)$, let
$\mathcal{R}_{sg}(p_{\boldU \boldV \boldS\boldX  Y})$ be defined as the set in
(\ref{Eqn:AchievableRateRegionForParticularTestChannelUsingGroupCodes}) and
$\mathcal{R}_{sg}(\boldTau)$ as in (\ref{Eqn:AchievableRateRegionUsingGroupCodes}).
\end{definition}
We conclude by stating that $\mathcal{R}_{sg}(\boldTau)$ is indeed achievable.
\begin{thm}
 \label{Thm:AchievableRateRegionForMAC-DSTxUsingAbelianGroupCodes}
 $\mathcal{R}_{sg}(\boldTau) \subseteq \mathbb{C}(\boldTau)$.
\end{thm}
\begin{remark}
The persistent reader will recognize that the achievable rate region based on group codes
hinges on the characterization of the minimum rate of a bin that is closed under group
addition with respect to a test channel $p_{V|S}$. For the general Abelian group we
stated this to be (\ref{Eqn:SourceCodingAbelialGroupMutualInformation}). Recent pursuit
has resulted in further reduction of this quantity and is available in
\cite{201305TITarXiv_SahPra}.
\end{remark}

\begin{remark}
\label{Rem:GroupSectionConcludingRemark}
The results in this section point to a rich theory of
strategies for
multi-terminal communication systems based on structured code ensembles. Gains crucially
rely on the compressive nature of the bivariate function and the ability to build
efficient codes with rich algebraic structure. It is therefore no surprise that all of
earlier findings were based on exploiting modulo$-2$ sum - the simplest compressive
function with binary arguments - using linear codes - an ensemble that has been studied at
length from different perspectives.
 
\end{remark}

\section{Concluding Remarks}
\label{Sec:ConcludingRemarks}
We provided a single letter characterization of a new achievable rate
region for the general MAC-DSTx. The reader will recognize that our findings are aimed at
developing a new framework for obtaining achievable rate region for multi-terminal
communication problems based on algebraic tools. We proposed achievable rate regions for
an arbitrary MAC-DSTx based on two algebraic structures - fields and Abelian groups. It should now
be clear to a persistent reader that a general rate region will involve a closure over all
algebraic structures of which fields and Abelian groups are just two of them. Furthermore, this
rate region will also incorporate the unstructured coding as indicated in section
\ref{Sec:AUnifiedAchievableRateRegion}. Indeed, a description of this will be involved,
and is justified by the presence of additional degrees of freedom in the
multi-terminal communication settings.
\appendices
\section{An upper bound on $P(\epsilon_{1}^{c} \cap \epsilon_{2})$}
\label{Sec:AnUpperBoundOnProbabilityofEpsilon1ComplementIntersectionEpsilon2}

Through out this appendix $\pi$ denotes $\pi
(\min \{  \left(|\setX|\cdot|\setS|\right)^{2},\left(|\setX|+|\setS|+|\setY|-2\right)\cdot|\setX|\cdot|\setS| \})$ and $\setV \define \fieldpi$. We begin with a simple lemma. The
following lemma holds for any $\fieldq$
and we state it in this generality.
\begin{lemma}
 \label{Lem:CodewordsOfRandomLinearCodeUniformlyDistributedAndPairwiseIndependent}
 Let $\fieldq$ be the finite field of cardinality $q$. If generator matrices $G_{I} \in
\fieldq^{k \times n}$, $G_{O/I} \in \fieldq^{l \times n}$ and bias vector $B^{n} \in
\fieldq^{n}$ of the random nested coset code $(n,k,l,G_{I},G_{O/I},B^{n})$ are mutually
independent and uniformly distributed on their respective range
spaces, then codewords $V^{n}(a^{k},m^{l}) \define a^{k} G_{I} \oplus m^{l}G_{O/I} \oplus
B^{n}$
are (i) uniformly distributed, and (ii) pairwise independent.
\end{lemma}
The proof follows form a simple counting argument and is omitted for the sake of brevity. The proof for the case $q=2$ is provided in \cite[Theorem 6.2.1]{Gal-ITRC68} and the same argument holds for any field $\mathcal{F}_{q}$.

We derive an upper bound on $P(\epsilon_{1}^{c} \cap \epsilon_{2})$ using a second moment method similar to that employed in \cite{198101TIT_GamMeu}.
\begin{eqnarray}
\label{Eqn:IndependenceOfThetaAndStateSequence}
P(\epsilon_{1}^{c} \cap \epsilon_{2})&=&\sum_{ s^{n} \in T_{
\frac{\delta}{4} }(p_{S}) }\sum_{ m^{l} \in \setV^{l} } P
\left(\substack{
 S^{n} = s^{n},M^{l}=m^{l}\\ \phi_{\frac{\delta}{2}}( s^{n},m^{l}) = 0}
\right)
=\sum_{ s^{n} \in T_{ \frac{\delta}{4} }(S) }\sum_{ m^{l}
\in \setV^{l} }P \left( \substack{S^{n} =
s^{n},\\M^{l}=m^{l}}\right)P( \phi_{\frac{\delta}{2}}( s^{n},m^{l}) = 0)\\
&\leq& \sum_{ s^{n} \in T_{ \frac{\delta}{4} }(S) }\sum_{ m^{l}
\in \setV^{l} } P ( S^{n} = s^{n},M^{l}=m^{l})P( |\phi_{\frac{\delta}{2}}( s^{n},m^{l})- \mathbb{E}\phi_{\frac{\delta}{2}}( s^{n},m^{l})| \geq \mathbb{E}\phi_{\frac{\delta}{2}}(
s^{n},m^{l}) )\nonumber\\
\label{Eqn:EmployingCheybyshevInequalityForEncoderError}
&\leq& \sum_{ s^{n} \in T_{ \frac{\delta}{4} }(S) }\sum_{ m^{l}
\in \setV^{l} } P \left( S^{n} =
s^{n},M^{l}=m^{l}\right)\frac{\mbox{Var}\left\{\phi_{\frac{\delta}{2}}(
s^{n},m^{l})\right\}}{\left\{\mathbb{E}\left\{\phi_{\frac{\delta}{2}}( s^{n},m^{l})\right\}\right\}^2},
\end{eqnarray}
where (\ref{Eqn:IndependenceOfThetaAndStateSequence}) is true since
$\phi_{\frac{\delta}{2}}( s^{n},m^{l})$ is a function of random objects $G_{I}$, $G_{O / I}$
and
$B^{n}$ that are mutually independent of $S^{n},M^{l}$, and
(\ref{Eqn:EmployingCheybyshevInequalityForEncoderError}) follows from Cheybyshev
inequality.

We now evaluate first and second moments of $\phi_{\frac{\delta}{2}}( s^{n},m^{l})$. The expectation of $\phi_{\frac{\delta}{2}}(
s^{n},m^{l})$ is 
\begin{eqnarray}
\label{eqn:MeanOfTheta}
\mathbb{E}\phi_{\frac{\delta}{2}}( s^{n},m^{l}) =
\sum_{\substack{v^{n} \in T_{\frac{\delta}{2}}^{n}\left( V|s^{n}
\right)}}\sum_{a^{k}
\in
\setV^{k}} P\left(V^{n}(a^{k},M^{l}) =
v^{n}\right)
=\frac{| T^{n}_{\frac{\delta}{2}}\left( V|s^n \right) |}{\pi^{n-k}},
\nonumber
\end{eqnarray}
where the last equality follows from Lemma
\ref{Lem:CodewordsOfRandomLinearCodeUniformlyDistributedAndPairwiseIndependent}(i). The
second moment is
\begin{eqnarray}
\label{eqn:ExpectedValueOfThetaSquared}
\lefteqn{\mathbb{E}{\phi_{\frac{\delta}{2}}^{2}( s^{n},m^{l})} = \sum_{\substack{v^{n},\tilde{v}^{n} \in 
T_{\frac{\delta}{2}}^{n}\left( V|s^{n} \right)}}\sum_{\substack{a^{k},\tilde{a}^{k}  
\in \setV^{k}}}   P\left(V^{n}(a^{k},M^{l}) = v^{n},
V^{n}(\tilde{a}^{k},M^{l}) =
\tilde{v}^{n}\right) } \nonumber\\
&=&\sum_{\substack{v^{n}\in \\ 
T_{\frac{\delta}{2}}^{n}\left( V|s^{n} \right)}}\sum_{a^{k} \in
\setV^{k}}   P\left(V^{n}(a^{k},M^{l}) = v^{n}\right)+ \sum_{\substack{v^{n},\tilde{v}^{n} \in \\ 
T_{\frac{\delta}{2}}^{n}\left( V|s^{n}
\right)}}\sum_{\substack{a^{k},\tilde{a}^{k} \in \\\setV^{k},a^{k}\neq
\tilde{a}^{k}}}\!\!\!\!\!\!\!P\left(V^{n}(a^{k},M^{l}) = v^{n},
V^{n}(\tilde{a}^{k},M^{l})
= \tilde{v}^{n}\right) \nonumber\\
\label{Eqn:SimplifiedExpressionForExpectedValueOfThetaSquared}
&=&\frac{\pi^{k}\left| T^{n}_{\frac{\delta}{4}}\left( V|s^{n}
\right)\right|}{\pi^n}+\frac{\left| T^{n}_{\frac{\delta}{2}}\left( V|s^n \right) \right|^2
\pi^{k}\left(\pi^{k}-1 \right)}{\pi^{2n}}, 
\end{eqnarray}
where second term in (\ref{Eqn:SimplifiedExpressionForExpectedValueOfThetaSquared})
follows from Lemma
\ref{Lem:CodewordsOfRandomLinearCodeUniformlyDistributedAndPairwiseIndependent}(ii).
Substituting for first and second moments of $\phi_{\frac{\delta}{2}}( s^{n},m^{l})$, we have
\begin{eqnarray}
\label{eqn:VarianceOfAlpha}
\mbox{Var}\left\{\phi_{\frac{\delta}{2}}( s^{n},m^{l})\right\} = \frac{\pi^{k}\left|
T^{n}_{\frac{\delta}{2}}\left( V|s^{n} \right)\right|}{\pi^n} \left( 1- \frac{\left|
T^{n}_{\frac{\delta}{2}}\left( V|s^{n} \right)\right|}{\pi^n} \right), \mbox{ thus } 
\frac{\mbox{Var}\left\{\phi_{\frac{\delta}{2}}(
s^{n},m^{l})\right\}}{\mathbb{E}\left\{\phi_{\frac{\delta}{2}}( s^{n},m^{l})\right\}^2}  \leq
\frac{\pi^{n-k}}{|
T^{n}_{\frac{\delta}{2}}\left( V|s^{n} \right)|}.
\end{eqnarray}
For $s^{n} \in T_{\frac{\delta}{4}}(S)$ lemma
\ref{Lem:BoundsOnSizeOfConditionalTypicalSet}, guarantees existence of $N_{3}(\eta) \in
\naturals$, such that for all $n \geq N_{3}(\eta)$, $|T_{\frac{\delta}{2}}(V|s^{n})| \geq
\exp\left\{ n\left( H(V|S)-\frac{3\delta}{4}\right) \right\}$. Substituting this lower
bound in (\ref{eqn:VarianceOfAlpha}), we note,
\begin{eqnarray}
 \label{eqn:UpperboundOnvarianceByMeanSquare}
\frac{\mbox{Var}\left\{\phi_{\frac{\delta}{2}}(
s^{n},m^{l})\right\}}{\mathbb{E}\left\{\phi_{\frac{\delta}{2}}( s^{n},m^{l})\right\}^2}  \leq
\frac{\pi^{n-k}}{|
T^{n}_{\frac{\delta}{2}}\left( V|s^{n} \right)|} \leq \exp\left\{-n \log \pi
\left( \frac{k}{n} -
\left(1-\frac{H\left(V|S \right)}{\log \pi} +\frac{3\delta}{4\log \pi}\right)
\right)\right\}.\nonumber\\
\end{eqnarray}
Substituting (\ref{eqn:UpperboundOnvarianceByMeanSquare}) in
(\ref{Eqn:EmployingCheybyshevInequalityForEncoderError}), we obtain
\begin{eqnarray}
 \label{Eqn:UpperBoundOnEncoderErrorProbability}
P(\epsilon_{1}^{c} \cap \epsilon_{2}) &\leq& \exp\left\{-n \log \pi
\left( \frac{k}{n} -
\left(1-\frac{H\left(V|S \right)}{\log \pi} +\frac{3\delta}{8\log \pi}\right)
\right)\right\}.\nonumber
\end{eqnarray}
From (\ref{Eqn:LowerBoundOnBinningRate}), we have
\begin{eqnarray}
 \label{Eqn:ChoiceOfkBynAndDeltaEnsures}
 \frac{k}{n} -
\left(1-\frac{H\left(V|S \right)}{\log \pi} +\frac{3\delta}{8\log \pi} \right)\geq
\frac{\frac{\eta}{8}-\frac{3\delta}{8}}{\log \pi} \geq \frac{\eta}{16\log \pi}
\end{eqnarray}
where the last inequality follows from choice of $\delta$.
Combining (\ref{Eqn:UpperBoundOnEncoderErrorProbability}) and
(\ref{Eqn:ChoiceOfkBynAndDeltaEnsures}), we have $P(\epsilon_{1}^{c} \cap \epsilon_{2})
\leq \exp\left\{-\frac{3n\delta}{8 \log \pi} \right\} \leq \frac{\eta}{16}$ for all $n \geq
N_{4}(\eta)$.

By choosing $\delta > 0$ sufficiently
small, $\frac{k}{n}$ can be made arbitrarily close
to $1-\frac{H\left(V|S \right)}{\log \pi}$ and probability of encoding error can be made arbitrarily
small by choosing a sufficiently large block length.
The above findings are summarized in the following lemma.
\begin{lemma}
 \label{Lem:WhenIsACosetCodeAGoodSourceCode?}
Let $\setS$ be a finite set, $\setV=\fieldq$ a finite field and $p_{SV}$, a
pmf on $\setS \times \setV$. Consider a random nested coset code
$(n,k,l,G_{I},G_{O/I},B^{n})$ denoted $\Lambda_{O}/\Lambda_{I}$, with bias vector $B^{n} \in \setV^{n}$, generator
matrices $G_{I} \in \setV^{k \times n}$ and $G_{O/I} \in \setV^{l \times n}$
mutually independent and uniformly distributed on their respective range spaces. Let $V^{n}(a^{k},m^{l}) \define a^{k}G_{I}\oplus m^{l}G_{O/I} \oplus B^{n}$ denote generic codeword in $\Lambda_{O}/\Lambda_{I}$. For $s^{n} \in \setS^{n}$, $m^{l} \in \setV^{l}$ and $\delta > 0$, let $\phi_{\delta}(s^{n},m^{l}) \define \sum_{a^{k} \in
\setV^{k}}1_{\{ (s^{n,}V^{n}(a^{k},m^{l})) \in T_{\delta}(S,V) \}}$. The following are true.
\begin{enumerate}
\item The codewords $V^{n}(a^{k},m^{l}):a^{k} \in \setV^{k}$ are uniformly distributed and pairwise independent.
\item For any $\delta > 0$, $s^{n} \in T_{\frac{\delta}{2}}(S)$, $m^{l} \in \setV^{l}$, there exists $N(\delta) \in \naturals$ such that for all $n \geq N(\delta)$,
\begin{equation}
 \label{Eqn:UpperboundOnEncodingErrorForFixedStateSequence}
P(\phi_{\delta}(s^{n},m^{l})=0) \leq \exp \left\{ -n \log q \left(\frac{k}{n} - \left( 1- \frac{H(V|S)}{\log q} -\frac{3\delta}{2 \log q}\right)  \right) \right\}.\nonumber
\end{equation}
\item If $(S^{n},M^{l}) \in \setS^{n} \times \setV^{l}$ are independent of $(G_{I},G_{O/I},B^{n})$, then for all $n\geq N(\delta)$,
\begin{equation}
\label{Eqn:UpperboundOnEncodingErrorForRandomStateAndMessages}
P(S^{n} \in T_{\frac{\delta}{2}}(S),\phi_{\delta}(S^{n},M^{l})=0) \leq \exp \left\{ -n \log q \left(\frac{k}{n} - \left( 1- \frac{H(V|S)}{\log q} -\frac{3\delta}{2 \log q}\right)  \right) \right\}.\nonumber
\end{equation}
\end{enumerate}
\end{lemma}
\section{An upper bound on $P((\epsilon_{1}\cup\epsilon_{2}\cup\epsilon_{3})^{c} \cap
\epsilon_{4})$}
\label{Sec:AnUpperBoundOnProbabilityofEpsilon4}
As is typical, our achievability proof hinges on independence of transmitted codeword
(and hence received vector) and the contending codewords that are not transmitted.
Towards this end, we begin with the following.
\begin{lemma}
 \label{Lem:IndependenceOfACosetAndACodewordInADifferentCoset}
Let $\setV$ be the finite field of cardinality $q$. If generator matrices $G_{I} \in
\fieldq^{k \times n}$, $G_{O/I} \in \setV^{l \times n}$ and bias vector $B^{n} \in
\fieldq^{n}$ of the random $(n,k,l,G_{I},G_{O/I},B^{n})$ nested coset code are mutually
independent and uniformly distributed on their respective range
spaces, then any coset is independent of any codeword in a different coset., i.e., the collection of codewords 
$(V^{n}(a^{k},m^{l}):a^{k}\in\fieldq^{k})$ and $V^{n}(\hat{a}^{k},\hatm^{l})$ are
independent if $m^{l}\neq \hatm^{l}$.
\end{lemma}
\begin{proof}
 Let $v^{n}_{a^{k}} \in \fieldq^{n}$ for each $a^{k} \in \fieldq^{k}$, and
$\hat{v}^{n} \in \fieldq^{n}$. We need to prove 
\begin{eqnarray}
P ( V^{n}(a^{k},m^{l})=v^{n}_{a^{k}}\!:\! a^{k} \in
\fieldq^{k}, V^{n}(\hata^{k},m^{l})\!=\!\hatv^{n} )= P( 
V^{n}(a^{k},m^{l})=v^{n}_{a^{k}}\! :\! a^{k} \in
\fieldq^{k})P(V^{n}(\hata^{k},m^{l})=\hatv^{n}).\nonumber
\end{eqnarray}
If $( v^{n}_{a^{k}+\hata^{k}} - v^{n}_{0^{k}} ) \neq
(v^{n}_{a^{k}} - v^{n}_{0^{k}}) + (
v^{n}_{\hata^{k}} - v^{n}_{0^{k}})$ for some pair $a^{k}$, $\hata^{k} \in
\fieldq^{k}$, the LHS and first term of RHS are zero and equality holds. Else,
\begin{eqnarray}
\label{eqn:PairwiseIndependenceCosetAndCodewordInAnotherCoset}
\lefteqn{P (
V^{n}(a^{k},m^{l})=v^{n}_{a^{k}} : a^{k} \in
\fieldq^{k},
V^{n}(\hata^{k},m^{l})=\hatv^{n} )} \nonumber\\
&=&P ( a^{k}G_{I} = v_{a^{{k}}}^{n} - v_{0^{k}}^{n} :
a^{k} \in \fieldq^{k}, V^{n}(0^{k},m^{l}) = v_{0^{{k}}}^{n}
,V^{n}(0^{k},\hatm^{l}) = \hatv^{n} -
v_{\hata^{k}}^{n}
)\nonumber\\
\label{Eqn:IndependenceOfG1DeltaGAndBn}
&=&P( a^{k}G_{I} = v_{a^{{k}}}^{n} - v_{0^{{k}}}^{n} :
a^{k} \in \fieldq^{k})P( V^{n}(0^{k},m^{l}) = v_{0^{{k}}}^{n}
,V^{n}(0^{k},\hat{m}^{l}) = \hatv^{n} - v_{\hata^{k}}^{n}) \\
\label{Eqn:GallagersIndependence}
&=&P( a^{k}G_{I} = v_{a^{{k}}}^{n} - v_{0^{{k}}}^{n} :
a^{k} \in \fieldq^{k})P( V^{n}(0^{k},m^{l}) =
v_{0^{{k}}}^{n})P(
V^{n}(0^{k},\hatm^{l}) = \hatv^{n} - v_{\hata^{k}}^{n}) \\
\label{Eqn:IndependenceOfG1DeltaGAndBnOnceAgain}&=&P( a^{k}G_{I} = v_{a^{{k}}}^{n} - v_{0^{{k}}}^{n} :
a^{k} \in \fieldq^{k}, V^{n}(0^{k},m^{l}) =
v_{0^{{k}}}^{n})P(
\hat{m}^{l} G_{O / I} + B^{n} = \hatv^{n} - v_{\hata^{k}}^{n}) \\
&=&P( 
V^{n}(a^{k},m^{l})=v^{n}_{a^{k}} : a^{k} \in
\fieldq^{k})P(V^{n}(\hata^{k},m^{l})=\hatv^{n},
)\nonumber
\end{eqnarray}
where (\ref{Eqn:IndependenceOfG1DeltaGAndBn}) and
(\ref{Eqn:IndependenceOfG1DeltaGAndBnOnceAgain}) follow from independence of $G_{O
/ I}$,
$B^{n}$ and $G_{I}$ (\ref{Eqn:GallagersIndependence}) follows from Lemma
\ref{Lem:CodewordsOfRandomLinearCodeUniformlyDistributedAndPairwiseIndependent}(ii), and
the last equality follows from invariance of the pmf of $V^{n}(a^{k},m^{l})$
with respect to $a^{k}$ and $m^{l}$.
\end{proof}
 We emphasize the consequence of Lemma
\ref{Lem:IndependenceOfACosetAndACodewordInADifferentCoset} in the following remark.
\begin{remark}
\label{Rem:IndependenceOfABinAndAVectorFromAnotherBin}If transmitted message $M^{l} \neq
\hatm^{l}$, then $Y^{n}$ is independent of $V^{n}(\hata^{k},\hatm^{l})$. Indeed
\begin{eqnarray}
&&\!\!\!\!\!\!\!\!\!\!P(V^{n}(\hata^{k},\hatm^{l})=\hatv^{n},Y^{n}=y^{n})
\!=\!\!\!\!\!\!\sum_{(v^{n}_{a^{k}}\in\setV^{n}:a^{k} \in \setV^{k})}
\sum_{x^{n} \in \setX^{n}}
P\left(\substack{C(M^{l})=(v^{n}_{a^{k}}\in\setV^{n}:a^{k} \in
\setV^{k}),\\V^{n}(\hata^{k},\hatm^{l})=\hatv^{n},E(S^{n},M^{l})=x^{n},Y^{
n } =y^ { n } }\right) \nonumber\\
\label{Eqn:IndependenceOfReceivedVectorAndACompetingCodewordInADifferentBin_1}
&=&\underset{(v^{n}_{a^{k}}\in\setV^{n}:a^{k} \in \setV^{k})}{\sum}
\underset{x^{n} \in \setX^{n}}{\sum} P\left(\begin{array}{l}C(M^{l})=(v^{n}_{a^{k}}\in\setV^{n}:a^{k} \in
\setV^{k}),\\E(S^{n},M^{l})=x^{n},Y^{
n } =y^ { n } \end{array}\right)P(V^{n}(\hata^{k},\hatm^{l} )=\hatv^{n}) \\
\label{Eqn:IndependenceOfReceivedVectorAndACompetingCodewordInADifferentBin_2}
&=&P(V^{n}(\hata^{k},\hatm^{l})=\hatv^{n})P(Y^{n}=y^{n})=\frac{P(Y^{n}=y^{n})}{q^{n}}.
\end{eqnarray}
We have used (1) independence of $V^{n}(\hata^{k},\hatm^{l})$ and $C(M^{l})$ (lemma
\ref{Lem:IndependenceOfACosetAndACodewordInADifferentCoset}), (2) $E(S^{n},M^{l})$ being
a function of $C(M^{l})$ and $S^{n}$ is conditionally independent of
$V^{n}(\hata^{k},\hatm^{l})$ given $C(M^{l})$, and (3) $Y^{n}$ is conditionally
independent of $V^{n}(\hata^{k},\hatm^{l})$ given $E(S^{n},M^{l})$ in arriving at
(\ref{Eqn:IndependenceOfReceivedVectorAndACompetingCodewordInADifferentBin_1}), and lemma
\ref{Lem:CodewordsOfRandomLinearCodeUniformlyDistributedAndPairwiseIndependent}(i) in arriving at the last equality in (\ref{Eqn:IndependenceOfReceivedVectorAndACompetingCodewordInADifferentBin_2}).
\end{remark}
We now provide an upper bound on $P((\epsilon_{1}\cup\epsilon_{2}\cup\epsilon_{3})^{c}\cap\epsilon_{4})$. Observe that
\begin{eqnarray}
 \label{Eqn:BoundingDecodingErrorProbability} 
&&P((\epsilon_{1}\cup\epsilon_{2}\cup\epsilon_{3})^{c}\cap\epsilon_{4})
\leq
P\left(\underset{\hata^{k} \in \setV^{k}}{\cup}\underset{\hatm^{l} \neq M^{l}}{\cup} \{
(V^{n}(\hata^{k}, \hatm^{l}),  Y^{n}) \in T_{\delta}(p_{VY})
\}\right)\nonumber\\
&\leq& \sum_{\substack{\hatm^{l}\in \setV^{l}\\\hatm^{l}\neq M^{l}}}
\sum_{\hata^{k} \in \setV^{k}} \sum_{\substack{y^{n}\\\in
T_{\frac{\delta}{2}}}}  \sum_{\substack{v^{n} \in\\ \small
T_{\delta}(V|y^{n} ) }} P( V^{n}(\hata^{k},\hatm^{l})=v^{n},
Y^{n}=y^{n})\nonumber\\
&=&\!\!\!\!\!\! \sum_{\substack{\hatm^{l}\in \setV^{l}\\\hatm^{l}\neq M^{l}}}
\sum_{\hata^{k} \in \setV^{k}} \sum_{\substack{y^{n}\\\in
T_{\frac{\delta}{2}}}}  \sum_{\substack{v^{n} \in\\ \small
T_{\delta}(V|y^{n} ) }}\!\! P( V^{n}(\hata^{k},\hatm^{l})=v^{n})
\label{Eqn:DecodingErrorProbabilityExpanded1}
P(Y^{n}=y^{n})=\!\!\! \sum_{\substack{\hatm^{l}\in \setV^{l}\\\hatm^{l}\neq M^{l}}}
\sum_{\hata^{k} \in \setV^{k}} \sum_{\substack{y^{n}\\\in
T_{\frac{\delta}{2}}}}  \sum_{\substack{v^{n} \in\\ \small
T_{\delta}(V|y^{n} ) }}\!\!\!\!\!\!\frac{P(Y^{n}=y^{n})}{\pi^{n}}
\\\label{Eqn:DecoderErrorEventInTermsOfCardinalityOfConditionalTypicalSet}
&\leq& \sum_{y^{n}\in
T_{\frac{\delta}{2}}} \!\!\! \frac{\pi^{k+l}|T_{\delta}(p_{V|Y}|y^{n})|P(Y^{n}=y^{n})}{\pi^{n}},
\end{eqnarray}
where, the two equalities in (\ref{Eqn:DecodingErrorProbabilityExpanded1}) follow from
(\ref{Eqn:IndependenceOfReceivedVectorAndACompetingCodewordInADifferentBin_2}). Lemma \ref{Lem:BoundsOnSizeOfConditionalTypicalSet} guarantees existence of $N_{5}(\eta) \in \naturals$ such that for all $n \geq N_{5}(\eta)$ and $y^{n}
\in T_{\frac{\delta}{2}}(p_{Y})$, $|T_{\delta}(V|y^{n})| \leq \exp \{ n(H(V|Y)+\frac{3\delta}{2})\}$. Substituting this upper bound in (\ref{Eqn:DecoderErrorEventInTermsOfCardinalityOfConditionalTypicalSet}), we conclude
\begin{equation}
\label{Eqn:UpperboundOnDecodingErrorProbability}
P((\epsilon_{1}\cup\epsilon_{2}\cup\epsilon_{3})^{c}\cap\epsilon_{4}) \leq \exp \left\{ 
 -n \log \pi \left( 1- \frac{H(V|Y)}{\log \pi} - \frac{3\delta}{2 \log \pi} -
\frac{k+l}{n}\right) \right\}\end{equation}
for all $n \geq N_{5}(\eta)$.
\section{An upper bound on
$P(\epsilon_{5})$}
\label{Sec:AnUpperBoundOnProbabilityofEpsilon4ForMAC-DSTx}
In this appendix, we derive an upper bound on $P(\epsilon_{5})$. As is typical in proofs of channel coding theorems, this step involves establishing statistical independence of cosets $C_{j}(M_{j}^{l_{j}}):j=1,2$ corresponding to the message pair and any codeword $V^{n}(\hata^{k},\hatm^{l})$ in a competing coset. We establish this in lemma \ref{Lem:CosetsCorrespondingToTheMessagesAreIndependentOfCodewordInADifferentSumCoset}. We begin with the necessary spadework. Throughout this appendix, we employ the notation introduced in proof of theorem \ref{Thm:AchievableRateRegionUsingNestedCosetCodes}.
\begin{lemma}
 \label{Lem:StatisticalIndependenceOfCosetLeaders}
If $m^{l} \neq \hatm^{l}$, then for any triple $\nu_{1},\nu_{2},\hatnu \in \setV^{n}$,
\begin{eqnarray}
 \label{Eqn:StatisticalIndependenceOfCosetLeaders}
P\left(\substack{V_{j}^{n}(0^{k_{j}},m_{j}^{l_{j}})=\nu_{j}^{n}:j=1,2,\\V^{n}(0^{k},\hatm^{l}
)=\hatnu^{n}}\right)=P\left( V_{j}^{n}(0^{k_{j}},m_{j}^{l_{j}})=\nu_{j}^{n}:j=1,2\right)P\left( V^{n}(0^{k},\hatm^{l}
)=\hatnu^{n} \right)\nonumber\end{eqnarray}
\end{lemma}
\begin{proof}
 By definition of $V_{j}(0_{k_{j}},m_{j}^{l_{j}}):j=1,2$ and $V(0^{k},m^{l})$,
\begin{eqnarray}
P\left(\substack{V_{j}^{n}(0^{k_{j}},m_{j}^{l_{j}})=\nu_{j}^{n}:j=1,2,\\V^{n}(0^{k},\hatm^{l}
)=\hatnu^{n}}\right)
&=&P\left( \substack{ \left[m_{1}^{l_{1}}~0^{l_{2}} \right]G_{O/I}\oplus B_{1}^{n} = \nu_{1}^{n},\left[0^{l_{1}}~m_{2}^{l_{2}} \right]G_{O/I}\oplus B_{2}^{n} =\nu_{2}^{n} \\ \left[ \hatm_{1}^{l_{1}} ~ \hatm_{2}^{l_{2}}\right]G_{O/I} \oplus B_{1}^{n}\oplus B_{2}^{n} = \hatnu^{n}} \right)\nonumber\\
\label{Eqn:CosetLeadersOfEncoderAndDecoderNCC}
&=&P\left( \substack{ \left[m_{1}^{l_{1}}~0^{l_{2}} \right]G_{O/I}\oplus B_{1}^{n} =
\nu_{1}^{n},\left[0^{l_{1}}~m_{2}^{l_{2}} \right]G_{O/I}\oplus B_{2}^{n}
=\nu_{2}^{n} \\ \left[ \tilde{m}_{1}^{l_{1}} ~ \tilde{m}_{2}^{l_{2}}\right]G_{O/I}
=\hatnu^{n}}
\right)
\end{eqnarray}
where $\tilde{m}_{j}^{l_{j}}=\hatm_{j}^{l_{j}}-m_{j}^{l_{j}}$. We now prove, using a counting argument similar to that employed in proof of lemma
\ref{Lem:CodewordsOfRandomLinearCodeUniformlyDistributedAndPairwiseIndependent}, the term
on right hand side of (\ref{Eqn:CosetLeadersOfEncoderAndDecoderNCC})
is $\frac{1}{\pi^{3n}}$. Since $\hatm^{l} \neq m^{l}$, there exists $t \in [l]$ such that
$\hatm_{t} \neq m_{t}$. Given any $(l-1)$ vectors $g_{O/I,j} \in \setV^{n}:j \in
[l]\setminus\{ t \}$, there exists a unique triple of vectors
$(g_{O/I,t},b_{1}^{n},b_{2}^{n}) \in \setV^{n} \times \setV^{n} \times \setV^{n}$
such that $\left[m_{1}^{l_{1}}~0^{l_{2}} \right]g_{O/I}\oplus b_{1}^{n} =
\nu_{1}^{n},\left[0^{l_{1}}~m_{2}^{l_{2}} \right]g_{O/I}\oplus b_{2}^{n}
=\nu_{2}^{n}$ and $\left[ \tilde{m}_{1}^{l_{1}} ~
 \tilde{m}_{2}^{l_{2}}\right]g_{O/I}
=\hatnu^{n}$, where row $j$ of $g_{O/I}$ is $g_{O/I,j}$. Hence
\begin{equation}
 \label{Eqn:CountingArgumentForMAC-DSTxIndependence}
\left|\left\{ (g_{O/I},b_{1}^{n},b_{2}^{n}) \in \setV^{k\times n} \times \setV^{n} \times \setV^{n}: \substack{\left[m_{1}^{l_{1}}~0^{l_{2}} \right]g_{O/I}\oplus B_{1}^{n} =
\nu_{1},\left[0^{l_{1}}~m_{2}^{l_{2}} \right]g_{O/I}\oplus B_{2}^{n}
=\nu_{2} \\ \left[ \tilde{m}_{1}^{l_{1}} ~ \tilde{m}_{2}^{l_{2}}\right]g_{O/I}
=\hatnu^{n}} \right\} \right| =\pi^{(l-1)n}.\nonumber
\end{equation}
The mutual independence and uniform distribution of $G_{O/I},B_{1},B_{2}^{n}$ implies the term on RHS of (\ref{Eqn:CosetLeadersOfEncoderAndDecoderNCC}) is indeed $\frac{1}{\pi^{3n}}$. It remains to prove
\begin{equation}
P\left( V_{j}^{n}(0^{k_{j}},m_{j}^{l_{j}})=\nu_{j}^{n}:j=1,2\right)P\left( V^{n}(0^{k},\hatm^{l} )=\hatnu^{n} \right)=\frac{1}{\pi^{3n}}.\nonumber
\end{equation}
It follows from lemma \ref{Lem:CodewordsOfRandomLinearCodeUniformlyDistributedAndPairwiseIndependent} that $P\left( V^{n}(0^{k},\hatm^{l} )=\hatnu^{n} \right)=\frac{1}{\pi^{n}}$. Using the definition of $V^{n}(0^{k},\hatm^{l} )$, we only need to prove
\begin{equation}
P\left( \substack{\left[m_{1}^{l_{1}}~0^{l_{2}} \right]G_{O/I}\oplus B_{1}^{n} =
\nu_{1},\\\left[0^{l_{1}}~m_{2}^{l_{2}} \right]G_{O/I}\oplus B_{2}^{n}
=\nu_{2}} \right)=\frac{1}{\pi^{2n}}.\nonumber
\end{equation}
This follows again from a counting argument. For every matrix $g_{O/I} \in \setV^{l \times n}$, there exists a unique pair of vectors $b^{n}_{1}, b_{2}^{n} \in \setV^{n}$ such that $\left[m_{1}^{l_{1}}~0^{l_{2}} \right]G_{O/I}\oplus B_{1}^{n} =
\nu_{1}$, and $\left[0^{l_{1}}~m_{2}^{l_{2}} \right]G_{O/I}\oplus B_{2}^{n}
=\nu_{2}$ thus yielding
\begin{equation}
 \left|\left\{ (g_{O/I},b_{1}^{n},b_{2}^{n}) \in \setV^{k\times n} \times \setV^{n} \times \setV^{n}: \substack{\left[m_{1}^{l_{1}}~0^{l_{2}} \right]G_{O/I}\oplus B_{1}^{n} =
\nu_{1},\\\left[0^{l_{1}}~m_{2}^{l_{2}} \right]G_{O/I}\oplus B_{2}^{n}
=\nu_{2}} \right\} \right|=\pi^{ln},
\end{equation}
and the proof is completed using the mutual independence and uniform distribution of $G_{O/I},B_{1}^{n},B_{2}^{n}$.
\end{proof}

\begin{lemma}
 \label{Lem:CosetsCorrespondingToTheMessagesAreIndependentOfCodewordInADifferentSumCoset}
For any ${\hatm}^{l} \neq m^{l}$, and any $\hata^{k} \in \setV^{k}$, the pair of cosets $C_{j}(m^{l_{j}}_{j}):j=1,2$ is statistically independent of
$V^{n}(\hata^{k},{\hatm}^{l})$.
\end{lemma}
\begin{proof}
For $j=1,2$, let $\nu_{j}^{n}(a_{j}^{k_{j}}) \in \setV^{n}$ for each $a_{j}^{k_{j}} \in \setV^{k_{j}}$, and $\hatnu^{n} \in \setV^{n}$. We need to prove
 \begin{eqnarray}
 P\left( \begin{array}{l}C_{1}^{n}(m_{1}^{l_{1}})=(\nu_{1}(a_{1}^{k_{1}}):a_{1}^{k_{1}} \in
\setV^{k_{1}})\\C_{2}^{n}(m_{2}^{l_{2}})=(\nu_{2}(a_{2}^{k_{2}}):a_{2}^{k_{2}} \in
\setV^{k_{2}})\\V^{n}(\hata^{k},\hatm^{l}
)=\hatnu^{n}\end{array} \right)
= P\left( \begin{array}{l} C_{1}^{n}(m_{1}^{l_{1}})=(\nu_{1}(a_{1}^{k_{1}}):a_{1}^{k_{1}} \in
\setV^{k_{1}})\\C_{2}^{n}(m_{2}^{l_{2}})=(\nu_{2}(a_{2}^{k_{2}}):a_{2}^{k_{2}} \in
\setV^{k_{2}}) \end{array}\right)P(\substack{V^{n}(\hata^{k},\hatm^{l}
)\\=\hatnu^{n}})\nonumber
 \end{eqnarray}
for every choice of $\nu_{j}(a_{j}^{k_{j}}) \in \setV^{n}:a_{j}^{k_{j}} \in
\setV^{k_{j}},j=1,2$ and $\hatnu^{n}\in \setV^{n}$.

If (i) for some $j=1$ or $j=2$, $( \nu_{j}(a_{j}^{k_{j}}\oplus\tilde{a}_{j}^{k_{j}}) - \nu_{j}({0^{k_{j}}}) )
\neq (\nu_{j}(a_{j}^{k_{j}}) - \nu_{j}({0^{k_{j}}})) \oplus (\nu_{j}(\tilde{a}_{j}^{k_{j}}) - \nu_{j}({0^{k_{j}}}))$ for any pair $a_{j}^{k_{j}}$,
$\tilde{a_{j}}^{k_{j}} \in \setV^{k_{j}}$, or (ii) $\nu_{1}(a^{k_{1}})-v_{1}(0^{k_{1}}) \neq \nu_{2}(a_{1}^{k_{1}}0^{k_{+}})-v_{2}(0^{k_{2}})$ for some $a_{1}^{k_{1}} \in \setV^{k_{1}}$,
then LHS and first term of RHS are zero and equality holds. Otherwise,
\begin{eqnarray}
 \lefteqn{P\left(\substack{ C_{j}^{n}(m_{j}^{l_{j}})=(\nu_{j}(a_{j}^{k_{j}}):a_{j}^{k_{j}} \in \setV^{k_{j}}):j=1,2,V^{n}(\hata^{k},\hatm^{l} )=\hatnu^{n}
}\right)}\nonumber\\
\label{Eqn:ConsequenceOfStatisticalIndependenceOfCosetLeaders}
&=& P\left(\substack{a_{2}^{k_{2}}G_{I_{2}}=\nu_{2}(a^{k_{2}})-\nu_{2}(0^{k_{2}}):a_{2}^{k_{2}} \in \setV^{k_{2}},
V_{j}^{n}(0^{k_{j}},m_{j}^{l_{j}})=\nu_{j}(0^{k_{j}}):j=1,2,\\V^{n}(0^{k},\hatm^{l}
)=\hatnu^{n}-(\nu_{2}(\hata^{k})-\nu_{2}(0^{k_{2}}))}\right)\\
\label{Eqn:GI2IsIndependentOfTheRest}
&=&P\left(\substack{a_{2}^{k_{2}}G_{I_{2}}=\nu_{2}(a^{k_{2}})-\\\nu_{2}(0^{k_{2}}):a_{2}^{k_{2}} \in \setV^{k_{2}}}\right)P\left(\substack{V_{j}^{n}(0^{k_{j}},m_{j}^{l_{j}})=\nu_{j}(0^{k_{j}}):j=1,2,\\V^{n}(0^{k},\hatm^{l}
)=\hatnu^{n}-(\nu_{2}(\hata^{k})-\nu_{2}(0^{k_{2}}))}\right)\\
\label{Eqn:DefinitionOfCosets}
&=&P\left(\substack{a_{2}^{k_{2}}G_{I_{2}}=\nu_{2}(a^{k_{2}})-\\\nu_{2}(0^{k_{2}}):a_{2}^{k_{2}} \in \setV^{k_{2}}}\right)P\left(\substack{ \left[m_{1}^{l_{1}}~0^{l_{2}} \right]G_{O/I}\oplus B_{1}^{n} =
\nu_{1}(0^{k_{1}}),\\\left[0^{l_{1}}~m_{2}^{l_{2}} \right]G_{O/I}\oplus B_{2}^{n}
=\nu_{2}(0^{k_{2}}) }\right)P\left( V^{n}(\hata^{k},\hatm^{l})=\hatnu^{n} \right)\\
\label{Eqn:GI2IsIndependentOfTheRestAgain}
&=& P\left( \substack{a_{2}^{k_{2}}G_{I_{2}}=\nu_{2}(a^{k_{2}})-\nu_{2}(0^{k_{2}}):a_{2}^{k_{2}} \in \setV^{k_{2}}\\\left[m_{1}^{l_{1}}~0^{l_{2}} \right]G_{O/I}\oplus B_{1}^{n} =
\nu_{1}(0^{k_{1}}),\left[0^{l_{1}}~m_{2}^{l_{2}} \right]G_{O/I}\oplus B_{2}^{n}
=\nu_{2}(0^{k_{2}})} \right)P\left( V^{n}(\hata^{k},\hatm^{l})=\hatnu^{n} \right)\\
\label{Eqn:DefinitionOfCosetsAgain}&=&P\left( \substack{C_{j}^{n}(m_{j}^{l_{j}})=(\nu_{j}(a_{j}^{k_{j}}):a_{j}^{k_{j}} \in
\setV^{k_{j}}):j=1,2} \right)P\left( V^{n}(\hata^{k},\hatm^{l})=\hatnu^{n} \right)
\end{eqnarray}
where i) (\ref{Eqn:DefinitionOfCosets}) and (\ref{Eqn:DefinitionOfCosetsAgain}) follow from definition of cosets $C_{j}(m_{j}^{l_{j}})$, (ii) (\ref{Eqn:GI2IsIndependentOfTheRest}) and (\ref{Eqn:GI2IsIndependentOfTheRestAgain}) follow from independence of $G_{I_{2}}$ and the collection $(G_{O/I},B_{1}^{n},B_{2}^{n})$ and (iii) (\ref{Eqn:ConsequenceOfStatisticalIndependenceOfCosetLeaders}) follows from lemma \ref{Lem:StatisticalIndependenceOfCosetLeaders}.
\end{proof}
We emphasize consequence of lemma
\ref{Lem:CosetsCorrespondingToTheMessagesAreIndependentOfCodewordInADifferentSumCoset} in
the following remark.
\begin{remark}
\label{Rem:CosetsCorrespondingToTheMessagesAreIndependentOfCodewordInADifferentSumCoset}
If $m^{l} \neq \hatm^{l}$, then conditioned on event $\left\{ M^{l}=m^{l} \right\}$,
received vector $Y^{n}$ is statistically independent of $V^{n}(\hata^{k},\hatm^{l})$ for
any $\hata^{k} \in \setV^{k}$. We establish truth of this statement in the sequel. Let
$\mathcal{C}_{j}$ denote the set of all ordered $\pi^{k_{j}}$-tuples of vectors in
$\setV^{n}$. Observe that
\begin{eqnarray}
 \lefteqn{P\left(
\substack{M^{l}=m^{l},Y^{n}=y^{n},\\V^{n}(\hata^{k},\hatm^{l})=\hatnu^{n} }
\right)=\sum_{C_{1} \in \mathcal{C}_{1} }
\sum_{C_{2} \in \mathcal{C}_{2}}
\sum_{\bolds^{n} \in \setS^{n}}
P\left(\substack{
M^{l}=m^{l},C_{j}(m_{j}^{l_{j}})=C_{j}:j=1,2,S^{n}=s^{n}\\V^{n}(\hata^{k},\hatm^{l}
)=\hatv^{n},Y^{n}=y^{n}}\right) }\nonumber\\
\label{Eqn:UsingStatisticalIndependenceOfCosetsCorrToMessageAndCompetingCodeword}
&=& \sum_{C_{1} \in \mathcal{C}_{1} }
\sum_{C_{2} \in \mathcal{C}_{2}}
\sum_{\bolds^{n} \in \setS^{n}} \!\!\!P\left(  \substack{M^{l}=m^{l}\\S^{n}=s^{n}}
\right)P\left( \substack{ C_{1}(m_{1}^{l_{1}})=C_{1}\\C_{2}(m_{2}^{l_{2}})=C_{2}}
\right)P\left(\substack{V^{n}(\hata^{k},\hatm^{l}
)=\hatv^{n}} \right)P\left(  Y^{n}=y^{n} |
\substack{C_{j}(m_{j}^{l_{j}})=C_{j}:j=1,2\\S^{n}=s^{n},M^{l}=m^{l} } \right)\\
&=&\sum_{C_{1} \in \mathcal{C}_{1} }
\sum_{C_{2} \in \mathcal{C}_{2}}
\sum_{\bolds^{n} \in \setS^{n}}
P\left(\substack{
M^{l}=m^{l},Y^{n}=y^{n},S^{n}=s^{n}\\C_{j}(m_{j}^{l_{j}})=C_{j}:j=1,2}\right)
P\left(\substack{V^{n}(\hata^{k},\hatm^{l}
)=\hatv^{n}} \right)
\nonumber\\
&=&P\left(M^{l}=m^{l},Y^{n}=y^{n}\right) P\left(
V^{n}(\hata^{k},\hatm^{l})=\hatnu^{n}\right)\nonumber
\end{eqnarray}
where (\ref{Eqn:UsingStatisticalIndependenceOfCosetsCorrToMessageAndCompetingCodeword})
follows from (i) independence of random objects that characterize codebook and
$(S^{n},M^{l})$, (ii) lemma
\ref{Lem:CosetsCorrespondingToTheMessagesAreIndependentOfCodewordInADifferentSumCoset} and
(iii) statistical independence of the inputs
$X_{j}(M_{j}^{l_{j}},S_{j}^{n}):j=1,2$ to the channel and the codeword
$V^{n}(\hata^{k},\hatm^{l})$ conditioned on the specific realization of cosets
$(C_{j}(M_{j}^{l_{j}}):j=1,2)$ and the event $\left\{ M^{l}=m^{l} \right\}$. Moreover,
since $P(V^{n}(\hata^{k},\hatm^{l})=\hatnu^{n})=\frac{1}{\pi^{n}}$, we have $P(
M^{l}=m^{l},Y^{n}=y^{n},V^{n}(\hata^{k},\hatm^{l})=\hatnu^{n} 
)=\frac{1}{\pi^{n}}P(M^{l}=m^{l},Y^{n}=y^{n})$.
\end{remark}
We are now equipped to derive an upper bound on $P(\epsilon_{5})$. Observe that
\begin{eqnarray}
 \label{Eqn:UpperboundingDecodingErrorProbability}
\lefteqn{ P(\epsilon_{5}) \leq P\left(\underset{\hata^{k} \in
\setV^{k}}{\bigcup} \underset{\substack{m^{l}, {\hatm^{l}}\\m^{l}
\neq{\hatm^{l}}}}{\bigcup}\left\{ \substack{ 
(V^{n}(\hata^{k},\hatm^{l}),Y^{n}) \in 
T_{\eta_{5}(\eta)}(p_{V_{1} \oplus
V_{2}, Y})\\M^{l}=m^{l} }  \right\}\right)}\nonumber\\
&\leq& \underset{\hata^{k} \in
\setV^{k}}{\sum} \underset{\substack{m^{l}, {\hatm^{l}}\\m^{l}
\neq{\hatm^{l}}}}{\sum} \sum_{\substack{   y^{n}\\\in
T_{\eta_{5}(\eta)} (Y)  }}  \sum_{\substack{v^{n} \in\\ \small
T_{\eta_{5}(\eta)}(V_{1}\oplus V_{2}|y^{n} ) }}
P\left(\substack{V^{n}(a^{k},{\hatm^{l}})=v^{n}\\M^{l}=m^{l},Y^{n}=y^{n}}
\right)\nonumber\\
\label{Eqn:IndependenceFollowingFromAboveRemark}
&\leq& \underset{\hata^{k} \in
\setV^{k}}{\sum} \underset{\substack{m^{l}, {\hatm^{l}}\\m^{l}
\neq{\hatm^{l}}}}{\sum} \sum_{\substack{y^{n}\\\in
T_{\eta_{5}(\eta)} (Y)}   }  \sum_{\substack{v^{n} \in\\ \small
T_{\eta_{5}(\eta)}(V_{1}\oplus V_{2}|y^{n} ) }}
P\left(V^{n}(a^{k},\hatm^{l})=v^{n}\right)P(M^{l}=m^{l},Y^{n}=y^{n})\nonumber\\
\label{Eqn:IndependenceFollowingFromAboveRemark2}
&\leq& \underset{\hata^{k} \in
\setV^{k}}{\sum} \underset{\hatm^{l} \in \setV^{l}}{\sum} \sum_{\substack{y^{n}\\\in
T_{\eta_{5}(\eta)} (Y)}   }  \sum_{\substack{v^{n} \in\\ \small
T_{\eta_{5}(\eta)}(V_{1}\oplus V_{2}|y^{n} ) }}
\frac{P(Y^{n}=y^{n})}{\pi^{n}}\nonumber\\
\label{Eqn:UpperboundOnDecodingErrorProbability}
&\leq&\!\!\!\!\!\!\! \sum_{\substack{y^{n}\\\in
T_{\eta_{5}(\eta)} (Y)}   }\!\!\!\!\!  \frac{\pi^{k+l}|T_{2\eta_{5}(\eta)}(V_{1} \oplus
V_{2}|y^{n} )|}{\pi^{n}} \leq
\exp \left\{-n \log \pi \left(1-\frac{H(V_{1}\oplus
V_{2}|Y)+3\eta_{5}(\eta)}{\log \pi}-\frac{k+l}{n}\right)\right\}.
\end{eqnarray}
where (\ref{Eqn:UpperboundOnDecodingErrorProbability}) follows
from the uniform bound of $\exp\left\{n\left(H(V_{1} \oplus V_{2}|Y)+3\eta_{5}(\eta)\right)\right\}$ on $|T_{2\eta_{5}(\eta)}(V_{1} \oplus V_{2}|y^{n} )|$
for any $y^{n}\in
T_{\eta_{5}(\eta)}(Y)$, $n \geq N_{6}(\eta)$ provided by lemma \ref{Lem:BoundsOnSizeOfConditionalTypicalSet} for $n\geq N_{6}(\eta)$. Substituting the upper bound for $\frac{k+l}{n}$ in (\ref{Eqn:MACDSTxUpperBoundOnTheUnionCode}), we have
\begin{equation}
 \label{Eqn:FinalUpperBoundOnProbabilityOfEpsilon5}
P(\epsilon_{5}) \leq \exp \left\{ -n \left( \eta_{2}(\eta)+\eta_{3}(\eta)-3\eta_{5}(\eta)\right) \right\} \text{ for all }n \geq \max\left\{ N_{1}(\eta),N_{6}(\eta) \right\}.
\end{equation}
\section*{Acknowledgment}
The first author thanks (i) Raj Tejas Suryaprakash for lending his expertise with regard
to computing and plotting results for examples
\ref{Ex:DoublyDirtyMACReplacedWithLogicalOr} - \ref{Ex:Blackwell}, (ii)
Aria Sahebi for useful discussions with respect to material presented in section
\ref{Sec:EnlargingAchievableRateRegionUsingCodesOverGroups}, and (iii) Deepanshu Vasal for
general discussions.
\bibliographystyle{../sty/IEEEtran}
{
\bibliography{MultipleAccessWithDistributedStates.bbl}
}
\end{document}